\documentclass[prx,aps,superscriptaddress,twocolumn,nofootinbib,longbibliography]{revtex4-2}

\usepackage{graphicx,color,amsmath,amsfonts,enumerate,amsthm,amssymb,mathtools,enumitem,thmtools,hyperref,subfigure,mathdots,enumitem,centernot,bm,soul,bbm,tikz,pgfplots,float, physics}
\usepackage[capitalise, noabbrev]{cleveref}
\usepackage[normalem]{ulem}

\usepackage{tcolorbox}
\hypersetup{colorlinks=true,linkcolor=teal,citecolor=teal,urlcolor=teal}
\pgfplotsset{compat=newest}
\tcbset{before skip=10pt,toptitle=2mm,bottomtitle=1mm,fonttitle=\bfseries}
\tcbuselibrary{theorems}
\tcbuselibrary{breakable}

\definecolor{NavyBlue}{rgb}{0.0, 0.0, 0.5}
\definecolor{OliveGreen}{rgb}{0.33, 0.42, 0.18}
\definecolor{def_color_frame}{RGB}{220,230,242}
\colorlet{def_color_back}{def_color_frame!30}
\definecolor{def_color_text}{RGB}{37,64,97}
\definecolor{def_color_frame2}{RGB}{242,200,200}
\colorlet{def_color_back2}{def_color_frame2!30}
\definecolor{def_color_text2}{RGB}{97,55,33}

\theoremstyle{definition} 

\theoremstyle{remark}

\usepackage{newfloat}
\usepackage[super]{nth}
\DeclareFloatingEnvironment[fileext=frm,placement={!ht},name=Box]{myfloat} 

\definecolor{beige}{rgb}{1.00,0.95,0.90}
\usepackage[framemethod=TikZ]{mdframed} 

\def\N{ {\cal N} }

\def\>{\rangle}
\def\<{\langle}

\newcommand{\dket}{\rangle\!\rangle}
\newcommand{\dbra}{\langle\!\langle}

\definecolor{bluecyan}{rgb}{0.27, 0.66, 0.88}
\definecolor{ppblue}{RGB}{46,117,182}
\definecolor{ppred}{RGB}{197, 90, 17}

\theoremstyle{plain}
\newtheorem{thm}{Theorem}

\newtheorem{lem}[thm]{Lemma}
\newtheorem{prop}[thm]{Proposition}

\theoremstyle{definition}
\newtheorem{defn}{Definition}

\def\cred {\color{red}}
\def\cblue {\color{blue}}

\begin{document}

\title{Catalytic enhancement in the performance of the microscopic two-stroke heat engine}

\author{Tanmoy Biswas}


\affiliation{Theoretical Division (T4), Los Alamos National Laboratory, Los Alamos, New Mexico 87545, USA.}
\affiliation{International Centre for Theory of Quantum Technologies, University of Gdansk, Wita Stwosza 63, 80-308 Gdansk, Poland.}
\author{Marcin {\L}obejko}
\affiliation{Institute of Theoretical Physics and Astrophysics, Faculty of Mathematics, Physics and Informatics, University of Gda\'nsk, 80-308 Gda\'nsk, Poland}
\affiliation{International Centre for Theory of Quantum Technologies, University of Gdansk, Wita Stwosza 63, 80-308 Gdansk, Poland.}
\author{Pawe{\l} Mazurek}
\affiliation{International Centre for Theory of Quantum Technologies, University of Gdansk, Wita Stwosza 63, 80-308 Gdansk, Poland.}
\affiliation{Institute of Informatics, Faculty of Mathematics, Physics and Informatics, University of Gdansk, Wita Stwosza 63, 80-308 Gdansk, Poland.}
\author{Micha{\l} Horodecki}
\affiliation{International Centre for Theory of Quantum Technologies, University of Gdansk, Wita Stwosza 63, 80-308 Gdansk, Poland.}

\begin{abstract}
We consider a model of  heat engine operating in the microscopic regime: the two-stroke engine. It produces work and exchanges heat in two discrete strokes that are separated in time. The working body of the  engine consists of two $d$-level systems initialized in thermal states at two distinct temperatures. Additionally, an auxiliary non-equilibrium system called catalyst may be incorporated  with the working body of the  engine, provided the state of the catalyst remains unchanged after the completion of a thermodynamic cycle. This ensures that the work produced by the engine arises solely from the temperature difference.  Upon establishing the rigorous thermodynamic framework, we characterize two-fold improvement stemming from the inclusion of a catalyst. Firstly, we prove that in the non-catalytic scenario, the optimal efficiency of the two-stroke heat engine with a working body composed of two-level systems is given by the Otto efficiency, which can be surpassed by incorporating a catalyst with the working body.  Secondly, we show that incorporating a catalyst allows the engine to operate in frequency and temperature regimes that are not accessible for non-catalytic two-stroke engines. We conclude with general conjecture about advantage brought by catalyst: including the catalyst with the working body always allows to improve efficiency over the non-catalytic scenario for any  microscopic two-stroke heat engines. We prove the conjecture for two-stroke engines when the working body is composed of two $d$-level systems initialized in thermal states at two distinct temperatures, as long as the final joint state leading to optimal efficiency in the non-catalytic scenario is not product, or at least one of the $d$-level system is not thermal.

\end{abstract}

\maketitle
\section{Introduction}

\begin{figure*}[t]
\includegraphics[width=18cm]{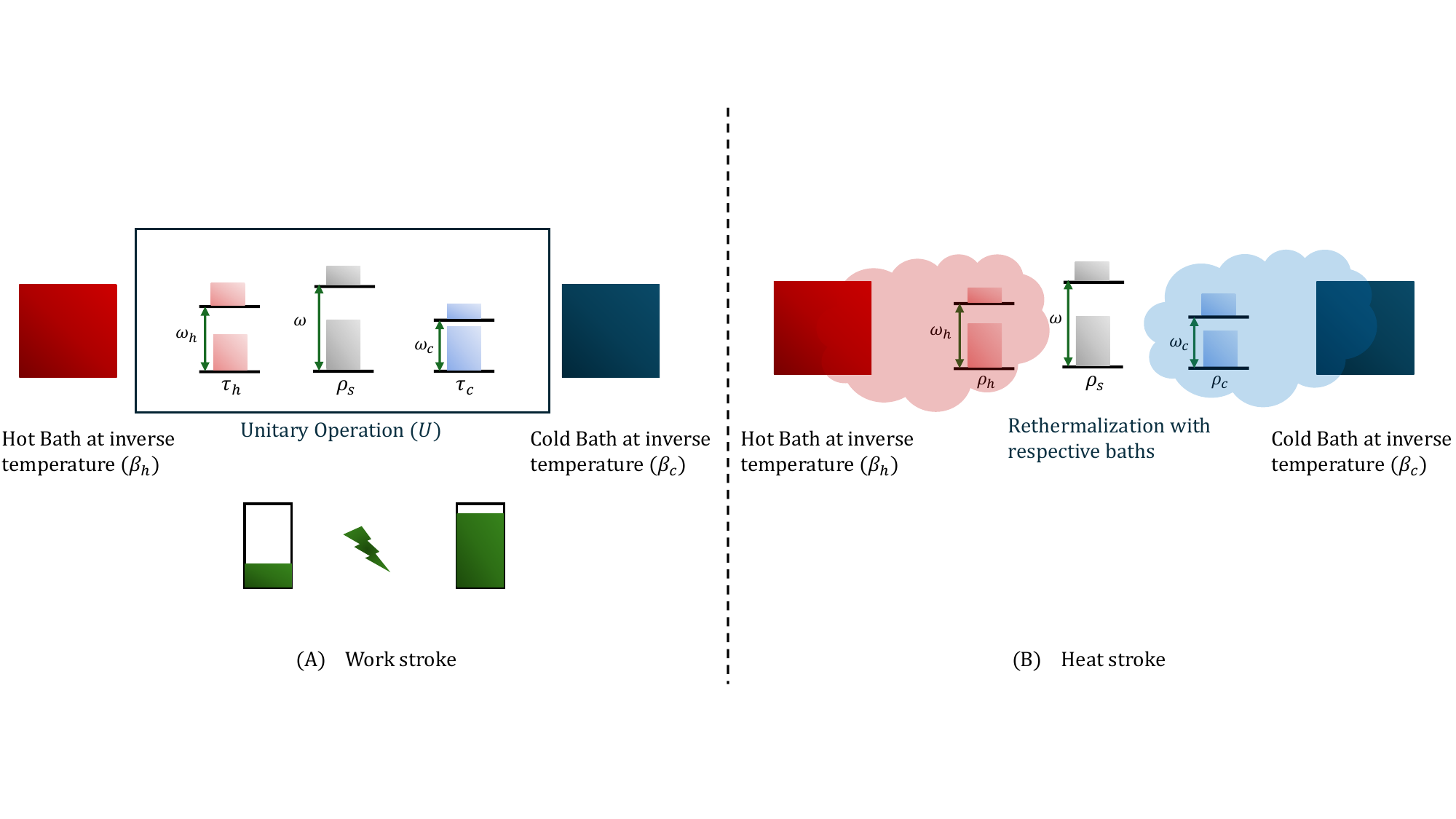}
\centering
\caption{  \textbf{ Schematic representation of the two-stroke heat engine with a working body composed of a hot two-level system, cold two-level system and a two-level catalyst:}  The initial state of the working body of the two-stroke engine consists of hot and cold two-level systems described by Hamiltonian $H_h = \omega_h|1\rangle\langle 1|$ and $H_c = \omega_c|1\rangle\langle 1|$ as well as a catalyst with Hamiltonian $H_s=\omega|1\rangle\langle 1|$. The initial state of the hot and the cold two-level systems are Gibbs at inverse temperatures $\beta_h$ and $\beta_c$ denoted by $\tau_h$ and $\tau_c$. The initial state of the two-level catalyst is $\rho_s$. The components of the working body (i.e., catalyst, the hot and the cold two-level system) are uncorrelated initially. The left panel (A) depicts the work stroke of the engine where the initial state of the working body transforms via a unitary $U$. Our goal is to construct $U$ that leads extraction of work by decreasing the average energy of the initial state of the working body subjected to the preservation of the catalyst. The extracted amount of work can be stored via charging a battery (In this paper we will not consider battery explicitly). The right panel (B) depicts the heat stroke of the engine where the final state of the hot and cold two-level systems $\rho_h$ and $\rho_c$ re-thermalize to their respective initial temperatures which enables the engine to operate in a cyclic manner.  In this figure, we illustrate a specific example of a catalyst-assisted two-stroke engine where each component in the working body is two-dimensional. However, in principle, each component present in the working body can have an arbitrary dimension.}
\label{fig:SCHE_without}
\end{figure*} 
In a thermodynamic framework, the role of a catalyst has been addressed at length in studying state interconversion under   unitary  and Gibbs-preserving transformations \cite{brandao2015second,ShiraishiSagawa21,WilmingPRL}. Later, state interconversion with the aid of a catalyst has been generalized to different domains within quantum information theory (see \cite{CDatta_Review, Bartosik_review} and the references therein). Rather than focusing on the generic state-interconversion in a thermodynamic scenario, we shall address the impact of catalysis on the performance of  {heat engines in the microscopic regime}  \cite{HenaoUzdin, Henao2021catalytic,Sparaciari2017}.  In particular, we explore how catalysis can lead to the enhancement of efficiency of a microscopic heat engine.  A crucial distinction emerges here when comparing the size requirements of catalysts for a generic state interconversion with those for catalytic enhancement in the efficiency of microscopic heat engines.   The construction of the required transformation and the catalyst for a generic state interconversion provided in Ref. \cite{WilmingPRL} is elegant, but the size of the catalyst has to be asymptotically large for any desired state interconversion, which makes required transformation extremely intricate and hard to implement experimentally. Our work explores the role of finite-dimensional catalyst. Unlike the asymptotically large catalysts, our investigation reveals that a catalyst of dimension as modest as two can yield catalytic enhancement in the efficiency and extends the range of operation of a microscopic two-stroke heat engines \footnote{Prior to this work, in Ref. \cite{Son_2024NJP} the role of two-dimensional catalyst has been explored in order to extend the set of achievable states via elementary thermal operation.}  This distinction underscores the novel and practical implications of catalysis in these devices.       

Significance of heat engines comes from the  pivotal role they play in transforming heat into work in the realm of thermodynamics. Traditionally, the study of heat engines has been rooted in macroscopic systems, described by classical thermodynamics. However, emergence of quantum mechanics has led to significant progress in understanding the foundational aspects of thermodynamics at the microscopic level. This line of research traces back to thermodynamical analysis of functioning of lasers \cite{Ramsey22,Scovil1959,Scovil2} in the 1950s. Since then, there has been a significant interest in exploring the functioning of quantum thermal machines \cite{Kosloff1984,Kosloff2014, Scully2002, Popescurefrigerator,Allahverdyan2004, Segal,Mahler,PopescuSmallestHeatEngine2010, BrunnerVirtualqubit2012, Alicki2004,Zhang2022,Myers2022,Klimkovsky2013,KosloffLevy2014,Popescurefrigerator,PopescuSmallestHeatEngineprinciple2010,Skrzypczyk_2011,UzdinLevyKosloff2015,MitchisonContemp,Woods2019maximumefficiencyof,Hofer_2017,Bohr_Brask_2015,Landauer_vs_Nernst} and developing their microscopic thermodynamical frameworks \cite{Pusz1978, Alicki_1979, Davies_1978, Esposito2009, horodecki2013fundamental, Skrzypczyk2014}. Advances in techniques of optical traps \cite{Blickle2011, Abah_2012, Rossnagel2016, Lindenfels_2019, Maslennikov_2019}, nitrogen-vacancy centers in diamond \cite{Klatzow_2019}, nuclear magnetic-resonance \cite{Peterson_2019}, atomic \cite{Bouton_2021} and phononic systems \cite{Zanin_2022}, as well as single-electron transistors \cite{Koski_2014, Koski_2015} already enable experimental exploration of quantum effects in   microscopic heat engines and thermal machines in general. 

  This work explores the catalytic enhancement in the performance of the heat engines  by modeling them as a microscopic \emph{two-stroke} thermal machine that begin with a working body composed of two $d$-level systems thermalized at two distinct temperatures
\cite{two_stroke_Allahverdyan,BrunnerVirtualqubit2012,Silvadimension,Landauer_vs_Nernst,ClivazPRL,ClivazPRE,BiswasLobejko,Melo2022,Stroboscopic_Molitor, Bhattacharjee_2020, Piccione,Kosloff1984,MitchisonContemp,UzdinLevyKosloff2015,Woods2019maximumefficiencyof, RRR_Ahmadi_2023}.  Specifically, this engine is subjected to two subsequent strokes: a \textit{work stroke} that in isolation with respect to heat baths extracts work via a unitary process that acts on the working body, and a \textit{heat stroke} that re-thermalizes the components of the working body to their respective initial temperatures.  On the top of this setting, we extend the working body of the heat engine by an auxiliary non-equilibrium system referred to as the \emph{catalyst},  which intervenes during the work stroke only, such that its marginal state is preserved after the operation.  (Therefore, in the non-catalytic scenario, the working body of the two-stroke engine consists of two components: two $d$-level systems thermalized at two distinct temperatures. In contrast, in the catalytic scenario, the working body includes three components: two $d$-level systems  thermalized at two distinct temperatures, along with a non-equilibrium system referred to as the catalyst whose marginal state is preserved in the operation.)  The cyclicity condition and no coupling to the baths  ensures that energy of the catalyst does not influence the work and heat associated with the heat engines (i.e. the Hamiltonian of the catalyst can be arbitrary).   We stress that this feature distinguish a working body containing a catalyst which has considered in this paper from standard working body (cf. a ``mediator'' in Ref. \cite{Piccione}), which ensures an advantage in the performed task solely by increasing a set of possible operations (likewise the concept of catalysis in quantum information). 

The main problem  we want to focus in this paper  is to optimize the performance of the catalyst assisted two stroke   heat engines.  In particular, our objective is to maximize efficiency and work production per cycle. The optimization is done over all unitary transformations performed during the work stroke. This turns out to be highly nontrivial task, even without a catalyst. For instance, efficiency is not convex in the transformations, as both the denominator - the heat transferred to hot bath, as well as the numerator - the work, depend on the transformation. Therefore one cannot restrict the calculations to the set of extremal transformations. The level of difficulty is raised even more in the case of a catalyst, where we need to preserve its state after completion of the cycle. 

In this paper, we aim to partially overcome the above difficulties, by working out a general theory of two stroke engines in the non-catalytic as well as in the  catalytic scenario,  and providing a series of results regarding optimization of their performance. We put more emphasis on the efficiency of the two-stroke engines rather than work extraction, which has been  explored in Ref. \cite{Sparaciari2017}. First of all, we show how the first and the second law holds within the framework of two-stroke heat engines in the presence and absence of the catalyst. We analyse a two-stroke protocol as a general thermal machine that can work as an engine, cooler or   heat accelerator , and we show that  efficiency satisfies the standard inequalities in each case \cite{Campisi1,Campisi2,Campisi3,Campisi4}.  These general considerations allow us to prove that for general two-stroke engine without a catalyst, the optimal transformation during the work stroke is a permutation which exchanges populations among energy levels. Thus to find optimal efficiency without catalyst one needs to optimize over permutations for different frequency and temperature regimes, which, as the problem grows factorially, is a formidable task. Yet, we show that for the two-stroke engine in the non-catalytic scenario, where the working body is composed of two two-level systems thermalized at distinct temperatures,  the Otto efficiency $1-\omega_c/\omega_h$ is the optimal (here $\omega_h,\omega_c$ are frequencies of two-level systems interacting with hot, and cold bath, respectively). This, together with the result from the companion paper in Ref. \cite{BiswasLobejko} shows that catalysis can enhance efficiency to a large extent, boosting from the Otto efficiency up to \emph{$d$-Otto} efficiency given by $1-\omega_c/(d \omega_h)$, where $d$ is dimension of the catalyst. As expected, in general the price for higher efficiency is lower work production per cycle. This phenomenon is analogue of power-efficiency trade-off in continuous engines \cite{Lobejko2020}. Therefore, we propose a family of permutations which we name \emph{simple permutations} labelled by $n=1,\ldots d$, giving rise to decreasing efficiencies $1-(n\omega_c)/(d\omega_h)$, while increasing amounts of work per cycle. Furthermore, we show that catalysis can remarkably enhance the frequency and temperatures regime of operation for the two-stroke engine by considering illustrative examples.   

Aiming at optimization of efficiency for catalyst assisted two-stroke engine that transforms via simple permutations during the work stroke, we identify the $d$-Otto efficiency is the optimal one. We also make step forward towards optimizing the work produced per cycle, by   formulating  a linear program to compute the upper bound on the produced amount of work in the catalytic scenario. Finally, we ask: for fixed Hamiltonians of the subsystems of the working body, is it always possible to enhance the efficiency of a given two-stroke engine by including a catalyst with the working body?  We answer this question positively for engines with final state of the components of the working body correlated, or not being a product of Gibbs states.

This paper is structured as follows: In Sec. \ref{Description_self_contained} we describe the components and functioning of a general two-stroke thermal machines. We define the thermodynamic quantities like heat, work and efficiency in Sec. \ref{sec:def_heat_work}. 
Our focus then narrows to the two-stroke thermal machines operating as a heat engine   characterised by positive production of work per cycle . In Sec. \ref{Self_contained_without_catalyst}, we characterize a finite set of transformations that can lead to optimal efficiency for the microscopic two-stroke engines. Consequently, we are able to devise the two-stroke heat engines with a working body composed of two two-level systems thermalized at two distinct temperatures, and calculate their maximum achievable efficiency.  Sec. \ref{catalytic_enhancement} analyzes the role of catalysis in enhancing the efficiency of two-stroke heat engines and broadening the operational regime of the two-stroke engine. We close the analysis of catalytic enhancements in efficiency by exploring specific scenarios involving microscopic two-stroke engines whose working body consists of $d$-level systems, emphasizing the guaranteed catalytic enhancement of efficiency.   Next, we focus on work produced by the two-stroke engines. In Sec. \ref{Work_extraction_section}, we provide a closed-form expression for the work produced per cycle when the working body of engine is composed of two two-level systems thermalized at two distinct temperatures and $d$- dimensional catalyst.

\section{Description of the two-stroke thermal machines}\label{Description_self_contained}

  In this section, we provide an overview of two-stroke thermal machines in general. In the following section, we will narrow our focus specifically to two-stroke heat engines in the microscopic regime.

A two-stroke  thermal machine  consists of a   working body composed of  two $d$-level systems that are in thermal equilibrium with two separate heat baths. The first $d$-level system $(h)$ is referred to as the \emph{hot $d$- level system}  and is described by the Hamiltonian $H_h$. It is connected to a heat bath at an inverse temperature $\beta_h$. Similarly, the second $d$-level system $(c)$ is known as the \emph{cold $d$- level system} and is described by Hamiltonian $H_c$. It is connected to a heat bath with an inverse temperature $\beta_c$, where $\beta_c>\beta_h$, which justifies their names.   Additionally, one can incorporate an auxiliary system $(s)$ with the working body of the two-stroke heat engines described by Hamiltonian $H_s$ in such a way that the marginal state of that system is preserved at the end of the transformation i.e.,
\begin{equation}\label{catalyst_first_eqn}
    \Tr_{h,c}(\rho^{i}_{s,h,c}) = \Tr_{h,c}(\rho^{f}_{s,h,c}),
\end{equation}
where $\rho^{i}_{s,h,c}$ and $\rho^{f}_{s,h,c}$ denotes the initial and final state of the working body of the heat engine, respectively.  This guarantees that any work or heat associated with the thermal machine results exclusively from the temperature difference of the bath. We shall refer to this auxiliary system as the \emph{catalyst}. In contrast to the hot and cold $d$-level systems, which are in thermal equilibrium with their respective baths, typically the catalyst is not in equilibrium with either of the baths.   

We assume that initially there is no correlation among the   components of the working body (i.e., hot and cold $d$-level system and the catalyst.) of the thermal machines. Thus, we can write the initial state of the   working body  of the thermal machine as 
\begin{equation}
    \rho^{i}_{s,h,c} = \rho_{s}\otimes\tau_h\otimes\tau_c,
\end{equation}
where $\rho_{s}$ is the initial state of the catalyst, $\tau_h$ and $\tau_c$ denotes the thermal state of hot and cold $d$-level system i.e., 
\begin{equation}\label{Gibbs_defn}
    \tau_h = \frac{e^{-\beta_hH_h}}{\Tr(e^{-\beta_hH_h})}\quad;\quad\tau_c = \frac{e^{-\beta_cH_c}}{\Tr(e^{-\beta_cH_c})}.
\end{equation}

We treat the   working body of the two-stroke thermal machine (i.e hot and cold $d$-level system and the catalyst) as an isolated system. In the work stroke, the thermal machine operates via switching on an interaction between all of the components of the working body, which results a joint unitary operation $U$ acting on the catalyst and the hot and cold $d$-level systems. Therefore, the initial state of   working body  of the thermal machine $\rho^{i}_{s,h,c}$ is related with its final state $\rho^{f}_{s,h,c}$ by the unitary transformation $U$ i.e.,
\begin{equation}
    \rho^{f}_{s,h,c} = U \rho^{i}_{s,h,c} U^{\dagger}.
\end{equation}
subjected to the condition of the preservability of the catalyst state given in Eq. \eqref{catalyst_first_eqn}. The work stroke is responsible for extraction of the work enabled by non-passivity of the initial state of the thermal machine $\rho^{i}_{s,h,c}$.   By non-passivity of a state, we refer to the ability to decrease the average energy of the state through a unitary operation.   Once the extraction of the work is complete, interaction among all the components of the working body is switched off. 

  After the implementation of $U$, in the heat stroke the hot and the cold $d$ level systems of  the working body rethermalize to their respective bath temperatures.  The stroke requires switching on the interaction between the systems and their respective heat baths, which is subsequently switched off once thermalization is achieved. The process of rethermalization leads to change in the energy of the respective baths which can be associated with the transfer of heat.
Having outlined all the essential ingredients required to define the catalyst-assisted two-stroke thermal machine, we now proceed to its formal definition:

\begin{defn}[Catalyst assisted two-stroke thermal machine]\label{principles}
The catalyst assisted two-stroke thermal machine is a thermodynamic device that begins   with a working body  in an initial state 
\begin{eqnarray}
    \rho^{i}_{s,h,c} = \rho_{s}\otimes\tau_h\otimes\tau_c,
\end{eqnarray}
and operate in two discrete strokes to accomplish a desired task in the following manner: 
\begin{enumerate}
    \item \emph{Work stroke:} In this stroke, the interaction among the components of   the working body  of thermal machine is switched on. As a consequence, the initial state of the   working body  $\rho_{s,h,c}^i$ undergoes a transformation governed by the unitary $U$:
    \begin{equation}
    \rho^{i}_{s,h,c} \to U\rho^{i}_{s,h,c} U^\dag = \rho^f_{s,h,c},
\end{equation}
where $U$ satisfies the preservability of the catalyst given in Eq. \eqref{catalyst_first_eqn}, i.e.,
\begin{equation}\label{eq:marginal_cyclic}
    \Tr_{h,c}\big(U\rho_{s,h,c}^iU^{\dagger}\big) = \Tr_{h,c}\big(\rho_{s,h,c}^i\big)=\rho_s.
\end{equation}
\item \emph{Heat stroke:} In this stroke the interaction among the components of the   working body  is switched off. Hot and the cold $d$-level systems   of the working body  rethermalize to their initial temperatures via  interaction with their respective heat baths, i.e.
\begin{eqnarray}
    U \rho^{i}_{s,h,c} U^\dag &\to& \Tr_{h,c}[U \rho^{i}_{s,h,c}  U^\dag] \otimes \tau_h \otimes \tau_c \nonumber\\&=& \rho_s\otimes\tau_h\otimes\tau_c = \rho^{i}_{s,h,c}.
\end{eqnarray}
\end{enumerate}   
\end{defn}
%

After the end of the heat stroke the   the working body of the  thermal machine returns to its initial state that enables the thermal machine to function in a cyclic manner. 

Note that the unitary   in the work stroke  should be constructed based on the intended task to be achieved, and in general does not need to decrease the average energy of the   initial state of the working body of the thermal machines . For instance, in order to lower the temperature of the hot $d$-level system   of the working body , we should apply the unitary transformation $U$ which decreases its local average energy. On the other hand, in order to extract work, $U$ minimizing the global average energy of the   working body  should be selected. In the former scenario, the thermal machine functions as a cooler, while in the latter, it operates as a heat engine.

Fig. \ref{fig:SCHE_without} represents a two-stroke thermal machine operating as a heat engine, where the   components of the working body i.e.,  the hot, the cold $d$-level system and the catalyst are of dimension two. Note that a two-stroke thermal engines differs from  stroke-based thermal engines considered in Ref. \cite{Lobejko2020, BiswasQuantum, Niedenzu2019, Non-markovian_Pstaz, BiswasDatta3stroke}. The latter involves three strokes.

Definition \ref{principles} serves also as the definition of \emph{two-stroke thermal machines without catalyst}, when the catalyst system is removed   from the working body of the thermal machines , together with the constraint assuring preservation of its state. The two-stroke refrigerator and engine without a catalyst has been considered earlier in the literature \cite{BrunnerVirtualqubit2012, Silvadimension, ClivazPRL, two_stroke_Allahverdyan, Melo2022,Stroboscopic_Molitor}. We compare with two-stroke engines without a catalyst to identify a catalyst-driven improvement in performance of these engines.

\section*{Main Results}
Having introduced the model of a generic two-stroke thermal machine, we now turn our attention to the main results of this paper, focusing specifically on two-stroke heat engines in the microscopic regime. We start by defining the concepts of heat and work, ensuring that the first and second laws hold for the catalyst-assisted two-stroke heat engine.

\section{Thermodynamics of two-stroke heat engines}\label{sec:def_heat_work}

Let us now focus on the two-stroke heat engine in the microscopic regime that starts with the working body consists of a hot and cold 
$d$-level system thermalized at two distinct temperatures, as well as a catalyst . A unitary transformation is applied during the work stroke to the working body in order to extract work. This unitary transformation is constrained to preserve the marginal state of the catalyst as given in Eq. \eqref{eq:marginal_cyclic}. This process reduces the average energy of the working body which manifested into produced amount of work by the engine. Therefore, we shall define work as
\begin{equation}\label{defn_of_work2}
    W = \Tr\big((H_s+H_h+H_c)(\rho^{i}_{s,h,c}-\rho^{f}_{s,h,c})\big),
\end{equation}
where $\rho^{f}_{s,h,c} = U \rho^i_{s,h,c} U^{\dagger}$ is the final state of the thermal machine. In other words, we define work as the change in the energy of the initial state of the working body of the thermal machine $\rho^i_{s,h,c}$ during the work stroke. This definition is in accordance with classical thermodynamics where work is referred as the change in the energy of a system at constant entropy.

On the other hand, the amount of heat transferred to the thermal machine from the hot heat bath is given by 
\begin{equation}\label{defn:hot_heat}
    Q_h = \Tr\big(H_h(\rho^i_{s,h,c}-\rho^f_{s,h,c})\big) = \Tr(H_h(\tau_h-\rho_h)),
\end{equation}
where $\rho_h$ is the final state of the hot $d$-level system. This definition of heat is also justified based on the traditional perspective of classical thermodynamics, where heat is defined as the amount of energy incoming from the hot heat bath. 
 
Similarly, we also introduce the amount of heat dumped by the thermal machine in the cold bath in a cycle as: 
\begin{equation}\label{defn:cold_heat}
    Q_c = \Tr\big(H_c(\rho^i_{s,h,c}-\rho^f_{s,h,c})\big) = \Tr\big(H_c(\tau_c-\rho_c)\big),
\end{equation} 
where $\rho_c$ denotes the final state of the cold $d$-level system. We refer $Q_c$  as \emph{cold heat}. Using the quantity $Q_h$ and $Q_c$, one can simplify the definition of work starting from Eq. \eqref{defn_of_work2} as follows:
\begin{eqnarray}\label{Work_redefined}
     W &=&\Tr\Big((H_s+H_h+H_c)(\rho^{i}_{s,h,c}-\rho^{f}_{s,h,c})\Big)\nonumber\\&=&\Tr\Big(H_h(\rho^{i}_{s,h,c}-\rho^{f}_{s,h,c})\Big)+\Tr\Big(H_c(\rho^{i}_{s,h,c}-\rho^{f}_{s,h,c})\Big)\nonumber\\
     &=& \Tr\big(H_h(\tau_h-\rho_h)\big)+\Tr\big(H_c(\tau_c-\rho_c))\nonumber\\
    &=& Q_h+Q_c,
\end{eqnarray}
where the second equality follows from the fact that the initial state of the catalyst is the same with its final state, as given in Eq. \eqref{eq:marginal_cyclic}. It is worth noting that work produced by the engine given in Eq. \eqref{Work_redefined} is in accordance with the first law.

Since the final marginal state of the catalyst is the same as the initial one, the change in the energy of the catalyst does not contribute to the work produced by the engine. This can also be seen from Eq. \eqref{Work_redefined}, where work produced by the heat engine depends only on $Q_h$ and $Q_c$. Thus, without loss of generality, for the rest of the paper we assume that the Hamiltonian of the catalyst is trivial.  

  As the focus of this paper is two-stroke heat engine in the microscopic regime, we shall consider only those unitary processes in work stroke that makes the work given in Eq. \eqref{defn_of_work2} positive. However, a generic microscopic two-stroke thermal machines upon choosing suitable unitary during the work stroke can also function as a cooler by reducing the average energy of a cold $d$-level system  $(W<0$,\; $Q_c>0)$ or as  heat accelerators $(W<0$,\; $Q_h>0)$ \cite{Campisi1,Campisi2,Campisi3,Campisi4}. We have provided a comprehensive thermodynamic framework for all modes (i.e heat engine, cooler, and heat accelerator) of two-stroke thermal machines in the appendix \ref{framework_thermodynamic}. In particular, we supplement the first law for the two-stroke thermal machines (given by Eq. \eqref{Work_redefined}) with the second law (Clausius) inequality:
\begin{eqnarray} \label{Clausius_ineq}
    \beta_h Q_h + \beta_c Q_c \le 0.
\end{eqnarray}
Let us now formally define two stroke heat engine:
\begin{defn}[Two-stroke heat engine]
    A two-stroke thermal machine works as an engine if work associated with it is positive i.e.,
    \begin{eqnarray}
        W > 0.
    \end{eqnarray}
\end{defn}

For heat engines, with positive work extraction $W > 0$, the inequality in Eq. \eqref{Clausius_ineq} implies that $Q_h > 0$ and $Q_c < 0$, such that $W = Q_h - |Q_c| > 0$. Then, the engine is characterized by its efficiency (defining the ratio of an output work to an input hot heat):
\begin{eqnarray} \label{efficiency_def}
    \eta = \frac{W}{Q_h} = 1 + \frac{Q_c}{Q_h} = 1 - \frac{|Q_c|}{Q_h}
\end{eqnarray}
(assuming that $W>0$ and $Q_h>0$). Finally, from the Eqs. \eqref{Work_redefined}, \eqref{Clausius_ineq} and definition \eqref{efficiency_def} follows that whenever a machine works as an engine (i.e., when $W>0$), then the efficiency is bounded by: 
\begin{equation}\label{Carnot_bd}
    0 < \eta < 1 - \frac{\beta_h}{\beta_c}:=\eta_{\text{Carnot}}.
\end{equation}

%

To demonstrate the catalytic enhancements in the engine performance, it is essential to understand what is the optimal performance of the engine without the catalyst. In the next section, we characterize the transformations that results in achieving the optimal efficiency for a two-stroke engine without a catalyst.  Subsequently, we derive the expression for the optimal efficiency for the two-stroke engine whose working body is composed of two two-level systems thermalized at two distinct temperatures. 

\section{Two-stroke heat engine without a catalyst}\label{Self_contained_without_catalyst}

We begin by emphasizing that efficiency in general is not a convex function of the protocols determined by the unitaries acting on the initial state of the working body of the engine during the work stroke. Thus optimizing the efficiency of the two-stroke engine is not a trivial task. Below we will prove that for two-stroke engines, optimal efficiency is nevertheless obtained by maximizing solely over extreme protocols - which are actually permutations that swaps the population among different energy levels of the working body.  We then use this result in Sec. \ref{qubit_otto} to show that, for a non-catalytic two-stroke heat engine where the hot and cold $d$-level systems in the working body are two level system, the maximum achievable efficiency is given by the Otto efficiency. 

\subsection{Optimal performance of the the two-stroke heat engine without catalyst}

Let us start with characterization of transformations that lead to optimal efficiency for a two-stroke engine without a catalyst. In order to optimize the efficiency, we show that among all the unitaries $U$ that acts on the initial state of the working body (composed of hot and cold $d$-level system only) during work stroke, only a permutation can lead to the maximum efficiency and production of work. This is captured by the following theorem:
\begin{thm}\label{thm_optimal_efficiency}
    The maximum work extraction and efficiency of a two-stroke heat engine without a catalyst are achieved when the initial state of the   working body  consists of hot and cold \(d\)-level system i.e., \(\rho^{i}_{h,c} = \tau_{h}\otimes\tau_{c}\), transforms via a unitary which permutes the population among different energy levels of the working body. 
\end{thm}
In other words, theorem \ref{thm_optimal_efficiency} is saying the optimal efficiency of the two-stroke thermal machine operating as heat engine can be obtained by taking the maximum value of efficiency associated with all those permutations for which the amount of extracted work is positive. 

Also, note that  $\rho^{i}_{h,c}$ is diagonal in the total energy eigenbasis and work defined in Eq. \eqref{defn_of_work2} is linear, it is easy to see that maximum work is achieved if the initial state of the working body transforms via a permutation which is given by the ergotropy of the initial state \cite{Allahverdyan2004,Francica2020}. On the other hand, optimization of the efficiency is not straightforward. The proof relies on the thermodynamic framework for the two-stroke thermal machine discussed in Appendix \ref{framework_thermodynamic}, and the full proof of Theorem \ref{thm_optimal_efficiency} is given in Appendix \ref{Permutation_optimal_eff}. However, there might exists temperature regimes and the Hamiltonians for hot and cold $d$-level system of the working body, such that optimal work produced by the engine and optimal efficiency may not be achieved for the same permutations.

The theorem recasts finding optimal efficiency as a problem of optimization over permutations, whose number grow factorially with system size. Nevertheless, as we shall see in the next section, Theorem \ref{thm_optimal_efficiency} extremely simplifies the calculation of optimal efficiency for the two-stroke heat engine with a working body composed of two two-level systems thermalized at two different temperatures.  On the other hand, optimization of work is much easier,  as for any given initial state, the optimal work is achieved by a specific permutation: the one leading to a passive state which is the state with least average energy that can be obtained via unitary transformation. \cite{Allahverdyan2004} 
 
More explicitly, the passive state obtained from a generic state $\rho$ with Hamiltonian $H$ is defined by the following relation:
\begin{equation}\label{defn_passive_state}
    \min_{U\in\mathbb{U}(d)} \Tr(HU\rho U^{\dagger}):=\Tr(H\rho_{\text{pass}}),
\end{equation}
where $\rho_{\text{pass}}$ referred as \emph{passive state} obtained from $\rho$ and $\Tr(H\rho_{\text{pass}})$ referred as \emph{passive energy} of the state $\rho$. Ergotropy of state denoted as $R(\rho)$ can be calculated as the difference of average energy and passive energy of the state $\rho$ i.e.,
\begin{equation}
    R(\rho):= \Tr(H\rho)-\Tr(H\rho_{\text{pass}}).
\end{equation}

\subsection{ Optimal performance of the non-catalytic two-stroke heat engine with a working body composed of two two-level systems thermalized at two distinct temperatures }\label{qubit_otto}

\begin{figure}[t]
\includegraphics[height=6cm]{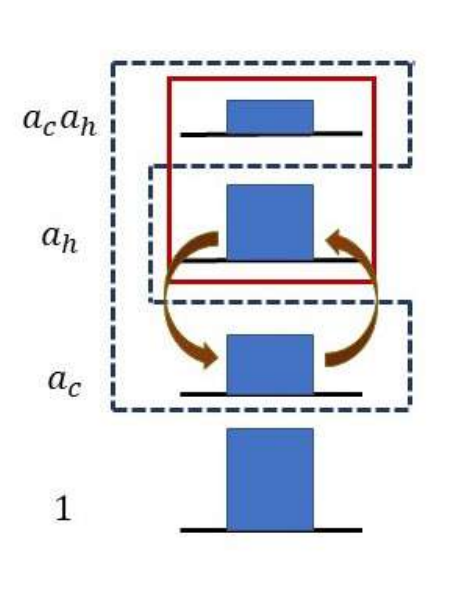}
\centering
\caption{This picture shows the permutation leads to optimal production of work and efficiency for the non-catalytic two-stroke heat engine. This permutation exchanges the population between the second and the third excited state.  This results a net population flow from the subspace where hot two-level system is excited (enclosed in red solid square)  to the subspace where cold two-level system is excited (enclosed by blue dashed lines).} 
\label{fig:smallest_self_cont}
\end{figure} 

In this section, we will calculate the optimal efficiency of a non-catalytic two-stroke heat engine where the hot and cold $d$-level systems, $\tau_h$ and $\tau_c$, present in the working body are two level system.   We state the result in the following theorem:

\begin{thm}[Optimal efficiency of two stroke engine with hot and cold $d$-level system of dimension two]
   Consider a two-stroke heat engine where the initial state of the working body is given by $\tau_h\otimes\tau_c$ such that $\tau_h$ and $\tau_c$ are two-level system with Hamiltonian 
 \begin{equation}
     H_i = \omega_i |1\rangle\langle 1| \quad\text{with}\quad i\in\{h,c\},
 \end{equation}  
 thermalized at inverse temperature $\beta_h$ and $\beta_c$ obeying the following condition:
 \begin{equation}
     \beta_h\omega_h\leq\beta_c\omega_c \quad \text{and} \quad \omega_h>\omega_c.
 \end{equation}
Then, the maximum efficiency achieved by such an engine is given by the Otto efficiency 
\begin{equation}
    \eta_{\max} = 1-\frac{\omega_c}{\omega_h}.
\end{equation}
\end{thm}
  
\begin{proof}[Proof sketch]
    Let us proceed by writing the initial state of $\tau_h$ and $\tau_c$ explicitly as follows:
    \begin{eqnarray}\label{qubit_description}
    \tau_i&=& \frac{1}{1+e^{-\beta_i\omega_i}}(|0\rangle\langle 0|+e^{-\beta_i\omega_i}|1\rangle\langle 1|)
\end{eqnarray}
where $i\in\{h,c\}$ and $\beta_h\omega_h\leq\beta_c\omega_c, \;\; \omega_h>\omega_c$. Observe that, the two-stroke engine in this scenario can not extract work if $\omega_c>\omega_h$ as this makes the state $\tau_h\otimes\tau_c$ passive. According to theorem \ref{thm_optimal_efficiency}, the optimal efficiency of any two-stroke engine is achieved when the working body is transformed via some permutation. Since, the dimension of the working body is $4$ (as dimension of $\tau_h$ and $\tau_c$ is $2$). In order to prove this, we can easily obtain the expression of efficiency for each of the $4!=24$ permutations (see the table \ref{tab:my_table_efficiency_24} from appendix \ref{SmallestHE}). One can see from the table \ref{tab:my_table_efficiency_24}, among all the 24 permutations, only four of them lead to a positive amount of work production. Thus, only on those four scenarios, the thermal machine is operating as an engine. Then, we can compare the efficiency for each of the $4$ permutations  and see the optimal efficiency is achieved for the permutation 
\begin{equation}\label{Pie}
    \Pi = (|0,0\rangle\langle 0,0|+ |0,1\rangle\langle 1,0|+|1,0\rangle\langle 0,1|+|1,1\rangle\langle 1,1|)_{h,c},
\end{equation}
where the first and second index in $|\cdot\rangle\langle\cdot|$ corresponds to the energy levels of the hot and cold $d$-level systems of the working body, respectively (see Fig. \ref{fig:smallest_self_cont}).  The work produced  and the consumed heat by the engine when the working body  is transformed by the permutation $\Pi$ is given by: 
\begin{eqnarray}\label{W}
    W &=& \mathcal{N}(a_h-a_c)(\omega_h-\omega_c),\label{Otto_work},\\
    Q_h&=& \mathcal{N}(a_h-a_c)\omega_h\label{Otto_heat},
\end{eqnarray}
where 
\begin{equation}
    a_h := e^{-\beta_h\omega_h},\quad a_c := e^{-\beta_c\omega_c},\quad\mathcal{N} = \frac{1}{(1+a_h)(1+a_c)}.
\end{equation}
This allows us to compute the optimal efficiency as
\begin{eqnarray}\label{Otto}
    \eta = 1-\frac{\omega_c}{\omega_h},
\end{eqnarray}
which is also known as the the \emph{Otto efficiency}  \cite{KosloffLevy2014,cangemni_Levy_engines}.
The detailed calculation can be found in  Appendix \ref{SmallestHE}. 
\end{proof}
Let us make two important observations here: \emph{First}, the maximum amount of work produced by the engine is attained by the same permutation $\Pi$ in Eq. \eqref{Pie} as well, since it leads to the passive version of the initial state of the working body. As the permutation in Eq. \eqref{Pie} leads to the simultaneous maximization of work and efficiency, for a qubit two-stroke heat engine the work-efficiency trade-off relation becomes trivial.

\emph{Second}, note that for $\beta_h\omega_h = \beta_c\omega_c$ the engine reaches the Carnot efficiency, whereas work produced by the engine becomes zero. This is in agreement to our understanding of classical thermodynamics, according to which the amount of work produced per cycle approaches zero when an engine operates with Carnot efficiency. In the next section, we show that one can surpass the optimal Otto efficiency given in Eq. \eqref{Otto} with the use of a catalyst.

\section{Catalytic enhancement}\label{catalytic_enhancement}
\begin{figure*}[htbp]
    \centering
    \includegraphics[width=17cm]{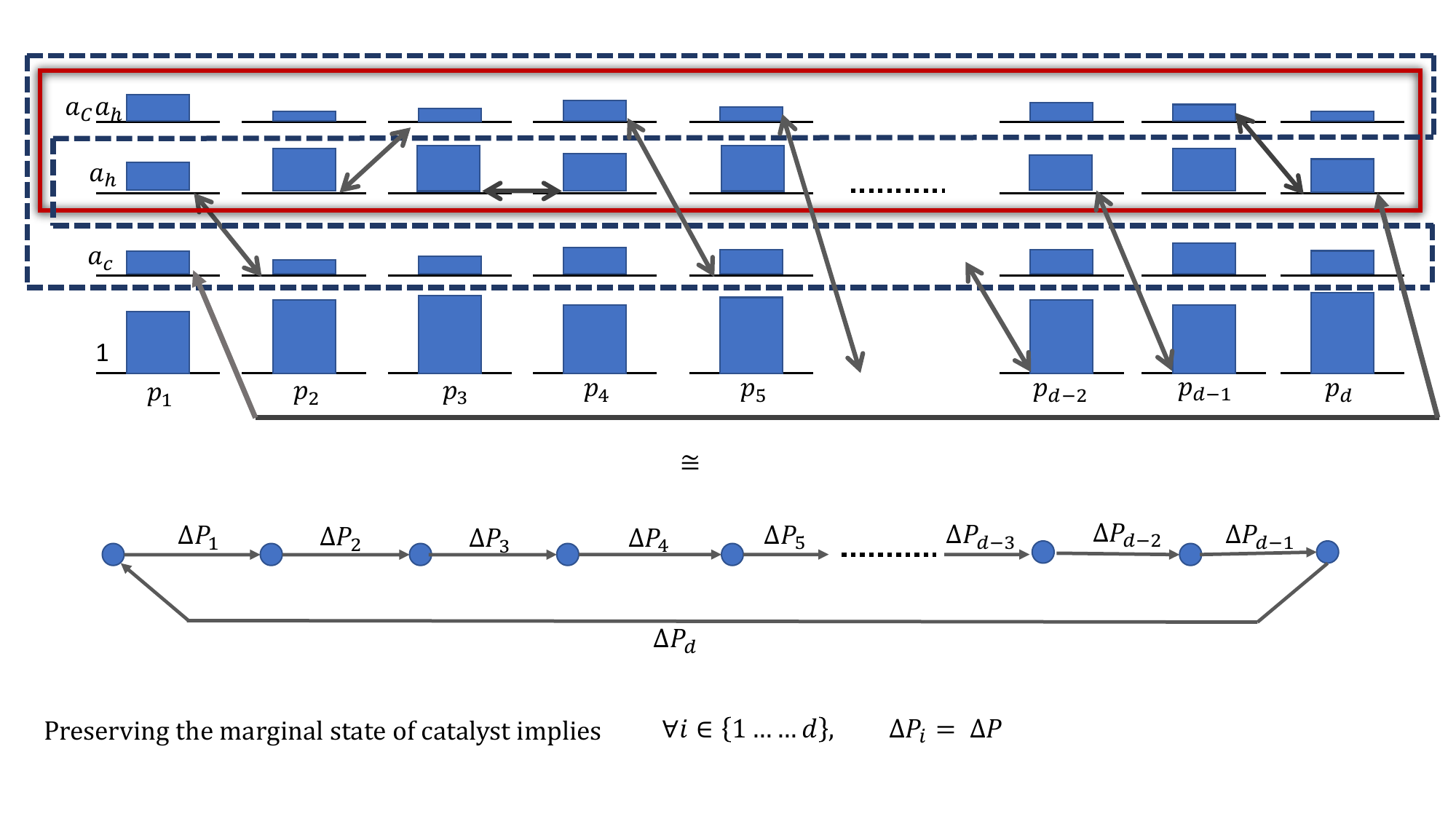}
    \caption{\label{fig:swap_vs_graph} \textbf{A simple permutation:} The above panel depicts a simple permutation with arrows that denotes the exchange the populations among the associated levels. We denote the population associated with energy levels of the working body $|i,0,0\rangle_{s,h,c}$, $|i,0,1\rangle_{s,h,c}$, $|i,1,0\rangle_{s,h,c}$, and $|i,1,1\rangle_{s,h,c}$ by $\frac{p_i}{(1+a_c)(1+a_h)}$, $\frac{p_ia_c}{(1+a_c)(1+a_h)}$, $\frac{p_ia_h}{(1+a_c)(1+a_h)}$, and $\frac{p_ia_ca_h}{(1+a_c)(1+a_h)}$, respectively. The enclosed subspace by the red rectangle denotes the hot subspace (taking the sum over all the population enclosed inside the red rectangle gives the population of the excited level of the hot qubit). Similarly, the enclosed subspace by the blue dashed line denotes the cold subspace (taking the sum over all the population enclosed by the blue dashed line rectangle gives the population of the excited level of the cold qubit). The below panel depicts the four energy levels for a fixed $i$ by a blue vertex. The net flow of population from $i^{\text{th}}$ vertex to $(i+1)^{\text{th}}$ is denoted by $\Delta P_i$. In order to preserve the marginal state of the catalyst as per Eq. \eqref{eq:marginal_cyclic},  $\forall j$ it requires to satisfy  $\Delta P_j = \Delta P$  }
\end{figure*}
The main challenge in optimizing work production and efficiency arises from the variability of the state of the catalyst along with the unitary $U$.   This leads to different initial state of the working body of the engine for different unitary implemented in the work stoke. This complexity makes the problem hard to solve, even when the catalyst is a two-level system. To tackle this, we will confine ourselves to a specific set of transformations, which we call the \emph{simple permutations}. This set of permutations enables for an arbitrary increase of efficiency of the two-stroke engine up to the Carnot limit. Our aim is to derive a closed form for the maximum efficiency of a two-stroke engine where the initial state of the working body undergoes transformation through simple permutations. 

Before proceeding to the optimization of the efficiency, we shall show that coherence in the energy basis of the state of the catalyst can not influence the engine performance. Hence, it is enough to investigate only energy-incoherent catalysts.

\subsection{Coherence in the initial state of the catalyst can not enhance the performance of the two-stroke heat engine}
The   sufficiency  to consider only catalysts which are diagonal in the energy basis is captured by the following Proposition:
\begin{prop}\label{proposition_coherence_useless}
    For any catalyst assisted two-stroke heat engine with an initial state of   the working body $\rho_s\otimes\tau_h\otimes\tau_c$ that produces work  $W$ with efficiency $\eta$, one can construct a unitary that transforms the initial state of working body of another two-stroke heat engine $\tilde{\rho}_s\otimes\tau_h\otimes\tau_c$,
     such that state of the catalyst $\tilde{\rho}_s$ satisfies $[\tilde{\rho}_s,H_s]=0$ and produces exactly same work $W$ with the efficiency $\eta$.     
\end{prop}
  We would like to make a remark here. The implication of the theorem \ref{proposition_coherence_useless} is trivial if the Hamiltonian of the catalyst is zero. From the above proposition, we see that even if catalyst Hamiltonian is non-trivial, then also coherence in the eigenbasis of the catalyst Hamiltonian do not play any role in determining the performance of the microscopic two-stroke heat engine. Therefore, we only consider catalyst that are diagonal in the eigenbasis of the catalyst Hamiltonian $H_s$. 

We refer the Reader to Appendix \ref{useless_coherence} for the proof. Now, we are ready to characterize the enhancements enabled by including the catalyst within the engine model. 

\subsection{Two kinds of catalytic enhancements: Summary}

 In this section, we summarize the kinds of catalytic enhancements that we have explored in this paper.  We find that both efficiency and regime of operation of the engine can be improved by utilizing a catalyst.    

\subsubsection{Catalytic enhancements in the efficiency}
As described in Eq. \eqref{Carnot_bd}, to produce a positive amount of work by a catalyst-assisted two-stroke engine its efficiency $\eta$ should satisfy 
\begin{eqnarray}\label{Efficiency_range_for_positive_work}
    0<\eta<1-\frac{\beta_h}{\beta_c},
\end{eqnarray}
where $\beta_h$ and $\beta_c$ are the temperature of the hot and the cold $d$-level system of the working body.

On the other hand, a two-stroke engine without a catalyst requires to have $\omega_h>\omega_c$ and $\beta_h\omega_h>\beta_c\omega_c$ to produce a positive amount of work, and Otto efficiency is the optimal one which is given by the formula $1-\frac{\omega_c}{\omega_h}$ where $\omega_c$ and the $\omega_h$ is the frequency associated with the Hamiltonian of the hot and the cold two-level system (see Sec. \ref{qubit_otto}).

Therefore, to prove the catalytic enhancement of the efficiency, we aim to achieve an efficiency $\eta$ for the catalyst assisted two-stroke thermal machine such that
\begin{equation}\label{range_ineq}
    1-\frac{\omega_c}{\omega_h}<\eta<1-\frac{\beta_h}{\beta_c}.
\end{equation}
In other words, we shall include a catalyst with the working body of the engine and construct a unitary that transforms the initial state of the working body such that the above inequality holds. 

\subsubsection{Extending the regime of operation due to catalysis}
The non-catalytic two-stroke heat engine whose working body composed of two two-level systems thermalized at two distinct temperatures  can not extract work if $\omega_c>\omega_h$, because by definition $\beta_c>\beta_h$,  which implies 
\begin{eqnarray}
    \omega_c>\omega_h \Rightarrow  e^{-\beta_h\omega_h} > e^{-\beta_c\omega_c},
\end{eqnarray}
hence forcing the initial state to be passive \cite{Allahverdyan2004}   (We have introduced the notion of passivity in Eq. \eqref{defn_passive_state}).  Nonetheless, with the aid of a catalyst we can run the engine at the same regime of frequencies and temperatures to extract positive amount of work with efficiency $\eta$ that satisfies the inequality in Eq. \eqref{Efficiency_range_for_positive_work}. Therefore, the catalysis can broaden the regime of operation for the two-stroke engine. 

After briefly summarizing the two types of catalytic enhancements, we will explore them in detail in the following section. 

\subsection{Improving the efficiency using a catalyst}\label{Sec:eff_imp_cat}

We proceed by elaborating the catalytic improvement in the efficiency in this section. In particular, we shall show that the catalytic enhancement in the efficiency in Sec. \ref{Improving_the_eff_two} for the microscopic two-stroke heat engine where the dimension of the hot and the cold $d$-level system present in the working body is two. Next in \ref{Improving_the_eff_d}, we identify certain scenarios where the catalytic enhancement in the efficiency can be guranteed for microscopic two-stroke heat engine where the dimension of the hot and the cold $d$-level system can be  arbitrary.   

\subsubsection{ The optimal efficiency - analysis for two-level working body
}\label{Improving_the_eff_two}
\begin{figure*}[htbp]
    \centering
    \includegraphics[width=17cm]{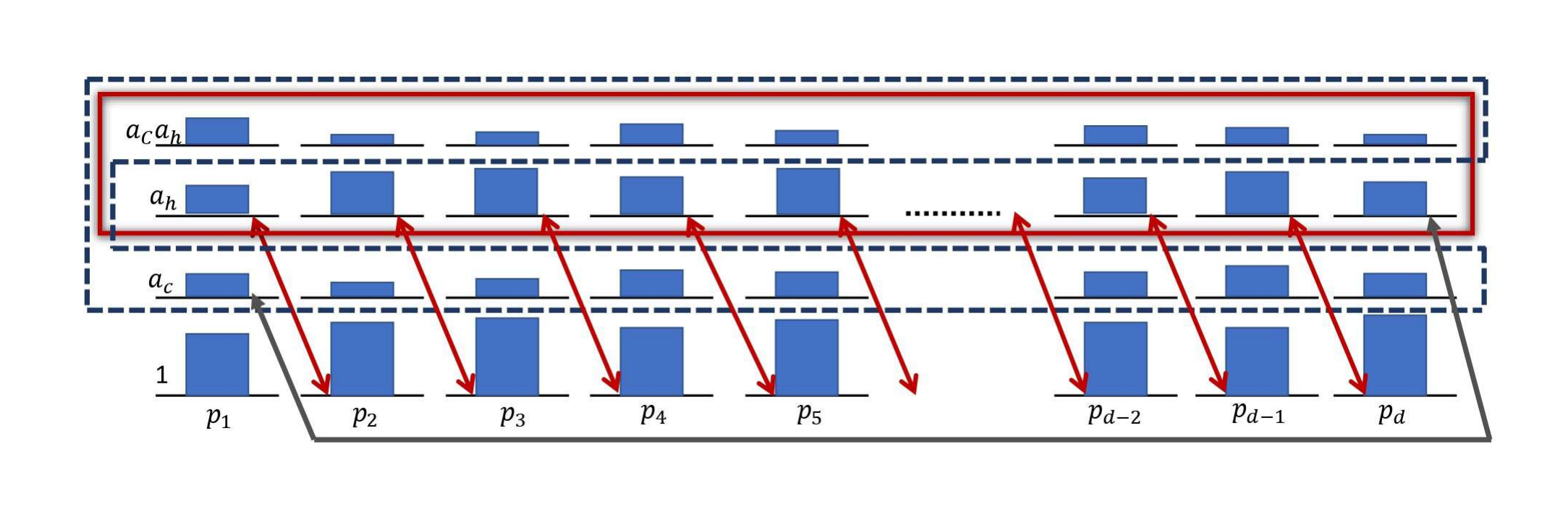}
    \caption{\label{fig:dswap} Depiction of the simple permutation described in Eq. \eqref{Pi1} that leads to efficiency $\eta = 1-\frac{\omega_c}{d\omega_h}$. } 
\end{figure*}
In this section, we optimize the efficiency of the two-stroke heat engine with working body composed of two two-level systems thermalized at two distinct temperatures, assisted with a  $d$-level catalyst. 

We start with introducing a class of unitaries called \textit{simple permutations}. 
  The main motivation for considering this class is two-fold. Firstly, they are straightforward to analyze. Secondly, any efficiency $\eta$ that satisfies Eq. \eqref{Efficiency_range_for_positive_work} can be achieved with arbitrary precision.  This is done by applying an appropriate simple permutation to the initial state of the engine's working body endowed with suitably chosen catalyst during work stroke. Let us define simple permutation formally:

\begin{defn}[Simple permutation] Consider a permutation that acts on the initial state of the engine's working body $\rho^i_{s,h,c} = \rho_s\otimes\tau_h\otimes\tau_c$ with dimension of each component is $d$, $d_h$, and $d_c$, respectively. We denote the energy eigenstate of the total Hamiltonian of working body as $H_s+H_h+H_c$ by $|i,j,k\rangle_{s,h,c}$. A permutation $\Pi$ is referred to as a simple permutation, if for any fixed $i$, $\Pi$ acts on exactly two distinct energy levels labelled by $|i,j_1,k_1\rangle$ and $|i,j_2,k_2\rangle$,  where $\Pi$ results an exchange of 

\begin{enumerate}
    \item the population associated with energy level $|i,j_1,k_1\rangle$ and $|i+1,j',k'\rangle$ and,
    \item the population associated with energy level $|i,j_2,k_2\rangle$ and $|i-1,j'',k''\rangle$,
\end{enumerate}
where $\pm$ involves in $i+1$ and $i-1$, denotes addition or subtraction modulo $d$.
\end{defn}

Fig. \ref{fig:swap_vs_graph} shows an example of a simple permutation, with $\tau_h$ and $\tau_c$ taken as two level systems. In the diagram in the upper part of the Figure, all energy levels of the working body of the engine (two two-level systems and a catalyst) are shown, grouped in four where the groups are labelled by eigenstate of the catalyst Hamiltonian. Let us use arrows to indicate pairs of distinct energy levels whose populations are mapped onto each other by a permutation. Then, a simple permutation corresponds to the configuration of arrows connecting the groups of levels, such that all groups are connected to their neighbours by a single arrow. The graph in the bottom of Fig. \ref{fig:swap_vs_graph} corresponds to the net flow of population among the associated energy levels of the catalyst (indicated by vertices), where direction of the arrow corresponds to the direction of the population flow.  For example, for the permutation shown in Fig. \ref{fig:swap_vs_graph}, $\mathcal{N}p_1a_h$ flows from vertex $1$ to vertex $2$ and $\mathcal{N}p_2a_c$ flows from vertex $2$ to vertex $1$. Thus a net population flow of 
\begin{equation}
    \Delta P_1 = \mathcal{N}(p_1a_h-p_2a_c)
\end{equation}
is created from vertex $1$ to $2$. Importantly, since we demand the marginal state of the catalyst does not change as a result of the permutation implemented in the work stroke, the net population flow inside a particular vertex has to be the same as the population flow out of it. By investigating the population flows associated with other vertices, we see that: 
\begin{equation}\label{cyc_delta_P}
    \Delta P_1 = \Delta P_2 = \ldots = \Delta P_{d-1} = \Delta P_d=\Delta P,
\end{equation}
where $\Delta P_i$ denotes the net transferred amount of population from block $i$ to block $i+1$.

Conceptually, this is similar to Kirchhoff's current law which says that the current flowing into a node (or a junction) must be equal to the current flowing out of it, as a consequence of charge conservation.  Furthermore, since work done by the engine (Eq. \eqref{Work_redefined}) and its efficiency  (Eq. \eqref{efficiency_def}) depend only on the amount of heats $Q_h$ and $Q_c$, the population flow contributing to these heats, fully determine the performance of the engine. To characterize the heat flows, we introduce the concepts of hot and cold subspace.
\begin{defn}[Hot and cold subspace]\label{hot_and_cold_subspaces} 
Consider the two-stroke engines with initial state of the working body is given by $\rho^{i}_{s,h,c}$ and the Hamiltonian $H_s+H_h+H_c$ where the dimension of each component is $d$, $d_h$, and $d_c$, respectively. We denote the energy levels of the Hamiltonian of the working body  $H_s+H_h+H_c$ by $|i,j,k\rangle_{s,h,c}$. We shall define the $j^\text{th}$ hot subspace by 
\begin{equation}
    \mathcal{H}_j:=\text{Span}\;\; \{\forall i, k\;\;|i,j,k\rangle_{s,h,c}\}, 
\end{equation}
and $k^\text{th}$ cold subspace by   
\begin{equation}
    \mathcal{C}_k:=\text{Span}\;\; \{\forall i, j\;\;|i,j,k\rangle_{s,h,c}\}.
\end{equation}
\end{defn}

For instance, in Fig. \ref{fig:swap_vs_graph}, the excited hot and cold subspaces for a two-stroke heat engine are shown by an enclosed red rectangle and a dotted blue line. 

Using the concept of the hot and the cold subspace, for a two stroke engine where the dimension of the hot and cold $d$-level system is $2$, we can simplify $Q_h$, as given in Eq. \eqref{defn:hot_heat} and $Q_c$ given by Eq. \eqref{defn:cold_heat}. As we have assumed that the energy associated with the ground state of the hot and cold two-level system is zero, the amount of heat $Q_h$ can be expressed as a sum of net population flows out of the excited hot subspace, multiplied by the corresponding energy $\omega_h$ of the excited level of the hot two level system, i.e.,
\begin{equation}\label{Heat_jth_subspace}
    Q_h = \sum_j\Delta P_{j}\omega_h.
\end{equation}
Similarly, the cold heat $Q_c$ given by Eq. \eqref{defn:cold_heat} reduces to 
\begin{equation}\label{Cold_Heat_jth_subspace}
    Q_c = \sum_k\Delta P_{k}\omega_c,
\end{equation}
where $\omega_c$ is the energy associated with the excited cold subspace. We use equations \eqref{Heat_jth_subspace} and \eqref{Cold_Heat_jth_subspace} to simplify the calculation of efficiency for the simple permutations due the cyclic structure. 

This simplification mainly arises because $\Delta P_k$ are equal for all $k$ (See Fig. \ref{fig:swap_vs_graph}). If $m$ number of arrows leads to population flow out of the excited hot subspace, then $Q_h = m\Delta P \omega_h$. Similarly, if $n$ number of arrows leads to population flow out of the excited cold subspace, then $Q_c = n\Delta P \omega_h$. On the other hand, if population flows in the excited hot/cold subspace, then corresponding heat will have negative sign (See Eq. \eqref{defn:hot_heat} and \eqref{defn:cold_heat} where $Q_h$ and $Q_c$ depend on the difference of initial and final state)    

We proceed with efficiency optimization of simple-permutation driven two-stroke heat engines, with dimension of $\tau_h$ and $\tau_c$ being two. The result is given by the following Theorem:

\begin{figure*}[htbp]
    \centering
    \includegraphics[width=17cm]{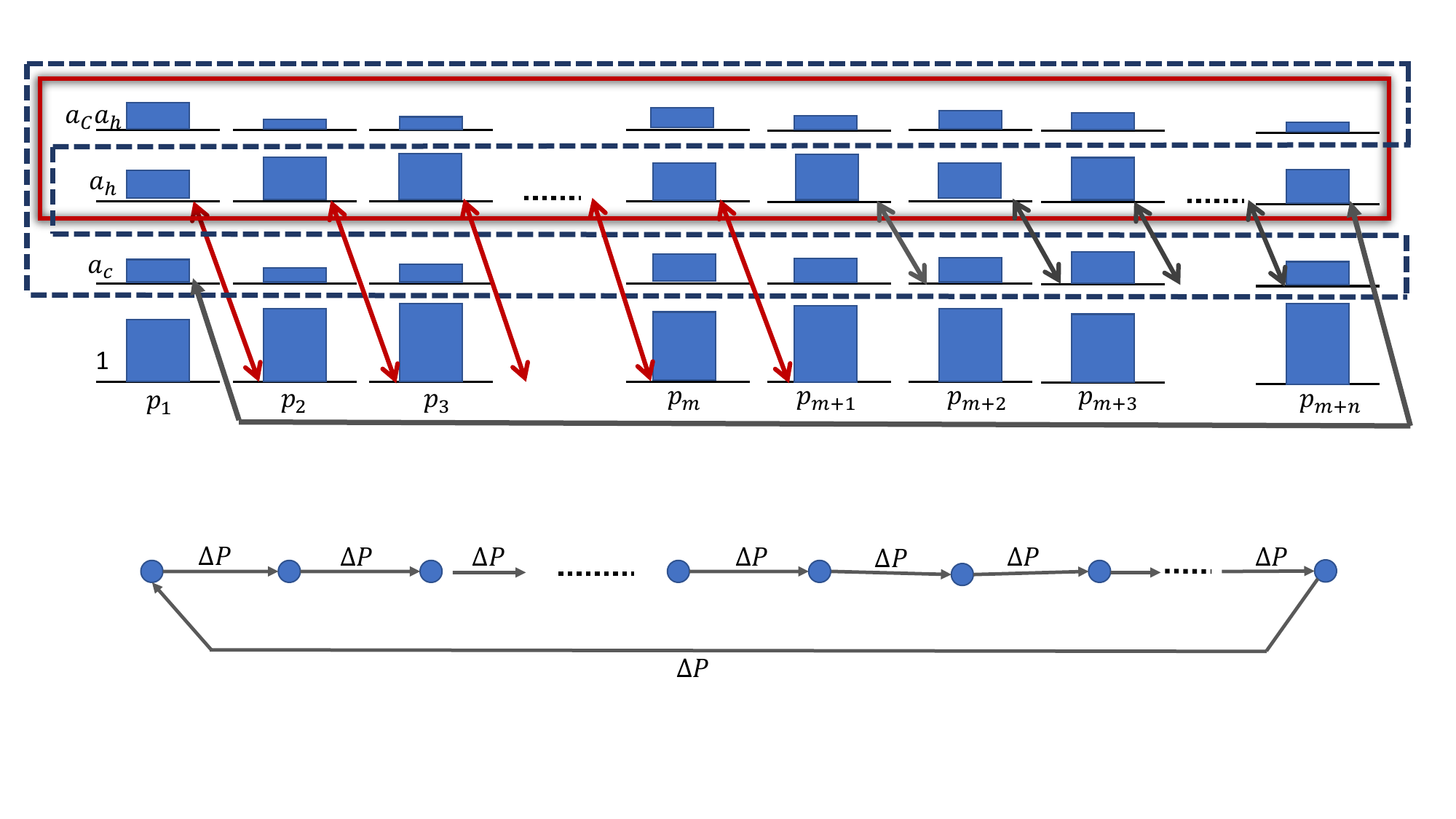}
    \caption{\label{fig:mhotncold} A generic simple permutation described in Eq. \eqref{Pigen} that contains the simple permutation showed in Fig. \ref{fig:dswap} as a special case with $n=1$ and $m=d-1$. This leads to efficiency $\eta = 1-\frac{n\omega_c}{(m+n)\omega_h}$. }
\end{figure*}
\begin{thm}\label{catalyst:enhancement}
    Consider a two-stroke heat engine with initial state of the working body is given by $\rho^{i}_{s,h,c} = \rho_s\otimes\tau_h\otimes \tau_c$ with dimension of $\tau_h$ and $\tau_c$ two. The initial state of the engine $\rho^{i}_{s,h,c}$ transforms via a simple permutation, where the Hamiltonians and the states of the components are 
    \begin{eqnarray}
        H_s &=& 0 \quad;\quad \rho_s= \sum_{j=1}^d p_j |j\rangle\langle j|,\\
        H_h &=& \omega_h|1\rangle\langle 1| \quad;\quad \tau_h = \frac{e^{-\beta_h H_h}}{\Tr(e^{-\beta_h H_h})},\\
        H_c &=& \omega_c|1\rangle\langle 1| \quad;\quad \tau_c = \frac{e^{-\beta_c H_c}}{\Tr(e^{-\beta_c H_c})},
    \end{eqnarray}
    such that 
    \begin{equation}\label{relation_carnot}
        \frac{\omega_c}{   d   \omega_h}\leq1\quad\text{and}\quad d \leq \frac{\beta_c\omega_c}{\beta_h\omega_h}.
    \end{equation}
    Then the optimal efficiency is given by 
    \begin{equation}\label{otto-d}
        \eta_{\text{opt}} = 1 - \frac{\omega_c}{d\omega_h},
    \end{equation}
where $d$ is the dimension of the catalyst. 
\end{thm} 
\begin{proof}[Proof sketch] 
We shall proceed with achievability of the efficiency $\eta_{\text{opt}}$. In order to prove this, we shall consider the simple permutation \cite{BiswasLobejko}:
\begin{equation}\label{Pi1}
\Pi = (\sum_{i=1}^{d-1}|i,1,0\rangle\langle i+1,0,0|+ |1,0,1\rangle\langle d,1,0|)_{s,h,c}+\mathbb{I}_{\text{Rest}},
\end{equation}
as shown in Fig. \ref{fig:dswap} to transform the initial state of the working body during work stroke. Here $\mathbb{I}_{\text{Rest}}$ is identity operator on the orthogonal complement of the first term in the parentheses. This permutation leads to efficiency $1-\frac{\omega_c}{d\omega_h}$. The reason for this can be easily understood from Fig. \ref{fig:dswap}.  Let us assume that the permutation $\Pi$ in Eq. \eqref{Pi1} causes $\Delta \tilde{P}$ amount of net population flow from a vertex to the one on its right. Therefore, there is $d\Delta\tilde{P}$ of net population flow out of the hot subspace (corresponding to red arrows), whereas $\Delta\tilde{P}$ of net population flows into the cold subspace (represented by the grey arrow). From Eq.  \eqref{Heat_jth_subspace} and Eq. \eqref{Cold_Heat_jth_subspace} we obtain
\begin{equation}
    Q_h = d\omega_h\Delta \tilde{P}, \quad Q_c = -\omega_c\Delta \tilde{P}. 
\end{equation}
As the net population flows into the cold subspace, $Q_c$ takes negative sign.
Using the formula of efficiency from Eq. \eqref{efficiency_def}, we obtain $\eta = 1-\frac{\omega_c}{d\omega_h}$. The condition $d\leq \frac{\beta_c\omega_c}{\beta_h\omega_h}$ given in Eq. \eqref{relation_carnot} ensures that the efficiency given in Eq. \eqref{otto-d} is upper bounded by Carnot efficiency. The optimality of the bound  when the working body of the engine is transformed via simple permutations, is proven in Appendix \ref{appendix4}. 
\end{proof}

Note that every permutation obtained from the one shown in Fig. \ref{fig:dswap} by permutating populations of catalyst energy levels $\{p_{1},p_{2},\dots,p_{d}\}$ is characterized by the same values of $Q_{h}$ and $Q_{c}$, and consequently, $\eta$ and $W$. Therefore, the state of the catalyst is determined up to a permutation.
In order to calculate the amount of work produced by the engine, we employ Eq. \eqref{Work_redefined}:
\begin{equation}
    W = Q_h+Q_c = (d\omega_h-\omega_c)\Delta \tilde{P}.
\end{equation}
To calculate $\Delta\tilde{P}$ one needs to apply the condition of preserving the marginal state of the catalyst (cyclicity), which we shall do in Section \ref{Work_extraction_section}.

We now turn our attention to the question: Can we always guarantee catalytic improvements in the efficiency of a two-stroke engine where both the hot and cold systems present in the working body are two-level systems? It is worth noting that the permutation in Eq. \eqref{Pi1} fails to facilitate work extraction whenever Eq. \eqref{relation_carnot} is not satisfied. Let us clarify this. Note that in order to extract a positive amount of work via simple permutation $\Pi$ given in Eq. \eqref{Pi1}, $\omega_c$, $\omega_h$, $\beta_c$, $\beta_h$, and $d$ have to obey Eq. \eqref{relation_carnot}. This sets the following constraint on the dimension of the catalyst 
\begin{equation}\label{range_of_d_eqn}
    \frac{\omega_c}{\omega_h}\leq d \leq \frac{\beta_c\omega_c}{\beta_h\omega_h}.
\end{equation}
Thus, catalytic enhancement in the efficiency is not always ensured via the permutation $\Pi$ given in Eq. \eqref{Pi1} if the inequality from Eq. \eqref{range_of_d_eqn} is not satisfied. (For instance, please see the example given in Sec. \ref{Extending_the_regime} where $\omega_c/\omega_h=1.5$ and $\beta_c\omega_c/\beta_h\omega_h=1.75$. Clearly, there will be no feasible values for  $d$ (i.e., dimension of catalyst) that satisfies Eq. \eqref{range_of_d_eqn} for  $\omega_c/\omega_h=1.5$ and $\beta_c\omega_c/\beta_h\omega_h=1.75$). 

Therefore to encompass such above-mentioned scenario, we would like to construct a a more general simple permutations which lead to catalytic enhancement in efficiency for any possible values of $\omega_c/\omega_h$ and $\beta_c\omega_c/\beta_h\omega_h$ as long as $\omega_c/\omega_h<\beta_c\omega_c/\beta_h\omega_h$.  In other words, we aim to construct a simple permutation that always enables us to exceed the optimal Otto efficiency of a two-stroke engine where both the hot and cold $d$-level systems present in the working body are two-level systems.  We describe the result in the following theorem.
\begin{thm}\label{theorem_imp}
Consider a two-stroke heat engine without a catalyst, where both the hot and cold systems in the working body are two-level systems. The energy of the excited state for the hot two-level system is greater than that for the cold two-level system, i.e., $\omega_h>\omega_c$ such that they satisfies the condition:
 
\begin{equation}
    \beta_h\omega_h<\beta_c\omega_c,
\end{equation}
where $\beta_h$ and $\beta_c$ denotes the inverse temperature associated with the hot and the cold two level systems. 

Then, it is always possible to incorporate a catalyst with the working body and construct a simple permutation that transforms the initial state of the working body during the work stroke, leading to an efficiency $\eta$  that is strictly greater than the optimal Otto efficiency of the non-catalytic two-stroke engine,  i.e., 
\begin{equation}\label{eff:drange2}
    1-\frac{\omega_c}{\omega_h}<\eta<1-\frac{\beta_h}{\beta_c}.
\end{equation}
\end{thm}
\begin{proof}
Consider the simple permutation $\Pi_{\text{gen}}$ that acts on the initial state of the working body (hot and cold two-level system, and the catalyst) given as follows:
\begin{eqnarray}\label{Pigen2}
\Pi_{\text{gen}}&=&\Big(\sum_{i=1}^{m}|i+1,0,0\rangle\langle i,1,0|+ \sum_{j=m+1}^{m+n-1}|j+1,0,1\rangle\langle j,1,0|\nonumber\\&+&|1,0,1\rangle\langle m+n,1,0|+\text{Herm. conjugate}\Big)_{s,h,c}+\mathbb{I}_{\text{Rest}}\nonumber\\\;,
\end{eqnarray}
where the dimension  of the catalyst denoted by $d$ is given by $m+n$ as shown in Fig. \ref{fig:mhotncold}. Here $\mathbb{I}_{\text{Rest}}$ is identity operator on the orthogonal complement of the first term in the parentheses. Let us assume that $\Delta P$ is the net population flow that happens from the $i^{\text{th}}$ vertex to $(i+1)^{\text{th}}$ vertex where  the index $i$ corresponds to the eigenvectors of the Hamiltonian of the catalyst and the `$+$' denotes the addition modulo $d$ (i.e., from a vertex to its neighbour to the right assuming periodic boundary conditions).

As mentioned earlier, due to preservability of the state of the catalyst, all the net transferred amount of population must be equal. Employing the formula from Eq. \eqref{Heat_jth_subspace} and Eq. \eqref{Cold_Heat_jth_subspace}, we can calculate 
\begin{equation}\label{hot_mn_cold_mn}
    Q_h = (m+n)\Delta P\omega_h=d\Delta P\omega_h\quad Q_c=-n\Delta P\omega_c,
\end{equation}
where negative sign in $Q_c$ implies population flow in the excited cold subspace. Thus, the efficiency can be calculated from Eq. \eqref{efficiency_def} as 
\begin{equation}\label{efficiency_mn}
    \eta = 1+\frac{Q_c}{Q_h} =1-\frac{n\omega_c}{(m+n)\omega_h}=1-\frac{n\omega_c}{d\omega_h}.
\end{equation}
In order to satisfy Eq. \eqref{eff:drange2}, it is enough to satisfy the following inequality.
\begin{eqnarray}\label{ineq:d_range2}
    1< \frac{d}{n} < \frac{\beta_c\omega_c}{\beta_h\omega_h}.
\end{eqnarray}
It is always possible to choose $d$ and $n$ such that the above inequality holds because $\frac{d}{n}$ is a rational number, thus it can  approximate any real number between $1$ and $\frac{\beta_c\omega_c}{\beta_h\omega_h}$ with arbitrary accuracy.
\end{proof} 

Here we would like to make a remark about the simple permutation $\Pi_{\text{gen}}$ defined in Eq. \eqref{Pigen2}. Note that, any possible efficiency that can be achieved by the two-stroke engine must lies in between $0$ and Carnot efficiency as mentioned in Eq. \eqref{Efficiency_range_for_positive_work}. Now any efficiency between $0$ and Carnot efficiency can be realized by choosing $n$ and $d$ in the simple permutation $\Pi_{\text{gen}}$. So, the inclusion of a catalyst of suitably chosen dimension enables to realize any possible value of efficiency of the heat engine, i.e., any $\eta$ that satisfy Eq. \eqref{Efficiency_range_for_positive_work}. 

Thus far, we have demonstrated the catalytic enhancement in efficiency for the  two-stroke engine, with a working body composed of a hot and a cold $d$-level system in two dimensions. We will now address the problem of catalytic enhancement in the efficiency for a generic two stroke engine having arbitrary finite dimensional hot and the cold $d$-level system in the working body.

\subsubsection{The optimal efficiency - analysis for  working body of an arbitrary dimension}\label{Improving_the_eff_d}

In this section, we would like to investigate whether the presence of a catalyst in the working body of the engine, along with arbitrary dimensional hot and cold $d$-level system , always leads to the enhancement of efficiency   with non-zero finite work production. We would like to make an important remark here. Employing the result from Ref. \cite{WilmingPRL} one can show that there will always exist a catalyst and unitary that enables the engine to attain the Carnot efficiency. But attaining Carnot efficiency by a two-stroke engine leads to zero work production per cycle. So, in this section, we aim to explore scenarios where the efficiency of a two-stroke engine, producing non-zero work, can be enhanced by incorporating a catalyst into the working body.  

Addressing this question presents a significant challenge, as to prove catalytic enhancement, knowledge about the optimal efficiency of the corresponding non-catalytic two-stroke engine is needed. Obtaining the optimal efficiency of a non-catalytic two-stroke engines in such scenarios are extremely difficult, as the size of the feasible set and number of relevant parameters grows factorially with dimensions (e.g. numbers of the different energy gaps in the Hamiltonian). Nevertheless, we characterize instances where the catalytic enhancements can be ensured.

Consider a non-catalytic two-stroke heat engine with the working body composed of arbitrary dimensional hot and cold $d$-level system. The working body starts in the initial state $\tau_h\otimes\tau_c$ and transforms to $\rho^{*}_{h,c}$ that leads to the maximum efficiency $\eta^{*}$. There can be two possibilities about the final state of the working body $\rho^{*}_{h,c}$:
\begin{enumerate}
    \item $\rho^{*}_{h,c}$ is a correlated state.
    \item $\rho^{*}_{h,c}$ is a product state.
\end{enumerate}

Let us start with the case when the final state of the working body of the engine $\rho^{*}_{h, c}$ is correlated. We employ a result from theorem 1 of Ref. \cite{HenaoUzdin} that proves the existence of a catalyst in the state $\rho_s$ and a unitary $U$ such that 
    \begin{eqnarray}
&&|\Tr_{s,h}\big(U(\rho_s\otimes\rho^{*}_{h,c})U^{\dagger}\big)\dket \succ |\Tr_{h}(\rho^{*}_{h,c})\dket,\label{naj_inp}\\
       &&\Tr_{s,c}\big(U(\rho_s\otimes\rho^{*}_{h,c})U^{\dagger}\big) = \Tr_{c}(\rho^{*}_{h,c}), \quad\text{and}\quad\label{naj_inp2}\\
       && \Tr_{h,c}\big(U(\rho_s\otimes\rho^{*}_{h,c})U^{\dagger}\big) = \rho_s\label{cat_last},
    \end{eqnarray}
if and only if $\rho^{*}_{h,c}$ is correlated. Here $|(\cdot)\dket$ denotes the spectrum of $(\cdot)$ and the preservability of the marginal state of the catalyst is guranteed by Eq. \eqref{cat_last}. One can observe from Eq. \eqref{naj_inp2} that the final state of hot $d$-level system in the two-stroke heat engine in the catalytic scenario is same as the final state of hot $d$-level system in the non-catalytic scenario. Therefore, the transferred amount of the heat $Q_h$ is same in both cases. Using the definition of heat transferred and cold heat from Eq. \eqref{defn:hot_heat} and Eq. \eqref{defn:cold_heat}, we can write optimal efficiency as
\begin{eqnarray}\label{Fin:eff}
    \eta^{*} =  1+\frac{Q_c}{Q_h} = 1-\frac{\Tr(H_c\rho^{*}_c)-\Tr(H_c\tau_c)}{\Tr(H_h\tau_h)-Tr(H_h\rho^{*}_h)},
\end{eqnarray}
where $\rho^{*}_h = \Tr_{h}(\rho^{*}_{h,c})$  and $\rho^{*}_c = \Tr_{c}(\rho^{*}_{h,c})$. Now, to attain the optimal efficiency the final marginal state of the cold $d$-level system i.e., $\rho^{*}_c$ has to be passive  (For definition of passivity see Eq. \eqref{defn_passive_state}).  If it is non-passive then after implementation of a unitary that transforms $\tau_h\otimes\tau_c$ to $\rho^{*}_{h,c}$, one can apply a local unitary $\mathbb{I}_h\otimes V_c$ on the cold $d$-level system $\rho^{*}_c$ which make the cold $d$-level system passive without altering the energy of the hot $d$-level system. This will increase the efficiency further as can be seen from Eq. \eqref{Fin:eff}. Using the same reasoning we see that in the catalytic scenario, the final state of the hot $d$-level system i.e., $\rho^{*}_h$ also has to be passive in order to obtain the optimal efficiency. Therefore, the final average energy of the cold $d$-level system equals to the passive energy   where we have defined passive energy in Eq. \eqref{defn_passive_state} as follows :
\begin{equation}\label{passivity}
    \min_{V\in \mathbb{U}(d)}\Tr(HV\rho V^{\dagger}) =\Tr(H\rho_\text{pass}).
\end{equation}
As the passive energy of a state is a Schur-concave function (which can be seen from Lemma 1 in Appendix of Ref. \cite{HenaoUzdin}), Eq. \eqref{naj_inp} implies that the final average energy of cold $d$-level system in the catalytic scenario is lower than the final average energy of cold $d$-level system in the non-catalytic scenario. Thus, final average energy of the working body of the engine in the catalytic scenario can be made lower than final average energy of the working body of the engine in the non-catalytic scenario which implies a greater extraction of work by withdrawing same heat from the hot heat bath. Therefore, from Eq. \eqref{Fin:eff} we conclude that efficiency of the two-stroke engine assisted with a catalyst can be made strictly greater than the efficiency of the two-stroke engine in the absence of a catalyst.

 Next, we move on to the case of $\rho^{*}_{h,c}$ being a product state that leads to the optimal efficiency $\eta^{*}$ in the non-catalytic scenario.   Here, we would like to employ the results from Ref. \cite{Sparaciari2017}, where in Eq. $(8)$ a permutation $\tilde{\Pi}$ is constructed such that for any passive state $\rho_{\text{pass}}$ that is not equal to the Gibbs state of the system, the following holds:
\begin{eqnarray}\label{sparaciari_perm}
    \Tr(H\rho_{\text{pass}}) &>& \Tr\big(H(\Tr_s\tilde{\Pi}(\rho_s\otimes\rho_{\text{pass}})\tilde{\Pi}^{\dagger})\big)\\
    \rho_s &=&\Tr_s\big(\tilde{\Pi}(\rho_s\otimes\rho_{\text{pass}})\tilde{\Pi}^{\dagger}\big),
\end{eqnarray}
 where $\rho_s$ is the state of the catalyst and $H$ is the Hamiltonian of the system in the state $\rho_{\text{pass}}$.
 In other words, $\tilde{\Pi}$ always allows to reduce the average energy of a passive state which is not Gibbs state with the aid of a catalyst.   By inspecting Eq. \eqref{Fin:eff}, we see that to achieve maximal efficiency the final state of the hot and cold $d$-level system has to be passive (one can use a similar argument as earlier to conclude the hot $d$-level system  also has to be in a passive state to obtain the optimal efficiency). Now, if the final state of either the hot or cold $d$- level system $\rho^{*}_{h,c}$ is not thermal i.e.,
 \begin{equation}
   \Tr_{h}(\rho^{*}_{h,c}) \neq \frac{e^{-\tilde \beta_hH_h}}{\Tr(e^{-\tilde \beta_h H_h})} \;\; \text{or} \;\; \Tr_{c}(\rho^{*}_{h,c}) \neq \frac{e^{-\tilde \beta_cH_c}}{\Tr(e^{-\tilde \beta_c H_c})},
 \end{equation}
 where $\tilde \beta_h$ and $\tilde \beta_c$ are some inverse temperatures, then we can apply the permutation $\tilde{\Pi}$ given in Eq. \eqref{sparaciari_perm} on that particular non-thermal passive state and a catalyst. If both  the hot or cold $d$-level system are non-thermal passive state, we can can apply the permutation $\tilde{\Pi}$ on either of them.
 
 Let us assume $ \Tr_{h}(\rho^{*}_{h,c})$ is passive but not thermal, then employing Eq. \eqref{sparaciari_perm} we say:
 \begin{eqnarray}\label{sparaciari_perm2}
    \Tr(H_h(\Tr_{h}(\rho^{*}_{h,c}))) &>& \Tr\big(H_h(\Tr_s\tilde{\Pi}(\rho_s\otimes\Tr_{h}(\rho^{*}_{h,c}))\tilde{\Pi}^{\dagger})\big)\nonumber\\
    \text{such that}\;\;\rho_s &=&\Tr_s\big(\tilde{\Pi}(\rho_s\otimes\Tr_{h}(\rho^{*}_{h,c}))\tilde{\Pi}^{\dagger}\big),
\end{eqnarray}
Therefore heat withdrawn from the hot bath in the catalytic scenario is strictly greater than heat withdrawn from the hot bath in the non-catalytic scenario as:
\begin{eqnarray}
    &&\Tr(H_h\tau_h)- \Tr\big(H_h(\Tr_s\tilde{\Pi}(\rho_s\otimes\Tr_{h}(\rho^{*}_{h,c}))\tilde{\Pi}^{\dagger})\big)\nonumber\\&>& \Tr(H_h\tau_h)- \Tr(H_h(\Tr_{h}(\rho^{*}_{h,c}))).
\end{eqnarray}
Thus, the resulting efficiency in the catalytic scenario is strictly greater than non-catalytic scenario as can be seen by employing Eq. \eqref{Fin:eff}: 
\begin{eqnarray}
    &&1-\frac{\Tr(H_c\left(\Tr_c(\rho^{*}_{h,c})\right))-\Tr(H_c\tau_c)}{\Tr(H_h\tau_h)- \Tr\big(H_h(\Tr_s\tilde{\Pi}(\rho_s\otimes\Tr_{h}(\rho^{*}_{h,c}))\tilde{\Pi}^{\dagger})\big)} \nonumber\\
    &>& 1-\frac{\Tr(H_c\left(\Tr_c(\rho^{*}_{h,c})\right))-\Tr(H_c\tau_c)}{\Tr(H_h\tau_h)- \Tr(H_h(\Tr_{h}(\rho^{*}_{h,c})))}.
\end{eqnarray}
A similar argument regarding catalytic enhancement of efficiency can be given when $\Tr_c(\rho^{*}_{h,c})$ is passive but not thermal.

 Let us make an interesting remark here. Consider $\rho^{*}_{h,c}$ that leads to the optimal efficiency for the two-stroke thermal machine in the non-catalytic scenario. Assume $\rho^{*}_{h,c}$ is either correlated or product such that the marginal state of the hot or cold $d$-level system is not thermal. Now, from the definition of work given in Eq. \eqref{Work_redefined}, we can write:
 \begin{eqnarray}
     W &=& Q_h+Q_c \nonumber\\&=&\Tr\big(H_h(\tau_h-\rho^{*}_h)\big)+\Tr\big(H_c(\tau_c-\rho^{*}_c))\nonumber\\
    &=& \Tr\big(H_h\tau_h+H_c\tau_c\big)-  \Tr\big(H_h\rho^{*}_h+H_c\rho^{*}_c\big)\label{Work_redfn_1},
 \end{eqnarray}
 where, 
 \begin{equation}\label{caWork}
     \rho^{*}_h = \Tr_{c}\rho^{*}_{h,c}\quad\text{and}\quad \rho^{*}_c = \Tr_{h}\rho^{*}_{h,c}.
 \end{equation}
We have argued that the average energy of \(\rho^{*}_{c}\) can be further reduced using a catalyst when \(\rho^{*}_{h,c}\) is correlated. Additionally, we have shown that if \(\rho^{*}_{h,c}\) is a product state where either the hot or cold \(d\)-level system is non-thermal, its average energy can also be further reduced with a catalyst. From Eq. \eqref{Work_redfn_1}, it is evident that this reduction can lead to simultaneous improvements in both work production and efficiency.

The only case which is left to be answered is when $\rho^{*}_{h,c}$ is a product of two Gibbs state. Note that, when the dimension of the hot and cold $d$-level system is equal to $2$, we have seen that 
$\rho^{*}_{h,c}$ became equals to $(\tau_c\otimes\tau_h)_{h,c}$ (One can easily check this by applying the optimal permutation given in Eq. \eqref{Pie} on the initial state $(\tau_h\otimes\tau_c)_{h,c}$). Now, the state $\tau_c$ can be thought as the Gibbs state of the hot two-level system at inverse temperature $\tilde{\beta}_h = \frac{\beta_c\omega_c}{\omega_h}$ with Hamiltonian $\omega_h\ketbra{1}{1}$. Similarly, the state $\tau_h$ can be thought as the Gibbs state of the cold two-level system at inverse temperature $\tilde{\beta}_c = \frac{\beta_h\omega_h}{\omega_c}$ with Hamiltonian $\omega_c\ketbra{1}{1}$. As we have shown, in that case we can always obtain catalytic enhancements in the efficiency. We believe that in the case when $\rho^{*}_{h,c}$ is a product of two Gibbs states of arbitrary dimension,  to obtain the catalytic enhancements in the efficiency one needs to figure out what is the optimal efficiency in the non-catalytic scenario first. As mentioned earlier, this is a challenging task due to the complex structure of the problem as it involves many parameters as well as size of the feasible set grows factorially. 

We summarize the results on the catalytic efficiency enhancement in engines in the following table \ref{tab:somelabel}: 

\begin{table}[htbp]
\centering
\begin{tabular}{|c|c|}
\hline
\textbf{Final state $\rho^{*}_{h,c}$ leading} & \textbf{Catalytic} \\
\textbf{ to optimal efficiency} & \textbf{Enhancement} \\
\textbf{in the non-catalytic scenario} & \textbf{} \\
\hline
Correlated & Yes \\
\hline
\hline
Product with either of the & Yes \\
$d$-level system is passive&  \\
but not Gibbs &  \\
\hline
\hline
Product with both of the & Yes \\
$d$-level system is Gibbs&  \\
when $d=2$ &  \\
\hline
\hline
Product with both of the & ? \\
$d$-level system is Gibbs&  \\
with $(d>2)$ &  \\
\hline
\hline
\end{tabular}
\caption{\label{tab:somelabel}}
\end{table}

We will explore now the another kind of catalytic enhancement which is about broadening of the operational regime of the two-stroke heat engine.

\subsection{Extending the regime of operations for the two-stroke thermal machine}\label{Extending_the_regime}

In this section, we shall explore the second kind of catalytic enhancement. We shall demonstrate that incorporating a catalyst with the working body of the engine enables the two-stroke heat engine to produce work in a regime of frequency and temperature where it is impossible for a non-catalytic two-stroke engine.   Let us proceed by defining the regime of operation of the engine. We shall restrict ourselves to microscopic two-stroke heat engine where the dimension of the hot and the cold $d$-level system is two. 
\begin{defn}[Regime of operation of the engine]
    Consider a catalyst assisted two-stroke engine with working body containing the hot and cold $d$-level system of dimension two. The regime of operation is characterized by the following four quantities:
    \begin{equation}
        (\omega_c,\omega_h,\beta_c,\beta_h),
    \end{equation}
    where $\omega_c$ and $\omega_h$ are the energy eigenvalue associated with the excited state of cold and hot $d$-level system Hamiltonian (assuming the ground state energy for both the Hamiltonian is zero), and $\beta_c$ and $\beta_h$ are the inverse temperature associated with the cold and the hot $d$-level system respectively. Here, $\omega_c$ and $\omega_h$ characterize the frequency regime, whereas $\beta_c$ and $\beta_h$ characterize the temperature regime.
\end{defn}
Note that there exists certain values for $(\omega_c,\omega_h,\beta_c,\beta_h)$ where the non-catalytic two-stroke engine fails to extract work. For instance, in the non-catalytic scenario when $\omega_c>\omega_h$, the engine fails to produce work because the initial state of the working body is passive (for the definition of passive state see Eq. \eqref{defn_passive_state}).
We shall show that incorporating a catalyst with the working body of the engine always enables to produce work at any possible value of $\omega_h$, and $\omega_c$ even if $\omega_c>\omega_h$.  Let us illustrate the idea via a numerical example for a two-stroke engine where the dimension of the hot and cold $d$-level system present in the working body is restricted to $2$.

\emph{Example: }Consider a two-stroke engine starts with a working body where the  cold two-level system  is thermalized at inverse temperature $\beta_c = 7$ and the  hot two-level system is thermalized at inverse temperature $\beta_h = 6$. We consider the Hamiltonian of the hot and cold two-level system as:
\begin{eqnarray}
    H_h &=& \omega_h|1\rangle\langle 1|,\quad H_c = \omega_c|1\rangle\langle 1| \nonumber\\&&\quad\text{with}\;\; \omega_h = 2,\;\omega_c =3.
\end{eqnarray}
In this regime of temperatures and frequencies, the two-stroke engine can not produce work since
\begin{equation}
    \frac{\omega_c}{\omega_h}=\frac{3}{2}>1,
\end{equation}
which makes the optimal efficiency (Otto efficiency) given in Eq. \eqref{Otto} negative. One can check this also by computing the initial state which turns out to be passive:
\begin{eqnarray}
    \tau_h\otimes\tau_c &=& \frac{1}{(1+a_h)(1+a_c)}(1,\;a_h,\;a_c,\;a_ca_h)\\
    &=& \frac{1}{(1+e^{-12})(1+e^{-21})}(1,\;e^{-12},\;e^{-21},\;e^{-33}),\nonumber
\end{eqnarray}
where $a_h:=e^{-\beta_h\omega_h}=e^{-12}$ and $a_c:=e^{-\beta_c\omega_c}=e^{-21}$. The passivity is reflected from the fact that $\omega_c>\omega_h$ whereas $a_c<a_h$. Therefore, one can not extract work from such an initial state via unitary operation. Now, we can make this engine to produce positive amount of work with the aid of a catalyst. We have 
\begin{equation}
    \frac{\omega_c}{\omega_h}=\frac{3}{2}<\frac{\beta_c\omega_c}{\beta_h\omega_h}=\frac{7}{4},
\end{equation}
thus, in order to extract positive amount of work via the simple permutation given in Eq. \eqref{Pigen2}, we need to choose the dimension of the catalyst $d$ and the parameter $n$ such that Eq. \eqref{Efficiency_range_for_positive_work} holds. We choose $d=5$ and $n=3$, then clearly
\begin{equation}
    \frac{\omega_c}{\omega_h}=\frac{3}{2}=1.5<\frac{d}{n}=\frac{5}{3}\simeq 1.67 <\frac{\beta_c\omega_c}{\beta_h\omega_h}=\frac{7}{4}=1.75.
\end{equation}
Employing the formula of efficiency from Eq. \eqref{efficiency_mn}, we can easily calculate the efficiency 
\begin{equation}
    \eta = 1-\frac{3\omega_c}{5\omega_h}=1-\frac{9}{10} = 0.1.
\end{equation}
As the efficiency $\eta$ is between $0$ and Carnot efficiency i.e $\eta$ satisfies Eq. \eqref{Efficiency_range_for_positive_work}, this ensures that engine will produce positive work. One can calculate produced amount of work by the engine via calculating $Q_h$ and $Q_c$ from Eq. \eqref{Heat_jth_subspace} and Eq. \eqref{Cold_Heat_jth_subspace}, respectively.

This example exhibits two very crucial features about the catalyst assisted two-stroke engine operating according to the simple permutation $\Pi_{\text{gen}}$.  First, the catalyst enables the engine to produce work in a regime of frequency and temperature, where it is impossible to work for a non-catalytic two-stroke engine. From the second law, we know that in order to extract work the efficiency $\eta$ should satisfy the inequality in Eq. \eqref{Carnot_bd} i.e., 
\begin{equation}
    0<\eta<1-\frac{\beta_h}{\beta_c}.
\end{equation}
Second, if we apply the simple permutation given in Eq. \eqref{Pi1} to extract work, it will not be possible because $n=1$ for this permutation, which implies
\begin{equation}
    \frac{d}{n} = d.
\end{equation}
Therefore, for any value of $d$ the Eq. \eqref{ineq:d_range2} can not be satisfied. In that case we need to suitably choose $d$ and $n$ in a way in the simple permutation Eq. \eqref{Pigen2} such that Eq. \eqref{ineq:d_range2} is satisfied.

In summary, whenever the initial state of the working body (i.e., catalyst, the hot and the cold $d$-level system) transforms via simple permutation given in Eq. \eqref{Pigen2} during the work stroke the resulting efficiency of the engine reduces to (as given in Eq. \eqref{efficiency_mn})
\begin{equation}
    1-\frac{n\omega_c}{d\omega_h},
\end{equation}
where $d$ is the dimension of the catalyst and $n$ is a parameter in the permutation $\Pi_{\text{gen}}$ in Eq. \eqref{Pigen2} (Recall that, $n$ is the number of swap represented by grey colored arrow in Fig. \ref{fig:mhotncold}). Now, one can choose $n$ and $d$ in the above expression of efficiency , such that the following inequality holds: 
\begin{equation}\label{ndrange}
    0<1-\frac{n\omega_c}{d\omega_h}<1-\frac{\beta_h}{\beta_c},
\end{equation}
which ensures that work produced by the engine will be positive. Importantly, inequality given in Eq. \eqref{ndrange} can be satisfied even if $\omega_c>\omega_h$ or $\omega_c=\omega_h$ upon suitably choosing $n$ and $d$.

Therefore, it implies that incorporating a catalyst with the working body can enable the two-stroke engine to function for any regime of frequency and temperature. We summarize the discussion in the following theorem:
\begin{thm}\label{theorem_imp2}
    Consider a catalyst assisted two stroke heat engine with working body containing the hot and cold $d$-level system of dimension two. It can produce work at any regime of frequency (i.e., for any value of $\omega_c$ and $\omega_h$) and temperature (i.e., for any value of $\beta_c$ and $\beta_h$) for suitably chosen dimension of the catalyst and the value of $n$ in the permutation $\Pi_{\text{gen}}$ defined in Eq. \eqref{Pigen2}.
\end{thm}
In Fig. \ref{fig:Enhancement_Range} we have schematically represent the result of theorem \ref{theorem_imp} and theorem \ref{theorem_imp2}. In Fig. \ref{fig:Regime_maximum} we plot the regime of operation for the two-stroke engine defined for two-dimensional hot and cold systems. We see that already catalysts of small dimensions, e.g, $2,\;3,\;4$, significantly extend the regime of the operation for the two-stroke engine. We see that a catalyst enables the engine to extract work in an arbitrary regime of temperatures, as $\frac{\omega_h}{\omega_c}>\frac{\beta_h\omega_h}{\beta_c\omega_c}$ is always satisfied.
\begin{figure*}[t]
    \centering
    \includegraphics[width=14 cm]{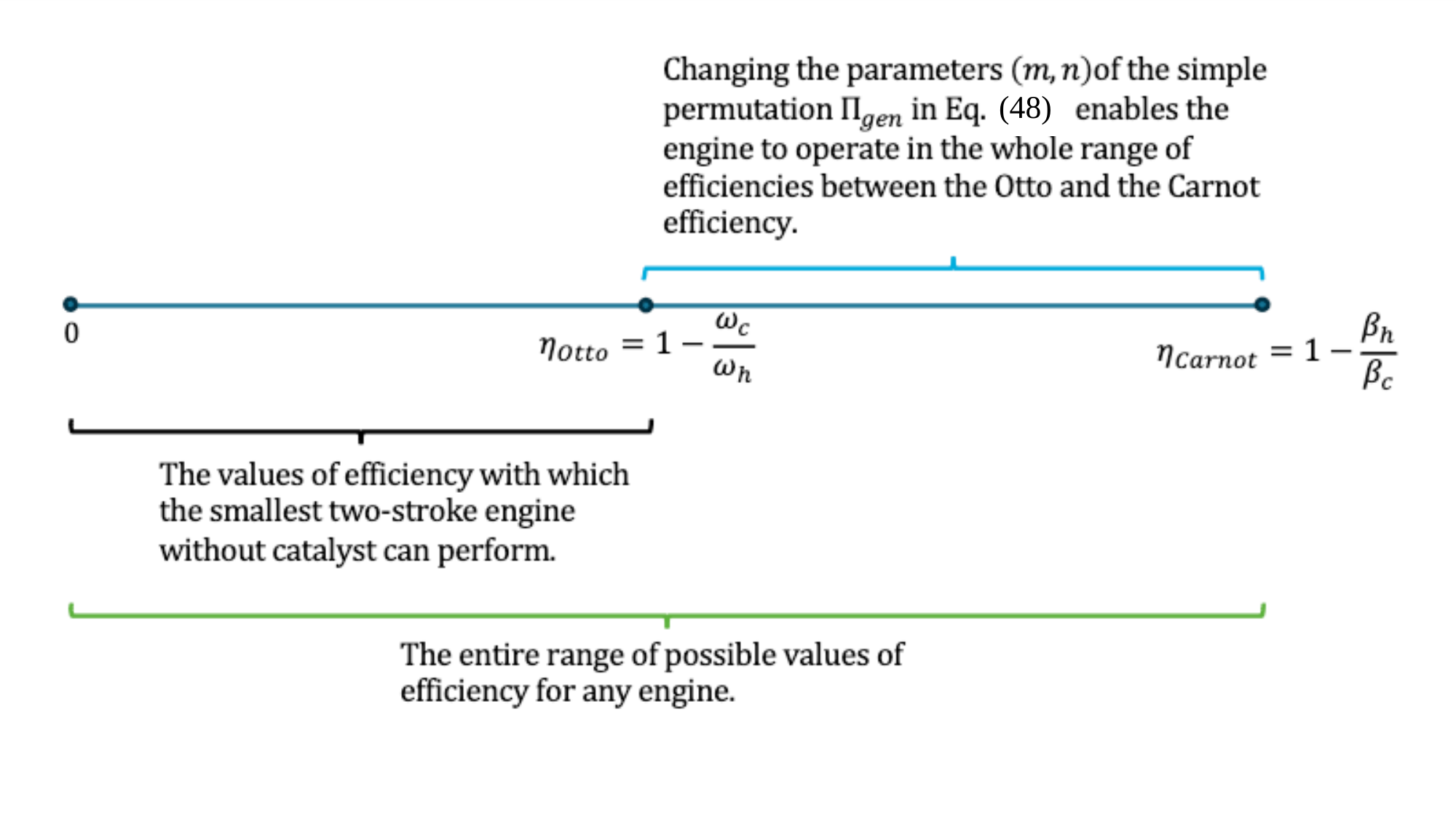}
    \caption{This figure illustrates results of Theorem \ref{theorem_imp} and Theorem \ref{theorem_imp2}. The efficiency $\eta$ for any two-stroke engine has to be in between $0$ and Carnot efficiency (as represented within the green curly brace). For a non-catalytic two-stroke engine where the dimension of the both hot and cold $d$-level system present in the working body is $2$ such that $\omega_h>\omega_c$, the value of the efficiency should be in between $0$ and the Otto efficiency as proved in Sec. \ref{qubit_otto} (as represented in the black curly brace). With suitably choosing the dimension of the catalyst $d$ and $n$ in the simple permutation $\Pi_{\text{gen}}$ in Eq. \eqref{Pigen2}, one can operate the catalyst assisted two-stroke engine with any efficiency between the Otto efficiency and the Carnot efficiency as described in Theorem \ref{theorem_imp} (as represented in the blue curly brace). On the other hand, for any frequency and temperature ratio $\omega_c/\omega_h$ and $\beta_h/\beta_c$, one can choose $d$ and $n$ in the simple permutation $\Pi_{\text{gen}}$ in Eq. \eqref{Pigen2} such that corresponding efficiency is between $0$ and the Carnot efficiency. Therefore, catalyst allows the two-stroke engine to operate with any physically plausible values of efficiency.} 
    \label{fig:Enhancement_Range}
\end{figure*}
\begin{figure*}[t]
    \centering
    \includegraphics[width= 14 cm]{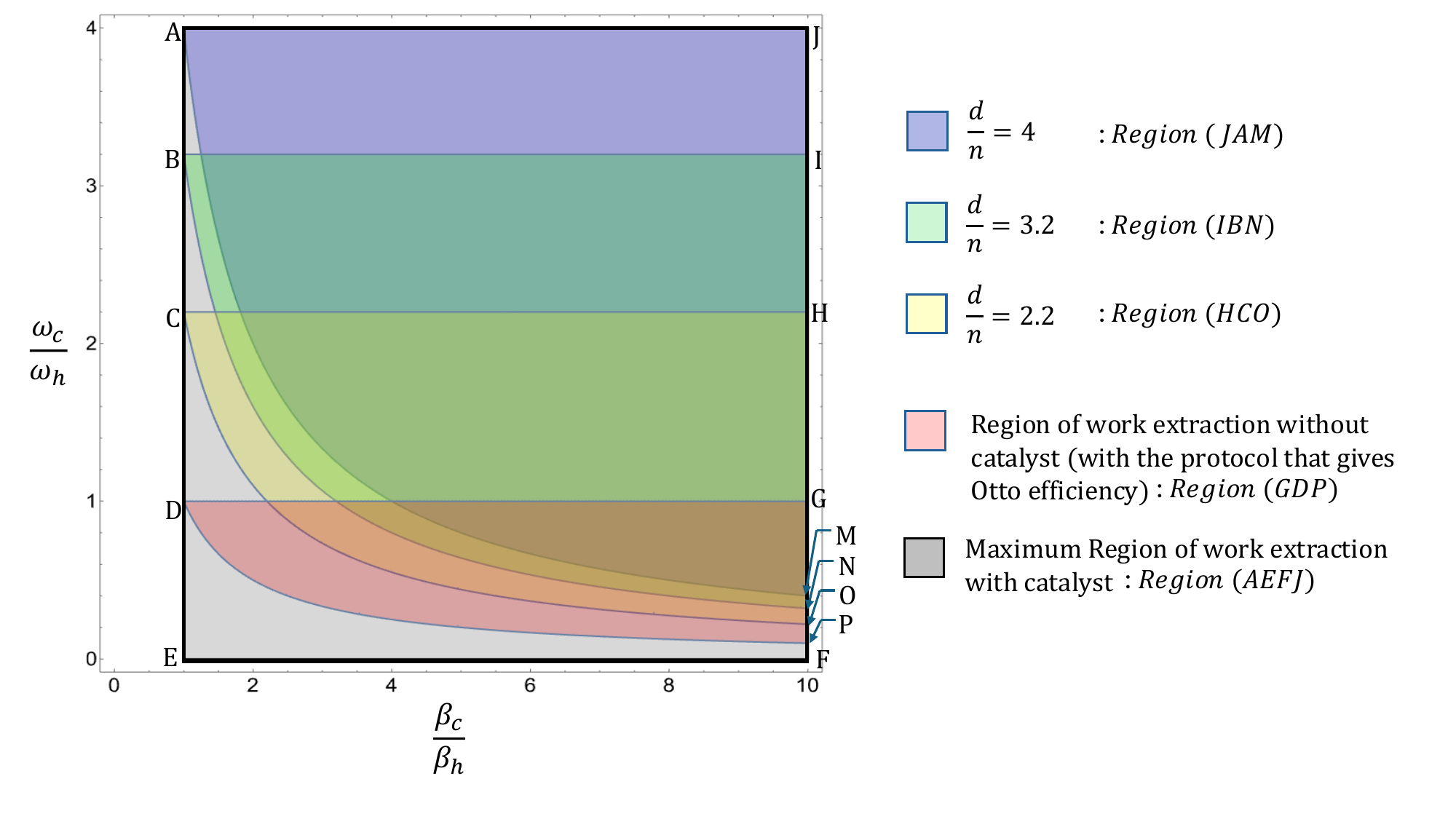}
    \caption{\label{fig:Regime_maximum} This figure depicts the different region where the engine can produce positive amount of work in presence and absence of catalyst in the working body. The region is parameterized by frequency ratio  of the hot and the cold $d$-level system i.e., $(\frac{\omega_c}{\omega_h})$ in the $y$-axis and the temperature ratio  of the hot and the cold $d$-level system i.e., $(\frac{\beta_c}{\beta_h})$ in the $x$-axis. The region characterized by efficiency of the engine $\eta$ satisfy the inequality $0<\eta<1-\frac{\beta_h}{\beta_c}$ shown in solid enclosed square (Region AEFJ). This is the maximum region where any engine can produce work. The region in light red (Region GDP) depicts the frequency and temperature regime where a two-stroke engine can produce work with the Otto efficiency $(i.e., 1-\frac{\omega_c}{\omega_h})$ in the non-catalytic scenario. This region is characterized by $\frac{\omega_c}{\omega_h}<1$ and $1-(\frac{\omega_c}{\omega_h})<1-\frac{\beta_c}{\beta_h}$. As a result region in light red is characterized by the condition: $\frac{\omega_c}{\omega_h}<1<\frac{\beta_c\omega_c}{\beta_h\omega_h}$.    If the working body of the engine transforms by simple permutation $\Pi_{\text{gen}}$ given in Eq. \eqref{Pigen2} during work stroke, one can obtain efficiency $\eta = 1-\frac{n\omega_c}{(m+n)\omega_h}= 1-\frac{n\omega_c}{d\omega_h}$ where $d = m+n$ is the dimension of the catalyst. By suitably choosing $m$ and $n$ in the simple permutation given in Eq. \eqref{Pigen2}, engine with any efficiency $\eta$ that satisfy the inequality $0<\eta<1-\frac{\beta_h}{\beta_c}$ can be constructed. Thus, absolute maximum region where any engine can produce work (characterized by the inequality $0<\eta<1-\frac{\beta_h}{\beta_c}$) is same with the maximum region where a two-stroke engine with working body containing a catalyst of suitably chosen dimension can produce work.  The regime for work extraction via simple permutation $\Pi_{\text{gen}}$ given in Eq. \eqref{Pigen2} for different values of $\frac{d}{n}$ is plotted for  $\frac{d}{n}=2.2$, $\frac{d}{n}=3.2$, and $\frac{d}{n}=4$ has been shown in this figure in yellow (Region HCO), green  (Region IBN), blue shaded region  (Region JAM). } 
\end{figure*}

 So far we have examined how introducing a catalyst into the working body can improve the efficiency and broaden the operational range of a two-stroke heat engine. However, we have not yet calculated the work produced by the engine. Next, we will calculate the work produced by a two-stroke heat engine with a working body consisting of two two-level systems, each thermalized at different temperatures, and a $d$-level catalyst. Our focus will be on a specific case where the working body undergoes a transformation through the simple permutation $\Pi_{\text{gen}}$ as defined in Eq. \eqref{Pigen2}.

\section{Work extraction}\label{Work_extraction_section}
In this section, we shall calculate the amount of work that can be extracted when the initial state of the working body is consists of catalyst of dimension $d$, and the hot and the cold $d$-level system is of dimension $2$, transformed via simple permutation $\Pi_{\text{gen}}$ given in Eq. \eqref{Pigen2} during the work stroke. Recall that
\begin{eqnarray}\label{Pigen}
\Pi_{\text{gen}}&=&\big(\sum_{i=1}^{m}|i+1,0,0\rangle\langle i,1,0|\nonumber\\&+& \sum_{j=m+1}^{m+n-1}|j+1,0,1\rangle\langle j,1,0|\\&+&|1,0,1\rangle\langle m+n,1,0|+\text{Herm. conjugate}\big)_{s,h,c}+\mathbb{I}_{\text{Rest}}\nonumber\;,
\end{eqnarray}
which is depicted in Fig. \ref{fig:mhotncold} where dimension of the catalyst 
\begin{equation}
    d=m+n,
\end{equation}
and, $\mathbb{I}_{\text{Rest}}$ is identity operator on the orthogonal complement of the first term in the parentheses.

The definition of work is given in Eq. \eqref{Work_redefined} $W=Q_h+Q_c$, where $Q_h$ is the heat consumed from the hot bath and $Q_c$ is the heat dumped into the cold bath. As we assume that the ground state energy of both the qubits is zero, $Q_h$ and $Q_c$ are simply given by the difference between the initial and final population in the excited hot and cold subspaces (see Eq. \eqref{Heat_jth_subspace} and Eq. \eqref{Cold_Heat_jth_subspace}). Following the result from Eq. \eqref{hot_mn_cold_mn}, we calculate $Q_h$ and $Q_c$ for the simple permutation $\Pi_{\text{gen}}$ as 
\begin{equation}\label{QcQh1}
    Q_h = (m+n)\omega_h\Delta P =d\omega_h\Delta P\quad,\quad Q_c = -n\omega_c\Delta P.
\end{equation}
Therefore, the amount of produced work by the engine is given by
\begin{equation}\label{eq:work_renewed1}
    W= \big((m+n)\omega_h-n\omega_c\big)\Delta P.
\end{equation}
In order to calculate the amount of work produced by the engine, we need to calculate $\Delta P$ in terms of $a_h$, $a_c$, $m$, and $n$. In the appendix \ref{Work_prod}, we provide the detailed calculation of the $\Delta P$ resulting in the following equality
\begin{equation}
    \Delta P = \frac{1}{f(a_h,a_c,m,n)}\N (a_h^{m+n}-a_c^n),
\end{equation}
where
\begin{equation}
    \N = \frac{1}{(1+a_h)(1+a_c)},
\end{equation}
and $f$ is a function of $a_h,\;a_c,\;m,\;n$ given as follows:
 
\begin{widetext}
\begin{equation}
    f(a_h,a_c,m,n):=\frac{a_h\left(1-a_c\right)^2\left\{\left(1-a_h^m\right)\left(a_h^n-a_c^n\right)\right\}+\left\{\left(a_h^{(m+n)}-a_c^n\right)\left(a_h-a_c\right)\left(1-a_h\right)\right\}\left\{n(1-a_h)-m(a_h-a_c)\right\}}{\left(a_h-a_c\right)^2\left(1-a_h\right)^2}.
\end{equation}
\end{widetext}

Therefore, the amount of the work produced by the engine is given by 
\begin{equation}\label{Eq:Produced_work}
    W = \frac{ \big((m+n)\omega_h-n\omega_c\big)}{f(a_h,a_c,m,n)}\N(a_h^{m+n}-a_c^n).
\end{equation}

In Fig. \ref{fig:with_without}, we plotted  the $W$ as a function of $n$ for $d=m+n=30$ for different values of $a_h$ and $a_c$,  observing that for certain values of $n$, amount of produced work by the two-stroke engine with a catalyst surpasses the work produced by the two-stroke engine without a catalyst. It also suggests that it is possible to achieve the simultaneous enhancements in the work extraction and efficiency by suitably choosing the values of $n$ and dimension of the catalyst $d$.

Let us now investigate two special instances of the work extraction by the simple permutation given in Eq. \eqref{Pigen}. Note that simple permutation $\Pi$ given in Eq. \eqref{Pi1} can be obtained by substituting $m=d-1$ and $n=1$ in the permutation $\Pi_{\text{gen}}$ given in Eq. \eqref{Pigen}. Thus, one can easily obtain the produced amount of work when the working body of the engine transformed via simple permutation $\Pi$ by substituting $m=d-1$ and $n=1$ in Eq. \eqref{Eq:Produced_work}. In this case, the work produced by the engine is given  
\begin{eqnarray}
    W_d := \frac{1}{f_d}(d \omega_h - \omega_c)\N(a_h^d-a_c),
\end{eqnarray}
where
\begin{equation}
    f_d := \frac{(1-a_h^d) (1-a_c)}{(1-a_h)^2} - \frac{d (a_h^d - a_c) }{(1-a_h)}.
\end{equation}

Next, let us consider the catalyst assisted two-stroke engine, with each component of the working body modelled by two-level systems (i.e., a hot two level system, a cold two level system, and a two-level catalyst). This forces us to set $m=1$ and $n=1$. The work produced by the engine boils down to
\begin{eqnarray}\label{W2}
  W_2:=\frac{(2\omega_h - \omega_c)}{1+a_c+2a_h}\N(a_h^2-a_c),
\end{eqnarray}
whereas the efficiency can be calculated from Eq. \eqref{efficiency_def} as
\begin{equation}\label{eta2}
    \eta_2:=1-\frac{\omega_c}{2\omega_h}.
\end{equation}

\begin{figure*}[t]
    \centering
    \includegraphics[width=14cm]{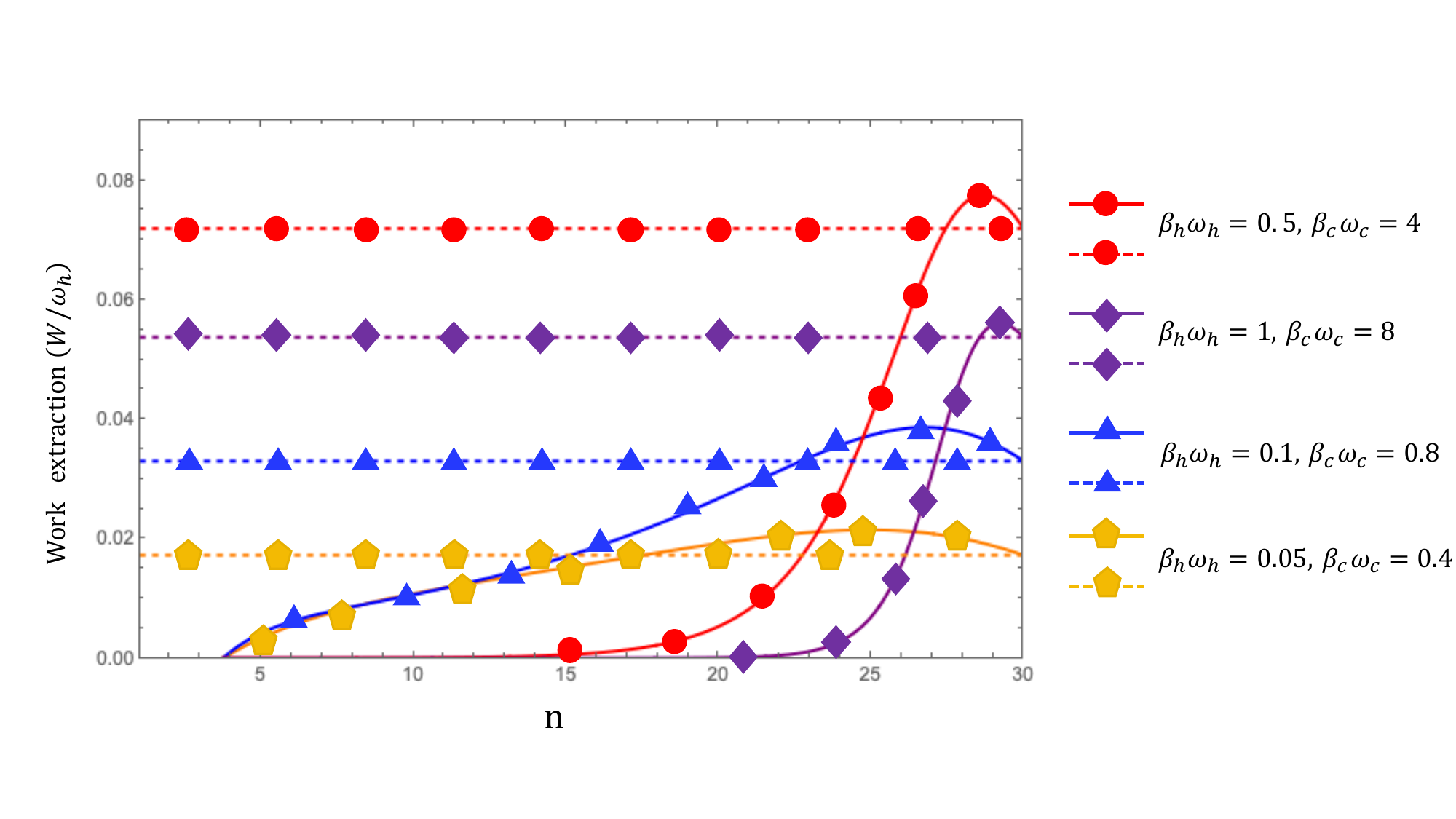}
    \caption{\label{fig:with_without} This figure depicts the variation of the work produced by the catalyst assisted two-stroke engine with $n$ as formulated in Eq. \eqref{Eq:Produced_work} for different values of $\beta_h\omega_h$ and $\beta_c\omega_c$ (represented by the solid curves with different data symbols)  with $\frac{\beta_c\omega_c}{\beta_h\omega_h}=8$. Here we have taken the dimension of catalyst is $30$. Therefore, $n>3.75$ we have positive work produced by the engine as it requires $\frac{d}{n}<\frac{\beta_c\omega_c}{\beta_h\omega_h}$ (See Eq. \eqref{ineq:d_range2}). The dashed lines with different data symbols represent the work produced by the two-stroke engine without a catalyst with the corresponding values of $\beta_h\omega_h$ and $\beta_c\omega_c$. One can see from the plot that work produced by the engine can be enhanced due to catalyst taking certain values of $n$ for each pair of values of $\beta_h\omega_h$ and $\beta_c\omega_c$. On the other hand, from left to right the efficiency of the catalyst assisted two-stroke engine is decreasing as the efficiency is given by the formula $1-\frac{n\omega_c}{d\omega_h}$, whereas the efficiency of the two-stroke engine without catalyst is given by the Otto efficiency. Hence, at rightmost point of the $x$-axis (i.e., at $n=30$) the efficiency of the catalyst assisted two stroke engine matches with the Otto efficiency whereas for $n<30$ the efficiency is strictly greater than Otto efficiency. Therefore, due to catalysis one can obtain a simultaneous enhancement on the work produced by the engine and the efficiency.} 
\end{figure*}

Our analysis centres here around a particular simple permutation $\Pi_{\text{gen}}$ for work extraction. 

  In the appendix \ref{Work_ext_LP}, we have shown that the problem for maximizing work extraction for any arbitrary unitary operation subjected to cyclicity condition can be caste as a convex programming problem. It involves optimization of probability vector achieved from a fixed one via unistochastic matrix, which gives rise to some non-linear constraints. This makes the problem hard to solve. We have relaxed the optimization problem by subjecting it to a bigger set of probability vector achieved from a fixed one via bistochastic matrix. This relaxation reduces the optimization problem to a linear programming problem. We formulate the primal and dual linear programs, with the solution of the latter providing an upper bound on the work produced by the catalyst-assisted qubit two-stroke engines.   The solution of the problem depends on the explicit form of the initial state which is out of scope for this paper. 

\section{Conclusion}

In our study, we present protocols aimed at boosting efficiency of two-stroke heat engines through the assistance of a catalyst. Our approach involves devising of a comprehensive thermodynamic framework for a two-stroke thermal machine operating with and without a catalyst. We focus on the two-stroke heat engine whose working body consists of two two-level systems thermalized at two distinct temperatures. Without a catalyst, the optimal efficiency for such a two-stroke heat engine aligns with the Otto efficiency. In the accompanying paper in Ref. \cite{BiswasLobejko}, a protocol was introduced to surpass the optimal Otto efficiency by incorporating a catalyst into the two-stroke engine. In particular, the accompanying work in Ref. \cite{BiswasLobejko} outlines a protocol resulting in an efficiency of $1-\frac{\omega_c}{d\omega_h}$ for a two-stroke engine when a catalyst is present. 

Here we identify a specific set of permutations, called simple permutations, where the efficiency expressed by $1-\frac{\omega_c}{d\omega_h}$ stands as the optimal value if $\frac{\omega_c}{\omega_h}<d<\frac{\beta_c\omega_c}{\beta_h\omega_h}$. We have constructed a simple permutation $\Pi_{\text{gen}}$ given in Eq. \eqref{Pigen2} (depicted in Fig. \ref{fig:mhotncold}) which allow us to operate the two-stroke engine having a working body consists of two-two level system and a catalyst at any efficiency between $0$ and Carnot efficiency by suitably choosing parameter $n$ and the dimension of the catalyst $d$. Additionally, this permutation enables the engine to operate in any frequency and the temperature regime. Moreover, within a subset of these simple permutations, we conducted an analysis of the trade-off between the work produced by the engine per cycle and the efficiency. Despite an intricate structure of the feasible set for work optimization, we are able to present an upper bound on the work produced by a catalyst-assisted two-stroke engine using linear programming methodologies. On the other hand, due to nonlinearity of efficiency, its optimization becomes more challanging. Finally, we identify scenarios where catalytic enhancements reliably guarantee efficiency enhancements in a two-stroke heat engine.

These results lead to numerous questions: \emph{Firstly, does incorporating a catalyst with the working body always allow to surpass the optimal efficiency of the two-stroke engines without a catalyst?} In this work we have shown that catalytic enhancement in the efficiency of the two-stroke heat engine is always possible if the optimal efficiency is achieved in the non-catalytic scenario for the final state having correlations between the hot and cold $d$-level system. Additionally, we have established that catalytic enhancements in efficiency are achievable even when the optimal efficiency in the non-catalytic scenario is obtained for the final state of the hot and cold $d$-level system is a product of two states, as long as at least one of them is not Gibbs.  Finally, our results show that when the dimension of the hot and cold systems is two, the catalyst always brings the efficiency advantage.  Whether this holds true for an arbitrary dimension is still unknown. We conjecture that catalytic enhancement is always possible.

\emph{Secondly, does incorporating a catalyst lead to enhancement in the performance of other thermodynamic devices and tasks?} Using methodology developed for the two-stroke heat engines, one can address the problem of enhancing the efficiency of the cooler and the  heat accelerators  \cite{Henao2021catalytic, HenaoUzdin}.   In Ref. \cite{Kuba_and_Alex}, it has been shown that one can enlarge the set of achievable states via Markovian thermal process using catalyst of infinite dimension. One could explore the role small-sized catalyst in this scenario to characterize the set achievable achievable states via Markovian thermal process assisted with catalyst of small dimension. 

\emph{Finally, we are interested in bridging the gap between the two-stroke heat engines and self-contained heat engines.} In this paper, our focus has been on stroke-based heat engines, where work extraction and rethermalization occur in distinct, discrete time intervals. It raises the question: whether the enhancement in the efficiency for the two-stroke heat engines is also valid for self-contained heat engines, where thermalization and work extraction happen simultaneously and continuously? We anticipate that there will be a correspondence between the two-stroke discrete heat engine and the self-contained heat engine. In particular, one can expect to show the catalytic enhancement in the efficiency for self-contained thermal machine as well. Note that, the analysis of the self-contained heat engine allows us to compute power or rate of work produced by the engine in time, whereas the proposed two-stroke heat engine provides average work per cycle. After establishing the correspondence between the two-stroke discrete heat engine and the self-contained heat engine rigorously, one can ask: what is the minimum time required by the self-contained heat engine to complete a cycle? We hope that the connection between the discrete two-stroke heat engine and the continuous self-contained heat engine may provide  insights into the experimental realization of two-stroke heat engines. In Ref. \cite{Yu19}, an experimental implementation of the self-contained refrigerator within the framework of cavity quantum electrodynamics was proposed.   Furthermore, in  a more recent work detailed in Ref. \cite{NYHSG}, it has been demonstrated that autonomous absorptive refrigerator can be used for a qubit resetting .   Precisely in this setup, one can add an additional superconducting qubit, acting as a catalyst, which with appropriately designed coupling can enhance the performance of this process. As far as this is implementation of the same theoretical idea, it goes a bit beyond the scope of our paper since it is an autonomous setup.  Hence, we believe that designing a two-stroke heat engine, with or without a catalyst, can also be achieved within the frameworks of cavity QED and superconducting circuits. 

\section*{Acknowledgement}
T.B acknowledges Pharnam Bakhshinezhad for insightful discussions and comments during a visit in quantum information and thermodynamics group at TU, Vienna.  The part of the work done at Los Alamos National Laboratory (LANL) was carried out under the auspices of the U.S. Department
of Energy and National Nuclear Security Administration under Contract No.~DEAC52-06NA25396. TB also acknowledges support by the Department of Energy Office of Science, Office of Advanced Scientific Computing Research, Accelerated Research for Quantum Computing program, Fundamental Algorithmic Research for Quantum Computing (FAR-QC) project. M.Ł, P.M and T.B acknowledge support from the Foundation for Polish Science through International Research Agendas Programme project co-financed by EU within the Smart Growth Operational Programme (contract no.2018/MAB/5). M.H. acknkowledges the support by the Polish National Science Centre grant OPUS-21 (No: 2021/41/B/ST2/03207). M.H. is also partially supported by the QuantERA II Programme, grant No. 2021/03/Y/ST2/00178, ExTRaQT (Experiment and  Theory of Resources in Quantum Technologies), under Grant Agreement No. 101017733, that has received funding from the European Union’s Horizon 2020.

\appendix
\section{Thermodynamical framework}\label{framework_thermodynamic}

In this section of the appendix, we shall develop the thermodynamic framework for the generic two-stroke thermal machines. By selecting the suitable unitary that transforms the working body in the work stroke, the thermal machine can function as an engine, cooler, or heat accelerator. We shall derive the Clausius inequality  from the definition of work given in Eq. \eqref{defn_of_work2}, definition of heat given in Eq. \eqref{defn:hot_heat}, and cold heat given in Eq. \eqref{defn:cold_heat}.
This will allow us to show how the efficiencies of different modes of a thermal machine are interrelated. Let us proceed by describing different modes of thermal machines:
\subsection{Modes of operation for the two-stroke thermal machines}\label{modes}
We can classify the two-stroke thermal machines into three distinct modes based on the work and heat that is associated with its transformation \cite{Campisi1,Campisi2,Campisi3,Campisi4}. 
\begin{enumerate}\label{modes_of_engine}
    \item \emph{Heat engines}: A two-stroke thermal machine work as a heat engine if the work associated with the transformation acting on the working body is positive, i.e., $W>0$. This means the thermal machine is producing the work using the temperature difference between the hot and the cold $d$-level system.
    \item \emph{Coolers}: On the other hand, a two-stroke thermal machine is classified as a cooler if the work associated with the transformation acting on the working body  is negative i.e., $W < 0$ whereas the cold heat $Q_c > 0$. This means that thermal machine utilizes the work in order to cool the target cold $d$-level system by releasing heat into a cold bath. 
    
    \item \emph{  Heat accelerators }: Lastly, the two-stroke thermal machine operates as a   heat accelerators  when the work associated with it is non-positive, indicated as $W \leq 0$, while the heat consumed by the thermal machine is positive, denoted as $Q_h \geq 0$.
\end{enumerate}

\subsection{Derivation of the second law for the two-stroke thermal machines}\label{Thermo_framework}

In this section, we will formulate the second law of thermodynamics for the two-stroke thermal machines. In particular, we shall show that a two-stroke thermal machine operating as a heat engine must withdraw a positive amount of heat from the hot bath in order to produce work. Furthermore, we will establish the efficiency of a two-stroke heat engine is upper bounded by the  Carnot efficiency. On the other hand, when the two-stroke thermal machine operates as a cooler, we will prove that the efficiency is, in turn, lower bounded by the Carnot efficiency. 

We begin by recalling that the initial state of the working body of the two-stroke thermal machine as $\rho^{i}_{s,h,c}=\rho_s\otimes\tau_h\otimes\tau_c$ that transforms via a unitary transformation $U$, that preserves the state of the catalyst as stated definition \ref{principles}. Now, we introduce a class of transformation on the working body of engine called  entropy non-decreasing transformation. This concept will be crucial to establish the thermodynamic framework for the two-stroke thermal machine. 
\begin{defn}
    [\textbf{Entropy non-decreasing transformation}] An entropy non-decreasing transformation $\mathcal{U}$ is a complete positive trace preserving (CPTP) map that acts on the initial state of the working body (i.e., catalyst, the hot and the cold $d$-level system) $\rho_{s,h,c}$, such that 
    \begin{equation}
        S(\rho_{s,h,c})\leq S(\mathcal{U}(\rho_{s,h,c})),
    \end{equation}
    where $S(\cdot)$ denotes the von Neumann entropy which is defined as 
    \begin{equation}\label{vNentropy}
        S(\cdot)=-\Tr((\cdot)\ln(\cdot)).
    \end{equation}
\end{defn}
Clearly, any unitary operation is an entropy non-decreasing transformation as it preserves the von-Neumann entropy. Moreover, the set of states that can be achieved from a fixed state $\rho$ via entropy non-decreasing transformations is convex. This can be seen as follows:
\begin{multline}
    S(\lambda\mathcal{U}_1(\rho)+(1-\lambda)\mathcal{U}_2(\rho))\geq \lambda S(\mathcal{U}_1(\rho))+(1-\lambda)S(\mathcal{U}_2(\rho))\\
    \geq\lambda S(\rho)+(1-\lambda)S(\rho) = S(\rho).
\end{multline}
The main motivation for introducing these class of operations is to consider a generic set of transformations on the Hilbert space of the working body (i.e., catalyst, hot and cold $d$-level system) that encompasses any arbitrary unitary transformations on them. 

In the next section, we shall see that optimization of the efficiency as well as work produced by the two-stroke thermal engines boils down to the set of all states that can be achieved via permutations from the initial state of the working body. Furthermore, employing the convexity of the achievable states via the entropy non-decreasing transformation allows us to write the extraction of work as a linear program.

Now, we shall show the second law holds even if the initial state of working body $\rho_s\otimes\tau_h\otimes\tau_c$ transforms via entropy non-decreasing transformation, which naturally gives the second law if the initial state of the working body transforms via a unitary.

\begin{lem} [Clausius inequality]\label{second_law_lemma}
For any  two-stroke thermal machine with initial state of the working body (i.e., catalyst, the hot and the cold $d$-level system)  $\rho^{i}_{s,h,c}$ that transforms via an entropy non-decreasing transformation to the final state $\rho^{f}_{s,h,c}$ should satisfy the following inequality:
\begin{eqnarray} \label{second_law_ineq}
    \beta_hQ_h+\beta_cQ_c\leq 0.
\end{eqnarray}
\end{lem}
\begin{proof} The proof of the lemma follows from the non-negativity of relative entropy i.e 
\begin{eqnarray}
    &&D(\rho^{f}_{s,h,c}\| \rho^{i}_{s,h,c}):=-S(\rho^{f}_{s,h,c}) -\Tr(\rho^{f}_{s,h,c}\log\rho^{i}_{s,h,c})\geq 0\nonumber\\
    &\Rightarrow& -S(\rho^{i}_{s,h,c}) -\Tr(\rho^{f}_{s,h,c}\log\rho^{i}_{s,h,c})-\delta \geq 0\nonumber \\
    &\Rightarrow& \Tr\big((\rho^{i}_{s,h,c}-\rho^{f}_{s,h,c})\log\rho^{i}_{s,h,c}\big)-\delta \geq 0\nonumber\\
    &\Rightarrow& \Tr\big((\rho^{i}_{s,h,c}-\rho^{f}_{s,h,c})(\log\rho_s+\log\tau_h+\log\tau_c)\big)-\delta \geq 0\nonumber\\&\Rightarrow& \Tr\big((\rho^{i}_{s,h,c}-\rho^{f}_{s,h,c})(-\beta_hH_h-\beta_cH_c)\big)-\delta \geq 0\nonumber\\
    &\Rightarrow& -(\beta_hQ_h+\beta_cQ_c)-\delta \geq 0\nonumber\\
    &\Rightarrow& 0 \geq-\delta \geq (\beta_hQ_h+\beta_cQ_c).\nonumber
\end{eqnarray}
Here, to write the first implication we use the definition of the entropy non-decreasing transformation with $\delta\geq 0$, to write the second implication we use simply the definition of von-Neumann entropy, to write the third implication we uses the fact that the initial state of the working body $\rho^i_{s,h,c}=\rho_s\otimes\tau_h\otimes\tau_c$. Now, in order to write the fourth implication, we use the fact that marginal state of the catalyst in the initial and the final state are same, and the definition of $\tau_h$ and $\tau_c$ from Eq. \eqref{Gibbs_defn}.
\end{proof}

Note that, the second law inequality given in Eq. \eqref{second_law_ineq} holds even if we consider a two-stroke thermal machines having working body that does not contain any catalyst. This can be seen easily from the derivation. Next, using the Clausius inequality from lemma \ref{second_law_lemma}, we shall prove that the transferred amount of heat from the hot bath is always positive for the two-stroke thermal machine operating as an engine.

\begin{prop}[Positivity of heat transfer in the two-stroke heat engines]\label{Lemma_Work_positiv(1-g_H)eat_Positive}
    Any two-stroke thermal machine operating as an engine with initial state of the working body (i.e., the catalyst, the hot and the cold $d$-level system) $\rho^{i}_{s,h,c}$ evolves via an entropy non-decreasing transformation to the final state $\rho^{f}_{s,h,c}$, always leads to a positive amount of heat transfer from the hot heat bath i.e.,
\begin{equation}\label{Work_pos_heat_pos}
    W> 0 \Rightarrow Q_h>0.
\end{equation}
\end{prop}
\begin{proof} 
This proof directly follows from second law inequality given in Eq. \eqref{second_law_ineq}, and the fact that work associated with thermal machine operated as an engine is always positive as defined in Sec. \ref{modes}. The second law inequality in Eq. \eqref{second_law_ineq}  can be rewritten as 
\begin{equation}\label{second_law_v2}
   -\frac{\beta_h}{\beta_c}Q_h \geq Q_c. 
\end{equation}
On the other hand, positivity of work implies 
\begin{equation}\label{above}
    W = Q_h+Q_c > 0 \Rightarrow Q_h(1-\frac{\beta_h}{\beta_c}) > 0 \Rightarrow Q_h>0,
\end{equation}
where we use the definition of work from Eq. \eqref{Work_redefined}, the inequality from Eq. \eqref{second_law_v2} and the fact $\beta_c>\beta_h$, to draw the implication in Eq. \eqref{above}.
\end{proof}
So, from this proposition, we conclude that if the two-stroke thermal machines operate as an engine, then it should withdraw a positive amount of heat from the hot heat bath to produce work. Finally, employing lemma \ref{second_law_lemma} and theorem \ref{Lemma_Work_positiv(1-g_H)eat_Positive}, we shall establish the ordering among the efficiency between different modes of the two-stroke thermal machines.
\begin{prop}[Ordering of the efficiency between different modes of thermal machines]\label{interplay}
    The following statements holds true for any two-stroke thermal machines operating with the efficiency $\eta$:
    \begin{enumerate}
        \item $0 <\eta \leq\left(1-\frac{\beta_h}{\beta_c}\right)$  if and only if the two-stroke thermal machines operates as an engine.
        \item $\eta \geq \left(1-\frac{\beta_h}{\beta_c}\right)$  if and only if the two-stroke thermal machines operates as a cooler.
        \item $\eta \leq 0$ if and only if the two-stroke thermal machines operates as a   heat accelerators .
    \end{enumerate}
\end{prop}
\begin{proof}
    From theorem \ref{Lemma_Work_positiv(1-g_H)eat_Positive}, we see that for two-stroke thermal machine operated as an engine, the transferred amount of heat from the hot bath $Q_h$ is positive which proves the efficiency is positive. On the other hand, from the second law inequality given in Eq. \eqref{second_law_v2}, and using the fact $Q_h>0$ we can derive $-\frac{\beta_h}{\beta_c}\geq \frac{Q_c}{Q_h}$. This implies $\eta = 1+\frac{Q_c}{Q_h}\leq 1-\frac{\beta_h}{\beta_c}$.

    On the other hand, for the two stroke thermal machine operated in the cooling mode have $Q_h<0$ (see the Sec. \ref{modes}). Therefore, the second law inequality given in Eq. \eqref{second_law_v2} can be reduced to $-\frac{\beta_h}{\beta_c}\leq \frac{Q_c}{Q_h}$ which implies $\eta=1+\frac{Q_c}{Q_h}\geq 1-\frac{\beta_h}{\beta_c}$.

    Finally, for the two stroke thermal machine operated as an   heat accelerators  have $W<0$ and $Q_h>0$ as defined in Sec. \ref{modes} which makes its efficiency $\eta\leq 0$. 
\end{proof}

Thus, the proposition \ref{interplay} sets the ordering the efficiency of the different modes of the thermal machine starting with the fixed  initial state of the working body $\rho_{s,h,c}^{i}=\rho_s\otimes\tau_h\otimes\tau_c$ in the following manner:
\begin{equation}\label{crucial_ordering_of_efficiency}
    \eta^{\text{max}}_{\text{  heat accelerators }} \leq \eta^{\text{min}}_{\text{heat engine}} \leq \eta^{\text{max}}_{\text{heat engine}} \leq \eta^{\text{min}}_{\text{cooler}},
\end{equation}
where $\eta^{\text{max}}_{x}$ and $\eta^{\text{min}}_{x}$ denotes the maximum and minimum efficiency when the two-stroke thermal machine operates in the mode $x$. 

  It is important to emphasize that Eq. \eqref{crucial_ordering_of_efficiency} only holds for fixed initial state of the thermal machine. Hence, we can not use this inequality directly to compare in a two-stroke thermal machine with catalyst where the initial state of the catalyst varies with the transformation acting on the initial state of the working body of the thermal machine. This results changing of the initial state of thermal machine with the choice of the unitary that is implemented on the working body, as the state of the catalyst depends on the unitary that has been implemented (In other words, state of the catalyst comes from the solution of Eq. \eqref{eq:marginal_cyclic}).  

Since proposition \ref{Lemma_Work_positiv(1-g_H)eat_Positive} and proposition \ref{interplay} both stems from the Clausius inequality established in lemma \ref{second_law_lemma} that holds true for the two-stroke thermal machines without a catalyst, these results in theorem \ref{Lemma_Work_positiv(1-g_H)eat_Positive} and corollary \ref{interplay} also hold true for such machines. As the initial state of the working body for the two-stroke thermal machine in the absence of the catalyst is $\rho^i_{h,c}=\tau_h\otimes\tau_c$ which does not depends on the transformation, the inequality given in Eq. \eqref{crucial_ordering_of_efficiency} applies directly   (Because there is no catalyst in the working body. The state of the catalyst in the working body depends on the unitary that has been implemented in the work stroke, whereas state of the hot and the cold $d$-level system does not depend on the unitary that is going to be implemented during the work stroke) . We shall use this inequality later to optimize the efficiency of two-stroke heat engines. 

\section{Proof of theorem \ref{thm_optimal_efficiency}}\label{Permutation_optimal_eff}
\begin{proof}
    Starting from Eq. \eqref{defn_of_work2}, we can write work produced by a two-stroke heat engine without catalyst as $W= \Tr\big((H_h+H_c)(\rho^i_{h,c}-U\rho^i_{h,c}U^{\dagger})\big)$. It is straightforward to see that $W$ is a linear function of $\rho^i_{h,c}$ and does not depend on the off-diagonal terms of the state $U\rho^{i}_{h,c}U^{\dagger}$. Thus we can write work produced by this engine as
    \begin{equation}\label{work_t4}
        W = \Tr\Big((H_h+H_c)\big(\rho_{h,c}^i-\mathcal{D}(U\rho_{h,c}^iU^{\dagger})\big)\Big),
    \end{equation}
    where $\mathcal{D}(\cdot)$ denotes the dephasing in the eigenbasis of the Hamiltonian $H_h+H_c$. According to Schur-Horn theorem, $|\mathcal{D}(U\rho_{h,c}^iU^{\dagger})\dket = B|\rho_{h,c}^i\dket$ where $|\cdot\dket$ denotes the spectrum of $(\cdot)$, and $B$ is some bistochastic matrix. Therefore, it is straight-forward to rewrite the expression of work $W$ in Eq. \eqref{work_t4} as 
    \begin{eqnarray}\label{work_in_proof_sketch}
        &&W=\dbra H_h+H_c|\rho_{h,c}^i\dket-\dbra H_h+H_c|B|\rho_{h,c}^i\dket.
    \end{eqnarray}
From Birkhoff-vonNeumann theorem we know that one can decompose any bistochastic matrix as a convex sum of permutations i.e $B=\sum_i \alpha_i\Pi_i$. Therefore, we can have the following:
\begin{eqnarray}
        W&=& \sum_{k}\alpha_k\Big(\dbra H_h+H_c|\rho_{h,c}^i\dket-\dbra H_h+H_c|\Pi_k|\rho^i_{h,c}\dket\Big)\nonumber\\
        &\leq& \dbra H_h+H_c|\rho_{h,c}^i\dket-\dbra H_h+H_c|\Pi^{*}|\rho^i_{h,c}\dket,
    \end{eqnarray}
where $\Pi^{*}$ is the permutation such that
\begin{equation}
   \min_k \dbra H_h+H_c|\Pi_k|\rho^i_{h,c}\dket = \dbra H_h+H_c|\Pi^{*}|\rho^i_{h,c}\dket := W_{\Pi^{*}}.
\end{equation}
Next, in order to optimize the efficiency we proceed by writing efficiency as follows:
\begin{eqnarray}\label{convex_sum_eff}
    \eta &=&\sum_{k}\frac{\alpha_k Q_{\Pi_k}}{Q_h}\frac{W_{\Pi_k}}{Q_{\Pi_k}}=\sum_{k}\frac{\alpha_k Q_{\Pi_k}}{Q_h}\eta_{\Pi_k},\label{eq:eff_decomposition2}
\end{eqnarray}
where 
\begin{eqnarray}
    W_{\Pi_k} &=& (\dbra H_h+H_c | \rho^{i}_{h,c}\dket - \dbra H_h+H_c |  \Pi_k|\rho^{i}_{h,c}\dket),
    \\Q_{\Pi_k} &=& (\dbra H_h | \rho^{i}_{h,c}\dket - \dbra H_h |  \Pi_k|\rho^{i}_{h,c}\dket),\\
    Q_h &=& (\dbra H_h | \rho^{i}_{h,c}\dket - \dbra H_h | B|\rho^{i}_{h,c}\dket) = \sum_k \alpha_k Q_{\Pi_k}.\label{Heat_in_proof_sketch}
\end{eqnarray}
Note that, Eq. \eqref{convex_sum_eff} is not a convex sum because there may exist some $k$ for which $Q_{\Pi_k}$ is negative. On the other hand $Q_h>0$ due to proposition \ref{Lemma_Work_positiv(1-g_H)eat_Positive}. Furthermore, employing proposition \ref{Lemma_Work_positiv(1-g_H)eat_Positive} again we infer that $Q_{\Pi_j}>0$ whenever $W_{\Pi_j}>0$, which implies any permutation appeared in Eq. \eqref{convex_sum_eff} must belong to either of the following sets
\begin{eqnarray}\label{sets}
    \mathcal{J}_E &:=& \{j \quad\text{such that}\quad W_{\Pi_j} > 0\},\\
    \mathcal{J}_C &:=& \{j \quad\text{such that}\quad W_{\Pi_j} < 0\;\text{and}\; Q_{\Pi_j}<0\},\\
    \mathcal{J}_H &:=& \{j \quad\text{such that}\quad W_{\Pi_j} < 0\;\text{and}\; Q_{\Pi_j}\geq0\}.
\end{eqnarray}
This simply means for $j\in\mathcal{J}_E$ the permutation $\Pi_j$ corresponds to engine mode,  for $j\in\mathcal{J}_C$ the permutation $\Pi_j$ corresponds to cooling mode, and for $j\in\mathcal{J}_H$ the permutation $\Pi_j$ corresponds to   heat accelerators  mode. Hence, the expression of efficiency given in Eq. \eqref{convex_sum_eff} reduces to
\begin{equation}\label{eff_intermediate}
    \eta = \sum_{j\in \mathcal{J}_E}\frac{\alpha_j Q_{\Pi_j}}{Q_h}\eta_{\Pi_j}+\sum_{j\in\mathcal{J}_C}\frac{\alpha_j Q_{\Pi_j}}{Q_h}\eta_{\Pi_j}+\sum_{j\in\mathcal{J}_H}\frac{\alpha_j Q_{\Pi_j}}{Q_h}\eta_{\Pi_j}.
\end{equation}
Now let us define 
\begin{eqnarray}
    \eta^{\text{max}}_{\mathcal{J}_x}:=\max_{\Pi_j\in \mathcal{J}_x}\eta_{\Pi_j};\quad \eta^{\text{min}}_{\mathcal{J}_x}:=\min_{\Pi_j\in \mathcal{J}_x}\eta_{\Pi_j} \nonumber 
\end{eqnarray}
where $x \in \{E,H,C\}$. Thus we can write the following inequality
\begin{eqnarray}\label{ineq_app_intermediate}
    &&\eta\leq \sum_{j\in \mathcal{J}_E}\frac{\alpha_j Q_{\Pi_j}}{Q_h}\eta^{\text{max}}_{\mathcal{J}_E}+\sum_{j\in\mathcal{J}_C}\frac{\alpha_j Q_{\Pi_j}}{Q_h}\eta^{\text{min}}_{\mathcal{J}_C}\nonumber\\&+&\sum_{j\in\mathcal{J}_H}\frac{\alpha_j Q_{\Pi_j}}{Q_h}\eta^{\text{max}}_{\mathcal{J}_H},
\end{eqnarray}
where we use the fact $Q_{\Pi_j}>0$ for all $j\in \mathcal{J}_E$ and for all $j\in \mathcal{J}_H$ whereas $Q_{\Pi_j}<0$ for all $j\in \mathcal{J}_C$ as defined in Eq. \eqref{sets}. Due to inequality at Eq. \eqref{crucial_ordering_of_efficiency} which is true for any unitaries that transforms the state of the thermal machine during work stroke, we can write
\begin{eqnarray}
    \eta^{\text{max}}_{\mathcal{J}_H}\leq \eta^{\text{max}}_{\text{  heat accelerators }}&\leq&\eta^{\text{min}}_{\text{heat engine}}\\&\leq& \eta^{\text{max}}_{\mathcal{J}_E}\leq\eta^{\text{min}}_{\text{cooler}}\leq\eta^{\text{min}}_{\mathcal{J}_C}\nonumber,
\end{eqnarray}
which allows us to further bound the inequality given in Eq. \eqref{ineq_app_intermediate}
\begin{eqnarray}
    &&\eta\leq \sum_{j\in \mathcal{J}_E}\frac{\alpha_j Q_{\Pi_j}}{Q_h}\eta^{\text{max}}_{\mathcal{J}_E}+\sum_{j\in\mathcal{J}_C}\frac{\alpha_j Q_{\Pi_j}}{Q_h}\eta^{\text{min}}_{\mathcal{J}_C}\nonumber\\&+&\sum_{j\in\mathcal{J}_H}\frac{\alpha_j Q_{\Pi_j}}{Q_h}\eta^{\text{max}}_{\mathcal{J}_H}\nonumber\leq \eta^{\text{max}}_{\mathcal{J}_E}\Big(\sum_{j}\frac{\lambda_j Q_{\Pi_j}}{Q_h}\Big)\nonumber\\&=& \eta^{\text{max}}_{\mathcal{J}_E},
\end{eqnarray}
which completes the proof.

\end{proof}
Thus, from the proof of the theorem \ref{thm_optimal_efficiency}, we see that optimal efficiency for a two-stroke heat engine is achieved for some permutation. In other words, if the initial state of the working body transformed via some unitary $\tilde{U}$ such that $|\tilde{\mathcal{D}}(\tilde{U}\rho^i_{h,c}\tilde{U}^{\dagger})\dket = \tilde{B}|\rho^i_{h,c}\dket$ where $\tilde{B}$ is some non-trivial mixture of permutations, then $\tilde{U}$ can not lead to optimal production of work or the efficiency. 
\section{ Performance of the non-catalytic two-stroke heat engine  where the dimension of the hot and cold $d$ level systems present in the working body is two. }
\label{SmallestHE}
\begin{table*}[ht]
\centering
\begin{tabular}{|*{5}{c|}}
\hline
 & \textbf{Permutation matrices} &  \textbf{Work $(W)$} &  \textbf{Efficiency $\left(\eta\right)$} \\
\hline
\hline
1 & $\mathbb{I}=\ketbra{00}{00}+\ketbra{01}{01}+\ketbra{10}{10}+\ketbra{11}{11}$  & $0$  & $0$ \\\hline
\cblue 2    &\cblue $\Pi=\ketbra{00}{00}+\ketbra{01}{10}+\ketbra{10}{01}+\ketbra{11}{11}$    & $\cblue \N\left(a_h-a_c\right)(\omega_h-\omega_c)$    & \cblue $1-\frac{\omega_c}{\omega_h}$   \\
\hline
 3 &$\ketbra{00}{00} + \ketbra{01}{11} + \ketbra{10}{01} + \ketbra{11}{10}$  &  $-\N\left[\omega_c\left(a_h-a_c\right)+a_c\omega_h(1-a_h)\right]$  &  $1+\frac{(a_h-a_c)\omega_c}{a_c(1-a_h)\omega_h}$ \\
\hline
\cblue4   & \cblue $\ketbra{00}{00} + \ketbra{01}{10} + \ketbra{10}{11} + \ketbra{11}{01}$    & \cblue $\N[\omega_h(a_h-a_c)-\omega_c(a_h-a_ha_c)]$    & \cblue $1-\frac{(a_h-a_ca_h)\omega_c}{(a_h-a_c)\omega_h}$    \\
\hline
5 &$\ketbra{00}{00} + \ketbra{01}{11} + \ketbra{10}{10} + \ketbra{11}{01}$  & $-\N \omega_ha_c(1-a_h)$  & $1$ \\
\hline
6 &$\ketbra{00}{00} + \ketbra{01}{01} + \ketbra{10}{11} + \ketbra{11}{10}$  & $-\N\left(1-a_c\right)a_h\omega_c$  & $-\infty$ \\
\hline
7 &$\ketbra{00}{01} + \ketbra{01}{00} + \ketbra{10}{10} + \ketbra{11}{11}$  & $-\N\left(1-a_c\right)\omega_c$  & $-\infty$ \\
\hline
8 &$\ketbra{00}{01} + \ketbra{01}{00} + \ketbra{10}{11} + \ketbra{11}{10}$  & $-\N\left(1-a_c\right)\left(1+a_h\right)\omega_c$  & $-\infty$ \\
\hline
9 & $\ketbra{00}{01} + \ketbra{01}{10} + \ketbra{10}{00} + \ketbra{11}{11}$   &  $-\N[\omega_h(1-a_h)+\omega_c(a_h-a_c)]$  &  $1+\frac{(a_h-a_c)\omega_c}{(1-a_h)\omega_h}$  \\
\hline
 10  &  $\ketbra{00}{01} + \ketbra{01}{11} + \ketbra{10}{00} + \ketbra{11}{10}$  &  $-\N[\left(a_h-a_c\right)\omega_c+(1-a_ca_h)\omega_h]$    &  $1+\frac{
(a_h-a_c)\omega_c}{(1-a_ca_h)\omega_h}$  \\
\hline
11 &$\ketbra{00}{01} + \ketbra{01}{10} + \ketbra{10}{11} + \ketbra{11}{00}$  & $-\N[(1-a_c)(1+a_h)\omega_c+(1-a_h)\omega_h]$  & $1+\frac{(1-a_c)(1+a_h)\omega_c}{(1-a_h)\omega_h}$ \\
\hline
12 &$\ketbra{00}{01} + \ketbra{01}{11} + \ketbra{10}{10} + \ketbra{11}{00}$ & $-\N[\omega_c(1-a_c)+\omega_h(1-a_ca_h)]$  & $1+\frac{\omega_c(1-a_c)}{\omega_h(1-a_ca_h)}$ \\
\hline
\cblue 13    & \cblue $\ketbra{00}{10} + \ketbra{01}{00} + \ketbra{10}{01} + \ketbra{11}{11}$    & $ \cblue \N[\omega_h(a_h-a_c)-\omega_c(1-a_c)]$    & \cblue $1-\frac{\omega_c(1-a_c)}{\omega_h(a_h-a_c)}$    \\
\hline
14 &$\ketbra{00}{10} + \ketbra{01}{11} + \ketbra{10}{01} + \ketbra{11}{00}$  & $-\N[(1-a_c)\omega_c+(1+a_c)(1-a_h)\omega_h]$  & $1+\frac{(1-a_c)\omega_c}{(1+a_c)(1-a_h)\omega_h}$ \\
\hline
\cblue15   & \cblue $\ketbra{00}{10} + \ketbra{01}{00} + \ketbra{10}{11} + \ketbra{11}{01}$    & \cblue $\N[(a_h-a_c)\omega_h-(1-a_ca_h)\omega_c]$    & \cblue $1-\frac{\omega_c(1-a_ca_h)}{\omega_h(a_h-a_c)}$    \\
\hline
16 &$\ketbra{00}{10} + \ketbra{01}{11} + \ketbra{10}{00} + \ketbra{11}{01}$  & $-\N(1+a_c)(1-a_h)\omega_h$  & $1$ \\
\hline
17 &$\ketbra{00}{10} + \ketbra{01}{01} + \ketbra{10}{00} + \ketbra{11}{11}$  & $-\N(1-a_h)\omega_h$  & $1$ \\
\hline
18 &$\ketbra{00}{10} + \ketbra{01}{01} + \ketbra{10}{11} + \ketbra{11}{00}$  & $-\N [\omega_c(1-a_ca_h)+\omega_h(1-a_h)]$  & $1+\frac{\omega_c(1-a_ca_h)}{\omega_h(1-a_h)}$ \\
\hline
19 &$\ketbra{00}{11} + \ketbra{01}{00} + \ketbra{10}{01} + \ketbra{11}{10}$  & $-\N[(1-a_c)(1+a_h)\omega_c+a_c(1-a_h)\omega_h]$  & $1+\frac{(1-a_c)(1+a_h)\omega_c}{(1-a_h)a_c\omega_h}$ \\
\hline
20 &$\ketbra{00}{11} + \ketbra{01}{10} + \ketbra{10}{01} + \ketbra{11}{00}$  & $\N[(1+a_h)(1-a_c)\omega_c+(1+a_c)(1-a_h)\omega_h]$  & $1+\frac{(1-ac)(1+a_h)\omega_c}{(1+a_c)(1-a_h)\omega_h}$ \\
\hline
21 &$\ketbra{00}{11} + \ketbra{01}{00} + \ketbra{10}{10} + \ketbra{11}{01}$  & $-\N[(1-a_ca_h)\omega_c+a_c(1-a_h)\omega_h]$  & $1+\frac{\omega_c(1-a_ca_h)}{\omega_h(a_c-a_ca_h)}$ \\
\hline
22 &$\ketbra{00}{11} + \ketbra{01}{01} + \ketbra{10}{00} + \ketbra{11}{10}$  & $-\N[(1-a_c)a_h\omega_c+\omega_h(1-a_ca_h)]$  & $1+\frac{(1-a_c)a_h\omega_c}{(1-a_ca_h)\omega_h}$ \\
\hline
23 &$\ketbra{00}{11} + \ketbra{01}{10} + \ketbra{10}{00} + \ketbra{11}{01}$  & $-\N[(1-a_c)a_h\omega_c+(1+a_c)(1-a_h)\omega_h]$  & $1+\frac{(1-a_c)a_h\omega_c}{(1+a_c)(1-a_h)\omega_h}$ \\
\hline
24 &$\ketbra{00}{11} + \ketbra{01}{01} + \ketbra{10}{10} + \ketbra{11}{00}$  & $-\N\left(1-a_ca_h\right)(\omega_c+\omega_h)$  & $1+\frac{\omega_c}{\omega_h}$ \\
\hline
\end{tabular}
\caption{The work produced and the efficiency of the engine when the working body composed of two two-level systems thermalized at two different temperatures transformed by permutations. We can see from the table that, among all $24$ permutations, only four of the permutations (i.e permutation $2$, $4$, $13$, $15$ coloured in blue) results a positive amount of work. Among these four permutations, the maximum efficiency given by the Otto efficiency is achieved for the permutation $\Pi$ (2nd permutation in the list). Here $a_h=e^{-\beta_h\omega_h},\;a_c = e^{-\beta_c\omega_c}\;\; \text{and}\;\;\N = \frac{1}{(1+a_h)(1+a_c)}$.}
\label{tab:my_table_efficiency_24}
\end{table*}
 In this section, we shall analyze the performance of the two-stroke heat engine where dimension of the hot and cold $d$ level systems present in the working body is two. The hot and cold $d$-level system are in Gibbs state at inverse temperature $\beta_h$ and $\beta_c$, respectively. The Hamiltonian associated with the hot and the cold $d$-level system is given by $H_h = \omega_h|1\rangle\langle 1|$ and $H_c = \omega_c|1\rangle\langle 1|$.  
Therefore, the spectrum of the initial state of the working body is given by
\begin{eqnarray}\label{in_state}
    |\rho^i_{h,c}\dket &=& \frac{1}{1+a_h}(1\;\; a_h)^T\otimes\frac{1}{1+a_c}(1\;\; a_c)^T, \nonumber\\
    &=& \frac{1}{(1+a_h)(1+a_c)}(1\;\; a_c\;\; a_h\;\; a_ca_h)^T,
\end{eqnarray}
where $a_h = e^{-\beta_h\omega_h}$ and $a_c = e^{-\beta_c\omega_c}$. 
The total Hamiltonian is given as 
$H = H_h\otimes \mathbb{I}+\mathbb{I}\otimes H_c$, thus the spectrum of the Hamiltonian $H$ can be written as
\begin{eqnarray}
    &&|H\dket = (0\;\;\omega_c\;\;\omega_h\;\;\omega_c+\omega_h)^T;\nonumber\\&&\quad|H_h\dket = (0\;\;0\;\;\omega_h\;\;\omega_h)^T.
\end{eqnarray}
From the initial state given in Eq. \eqref{in_state}, we see that only four permutation $\Pi$ leads to positive value for extraction of work that is given in table \ref{tab:my_table_efficiency_24}. Among these four permutations that lead to positive work extraction, the permutation
\begin{equation}\label{Pie_at_app}
    \Pi = |0,0\rangle\langle 0,0|+ |0,1\rangle\langle 1,0|+|1,0\rangle\langle 0,1|+|1,1\rangle\langle 1,1|,
\end{equation}
results the maximum efficiency of the engine which is given by the Otto efficiency. This results the following final state,
\begin{eqnarray}
|\rho^{i}_{h,c}\dket&=&\frac{1}{(1+a_h)(1+a_c)}(1\;\; a_c\;\; a_h\;\; a_ca_h)^T, \\&\xrightarrow{\Pi}& \frac{1}{(1+a_h)(1+a_c)}(1\;\; a_h\;\; a_c\;\; a_ca_h)^T := |\rho^{f}_{h,c}\rangle.\nonumber
\end{eqnarray}
Therefore, the maximum work produced by the engine is given by

\begin{eqnarray}
     W &=& \Tr\Big(H(\rho^{i}_{h,c}-\rho^{f}_{h,c})\Big) =  \dbra H| \rho^{i}_{h,c}\dket -\dbra H| \rho^{f}_{h,c}\dket \nonumber\\&=& \frac{1}{(1+a_h)(1+a_c)} (a_h-a_c)(\omega_h-\omega_c).
\end{eqnarray}
Similarly, the amount of heat withdrawn from the hot bath can be calculated as
\begin{eqnarray}
     Q_h &=& \Tr\Big(H_h\otimes\mathbb{I}(\rho^{i}_{h,c}-\rho^{f}_{h,c})\Big)= \dbra H_h| \rho^{i}_{h,c}\dket -\dbra H_h| \rho^{f}_{h,c}\dket\nonumber\\&=& \frac{1}{(1+a_h)(1+a_c)}(a_h-a_c)\omega_h.
\end{eqnarray}
For this permutation, efficiency is given by the Otto efficiency
\begin{equation}
    \eta = \frac{W}{Q} = 1- \frac{\omega_c}{\omega_h},
\end{equation}
which is the optimal efficiency.

\section{Proof of proposition \ref{proposition_coherence_useless}}\label{useless_coherence}
\begin{proof}
    We begin by rewriting the definition of work and the efficiency of the two-stroke heat engine given in Eq. \eqref{Work_redefined} and Eq. \eqref{efficiency_def}  i.e,
    \begin{equation}
        W = Q_c+Q_h\;,\; \quad\eta = 1+\frac{Q_c}{Q_h}.
    \end{equation}
    From the above equations we can infer that in order to prove the claim it is enough to construct an unitary $\tilde{U}$ and a catalyst $\tilde{\rho}_s$ that gives the transferred amount of heat $\tilde{Q}_h$ and the cold heat $\tilde{Q}_c$ exactly same as $Q_h$ and $Q_c$ i.e., $\tilde{Q}_h=Q_h$ and $\tilde{Q}_c=Q_c$. Let us assume the initial state of the working body of two-stroke thermal machine $\rho_s\otimes\tau_h\otimes\tau_c$ transforms via an unitary $U$ during the work stroke, we have 
    \begin{eqnarray}
        Q_h &=& \Tr\big(H_h(\rho_s\otimes\tau_h\otimes\tau_c-U(\rho_s\otimes\tau_h\otimes\tau_c)U^{\dagger})\big)\nonumber\\
        Q_c &=& \Tr\big(H_c(\rho_s\otimes\tau_h\otimes\tau_c-U(\rho_s\otimes\tau_h\otimes\tau_c)U^{\dagger})\big)\nonumber\\
        &&\text{such that   } \Tr_{h,c} U(\rho_s\otimes\tau_h\otimes\tau_c)U^{\dagger} = \label{cata}\rho_s .
    \end{eqnarray}
   One can diagonalize $\rho_s = K\rho_s^D K^{\dagger}$ where $K$ is unitary and $\rho_s^D$ is a diagonal matrix that contains the eigenvalues of $\rho_s$. Now we choose 
   \begin{equation}\label{forty}
       \tilde{\rho}_s = \rho^D_s\quad,\quad \tilde{U}=K^{\dagger} U K. 
   \end{equation}
   Then one can write 
   \begin{eqnarray}
       \tilde{Q}_h &=& \Tr\big(H_h(\tilde{\rho}_s\otimes\tau_h\otimes\tau_c-\tilde{U}(\tilde{\rho}_s\otimes\tau_h\otimes\tau_c)\tilde{U}^{\dagger})\big) \nonumber\\
        &=& \Tr\big(H_h(K^{\dagger}\rho_s K\otimes\tau_h\otimes\tau_c\nonumber\\&-&K^{\dagger} U K(K^{\dagger}\rho_s K\otimes\tau_h\otimes\tau_c)K^{\dagger}U^{\dagger} K)\big) \nonumber\\
        &=& \Tr\big(H_h(\rho_s\otimes\tau_h\otimes\tau_c-U(\rho_s\otimes\tau_h\otimes\tau_c)U^{\dagger}) = Q_h.\nonumber\\
   \end{eqnarray}
   where we use unitarity of $K$ i.e., $KK^{\dagger} = \mathbb{I}$ and $[K,H_h]=[K^{\dagger},H_h] = 0$ as they act on different Hilbert space. In a similar manner we can show 
   \begin{equation}
       \tilde{Q}_c = \Tr\big(H_c(\tilde{\rho}_s\otimes\tau_h\otimes\tau_c-\tilde{U}(\tilde{\rho}_s\otimes\tau_h\otimes\tau_c)\tilde{U}^{\dagger})\big) =Q_c.
   \end{equation}
   Furthermore, substituting $U$ and $\rho^D_s$ from Eq. \eqref{forty} in Eq. \eqref{cata} one can straight-forwardly see 
   \begin{equation}
       \Tr_{h,c} \tilde{U}(\tilde{\rho}_s\otimes\tau_h\otimes\tau_c)\tilde{U}^{\dagger} = \tilde{\rho}_s.
   \end{equation}
\end{proof}
\section{Proof of theorem \ref{catalyst:enhancement}}\label{appendix4}

In order to do the proof, we shall proceed by writing the efficiency from Eq. \eqref{efficiency_def} as 
\begin{equation}
    \eta = 1+\frac{Q_c}{Q_h}.
\end{equation}
From the definition of transferred amount of heat in Eq. \eqref{defn:hot_heat} and the cold heat in Eq. \eqref{defn:cold_heat}, and the definition of the hot and cold subspaces from Def. \ref{hot_and_cold_subspaces}, one can calculate $Q_c$ and $Q_h$ in this scenario as the net amount of population transfer in the excited hot and cold subspaces. Thus, using Table \ref{blah} one can calculate $Q_c$ and $Q_h$ as follows
\begin{equation}
    Q_c = \omega_c(B-A)\Delta P,\quad Q_h = \omega_h(D-C) \Delta P,
\end{equation}
where
\begin{eqnarray}
    (k_{12}+k_{14}+k_{32}+k_{34})&=&A,\\
    (k_{21}+k_{23}+k_{41}+k_{43})&=&B,\\
    (k_{13}+k_{14}+k_{23}+k_{24})&=&C,\\
    (k_{31}+k_{32}+k_{41}+k_{42})&=&D.
\end{eqnarray}
as shown in Table \ref{blah}. Thus the efficiency of reduces to
\begin{equation}\label{efficiency_redefined23}
    \eta = 1+\frac{\omega_c(B-A)}{\omega_h(D-C)} = 1-\frac{\omega_c(A-B)}{\omega_h(D-C)},
\end{equation}

\begin{table}[t]
\begin{center}
\begin{tabular}{ |c|c|c|c| } 
 \hline
 Swaps between & Population flow  & Population flow & number  \\
 energy levels  & in the  & in the & of \\
(Convention $|\cdot\rangle_{s,h,c}$) & hot  & cold & swaps\\
 & subspace & cold subspace & \\
 \hline
 $|i,0,0\rangle\leftrightarrow|i+1,0,0\rangle$ & 0 & 0 & $k_{11}$\\ 
 $|i,0,0\rangle\leftrightarrow|i+1,0,1\rangle$ & 0 & $\Delta P$ & $k_{12}$\\
 $|i,0,0\rangle\leftrightarrow|i+1,1,0\rangle$ & $\Delta P$ & 0 & $k_{13}$\\
$|i,0,0\rangle\leftrightarrow|i+1,1,1\rangle$ & $\Delta P$ & $\Delta P$ & $k_{14}$\\
 \hline
 $|i,0,1\rangle\leftrightarrow|i+1,0,0\rangle$ & $0$ & $-\Delta P$ & $k_{21}$ \\ 
 $|i,0,1\rangle\leftrightarrow|i+1,0,1\rangle$ & $0$ & $0$ & $k_{22}$\\
 $|i,0,1\rangle\leftrightarrow|i+1,1,0\rangle$ & $\Delta P$ & $-\Delta P$ & $k_{23}$\\
$|i,0,1\rangle\leftrightarrow|i+1,1,1\rangle$ & $\Delta P$ & $0$ & $k_{24}$\\
\hline
 $|i,1,0\rangle\leftrightarrow|i+1,0,0\rangle$ & $-\Delta P$ & $0$ & $k_{31}$\\ 
 $|i,1,0\rangle\leftrightarrow|i+1,0,1\rangle$ & $-\Delta P$ & $\Delta P$ & $k_{32}$\\
 $|i,1,0\rangle\leftrightarrow|i+1,1,0\rangle$ & $0$ & $0$ & $k_{33}$\\
$|i,1,0\rangle\leftrightarrow|i+1,1,1\rangle$ & $0$ & $\Delta P$ & $k_{34}$\\
\hline
 $|i,1,1\rangle\leftrightarrow|i+1,0,0\rangle$ & $-\Delta P$ & $-\Delta P$ & $k_{41}$\\ 
 $|i,1,1\rangle\leftrightarrow|i+1,0,1\rangle$ & $-\Delta P$ & $0$ & $k_{42}$\\
 $|i,1,1\rangle\leftrightarrow|i+1,1,0\rangle$ & $0$ & $-\Delta P$ & $k_{43}$\\
$|i,1,1\rangle\leftrightarrow|i+1,1,1\rangle$ & $0$ & $0$ & $k_{44}$\\
\hline
\end{tabular}
\caption{The table depicts the amount of net population flow in the hot and cold subspaces due to different swaps. Note that the transferred amount of population is $\Delta P$ for all swaps. This involves fro the condition of preserving the catalyst.}
\label{blah}
\end{center}
\end{table}
Now, $(A-B)$ and $(D-C)$ are integers, and in order to work as an engine the efficiency has to be $1>1-\frac{\beta_h}{\beta_c}\geq\eta>0$ which implies both $(A-B)$ and $(D-C)$ has to be either positive or negative, simultaneously and $|D-C|>|A-B|$. Furthermore, by inspection one can write the following  inequalities,
 
\begin{eqnarray}
    &&d>|D-C|>1\\
    &&d>|A-B|>1.
\end{eqnarray}

where $d$ is the dimension of catalyst. This inequality straightforwardly led to the following bound on the efficiency: 
\begin{equation}
    1-\frac{(A-B)\omega_c}{(D-C)\omega_h} \leq 1-\frac{\omega_c}{d\omega_h},
\end{equation}
where we obtain the inequality by putting the least possible value of $(A-B)$ appeared in the numerator, and the highest possible value of $(D-C)$ appeared in the denominator.

Now, in order to extract positive amount of work this efficiency should be upper bounded by the Carnot efficiency, which implies: 
\begin{eqnarray}
    &&1-\frac{\omega_c}{d\omega_h} \leq 1-\frac{\beta_h}{\beta_c}\nonumber\\&\Rightarrow& d\leq \frac{\beta_c\omega_c}{\beta_h\omega_h}.
\end{eqnarray}
\section{Calculation of work production}\label{Work_prod}
In this appendix, we shall provide the explicit calculation for work production when the initial state of the working body of the catalyst-assisted two-stroke engine transformed via the simple permutation $\Pi_{\text{gen}}$ given in Eq. \eqref{Pigen} and presented diagramatically in Fig. \ref{fig:mhotncold} during work stroke. Employing the identity from Eq. \eqref{Work_redefined}, we can write the amount of work produced by the engine as sum of the heat consumed from the hot bath and heat dumped into the cold bath i.e.,
\begin{equation}
    W = Q_h+Q_c.
\end{equation}
As we have assumed that the ground state energy of both the qubits is zero, $Q_h$ and $Q_c$ are simply given by the difference between the initial and final population in the excited hot and cold subspaces (see Eq. \eqref{defn:hot_heat} and \eqref{defn:cold_heat}). Therefore we can calculate $Q_h$ and $Q_c$ for the simple permutation depicted in \ref{fig:mhotncold} as 
\begin{eqnarray}\label{QcQh}
    Q_h = (m+n)\omega_h\Delta P\quad,\quad Q_h = -n\omega_c\Delta P;
\end{eqnarray}
where $\Delta P$ is the transferred amount of population from $i^{\text{th}}$ block to $(i+1)^{\text{th}}$ block where the addition appeared in $(i+1)$ is addition modulo $d$, where $d$ is the dimension of catalyst. The index $i$ is labelled by the energy eigenstate of Hamiltonian of catalyst i.e., $H_s$.

Note that the transferred amount of population $\Delta P$ has to be fixed for all $i$, in order to preserve the state of the catalyst after the transformation. Therefore, the amount of produced work by the engine is given as
\begin{equation}\label{eq:work_renewed}
    W = \big((m+n)\omega_h-n\omega_c\big)\Delta P.
\end{equation}
Now, our goal is to calculate $\Delta P$ as a function of the inverse temperature of the bath $\beta_h$ and $\beta_c$, energy associated with the excited state of the Hamiltonian of the hot and the cold two-level system  $\omega_h$ and $\omega_c$, and dimension of the catalyst $d=m+n$. In order to do so, we begin by explicitly writing the conditions for preserving the catalyst
\begin{eqnarray}\label{eq:cyclic_app}
    \Delta P &=& \N(p_sa_h-p_{s+1}) \quad \text{for}\; s\in\{1,2,\ldots ,m\},\nonumber\\
    \Delta P &=& \N(p_{m+t}a_h-p_{m+t+1}a_c) \quad \text{for}\; t\in\{1,2,\ldots ,(n-1)\},\nonumber\\
    \Delta P &=& \N(p_{m+n}a_h-p_1a_c).
\end{eqnarray}
where 
\begin{equation}
    \N = \frac{1}{(1+a_h)(1+a_c)}.
\end{equation}
Using this Eq. \eqref{eq:cyclic_app} we can obtain
\begin{eqnarray}
    &&\text{for}\; s\in\{1,2,\ldots ,m\};\quad p_{s+1} = p_sa_h-\frac{\Delta P}{\N},\\
    &&\text{for}\; t\in\{1,2,\ldots ,(n-1)\};\quad p_{m+t+1} = p_{m+t}\frac{a_h}{a_c}-\frac{\Delta P}{\N a_c},\nonumber\\\\
    && p_1 = p_{m+n}\frac{a_h}{a_c}-\frac{\Delta P}{\N a_c}.
\end{eqnarray}
that can be further simplified as
\begin{eqnarray}
    &&\text{for}\; s\in\{1,\ldots ,m\};\; p_{s+1} = p_1a_h^{s}-\frac{\Delta P}{\N}\Big(\frac{1-a_h^s}{1-a_h}\Big)\label{G8}\\
    &&\text{for}\; t\in\{1,2,\ldots ,(n-1)\};\nonumber\\&&\quad p_{m+t+1} = p_{m+1}\bigg(\frac{a_h}{a_c}\bigg)^{t}-\frac{\Delta P}{\N a_c}\Bigg(\frac{(\frac{a_h}{a_c})^t-1}{\frac{a_h}{a_c})-1}\Bigg)\\
    && =\Bigg(p_1a_h^m-\frac{\Delta P}{\N}\Big(\frac{1-a_h^m}{1-a_h}\Big)\Bigg)\Big(\frac{a_h}{a_c}\Big)^t-\frac{\Delta P}{\N a_c}\Bigg(\frac{(\frac{a_h}{a_c})^t-1}{(\frac{a_h}{a_c})-1}\Bigg)\nonumber\\\label{G10}\\
    && p_1 = p_{m+n}\frac{a_h}{a_c}-\frac{\Delta P}{\N a_c}\label{G11}.
\end{eqnarray}
From the normalization of the probability, we can write
\begin{eqnarray}\label{sumequaltoone}
    \sum_{x=2}^{m+n} p_x &=& \sum_{x=2}^{m+1} p_x+\sum_{x=m+2}^{m+n} p_x = \sum_{s=1}^{m} p_{s+1}+\sum_{t=1}^{n-1} p_{m+t+1}\nonumber\\ &=& 1-p_1.
\end{eqnarray}
Substituting the value of $p_x$ from Eq. \eqref{G8} and \eqref{G10} in Eq. \eqref{sumequaltoone}, we obtain $p_1$ in terms of $\Delta P$. In order to get the solution of $\Delta P$, we substituted the obtained expression of $p_1$ in terms of $\Delta P$ in Eq. \eqref{G11} and solved it. This gives 
\begin{equation}
    \Delta P = \frac{1}{f(a_h,a_c,m,n)}\N (a_h^{m+n}-a_c^n),
\end{equation}
with 
 
\begin{widetext}
\begin{equation}\label{eq:Exp_of_f}
    f(a_h,a_c,m,n):=\frac{a_h\left(1-a_c\right)^2\left\{\left(1-a_h^m\right)\left(a_h^n-a_c^n\right)\right\}+\left\{\left(a_h^{(m+n)}-a_c^n\right)\left(a_h-a_c\right)\left(1-a_h\right)\right\}\left\{n(1-a_h)-m(a_h-a_c)\right\}}{\left(a_h-a_c\right)^2\left(1-a_h\right)^2}.
\end{equation} 
\end{widetext}

Thus, employing Eq. \eqref{eq:work_renewed} we get work as 
\begin{equation}\label{work_appendix_final}
    W= \frac{\N\big((m+n)\omega_h-n\omega_c\big)(a_{h}^{m+n}-a_{c}^{n})}{f(a_h,a_c,m,n)}.
\end{equation}
On the other hand, the efficiency of the engine when the working body transforms via simple permutation $\Pi_{\text{gen}}$ in the work stroke given in Eq. \eqref{Pigen} and depicted in Fig. \ref{fig:mhotncold}, can be calculated using Eq. \eqref{QcQh} as follows:
\begin{equation}
    \eta = 1+\frac{Q_c}{Q_h}=1-\Big(\frac{n}{m+n}\Big)\frac{\omega_c}{\omega_h} = 1-\frac{n\omega_c}{d\omega_h},
\end{equation}
where $m+n=d$ is the dimension of the catalyst.

\section{Work extraction as linear program}\label{Work_ext_LP}
As we have mentioned earlier, the maximum amount of the work that can be produced by the two-stroke heat engine is a challenging task. Unlike non-catalytic scenario, we do not know whether the implementation of permutation on the initial state will lead to the optimal work extraction. In this section we shall formulate a linear program whose solution provide an upper bound on the amount of work that can be produced by the engine in presence of a fixed catalyst. In order to do so, we first write the initial state of the working body as 
\begin{equation}
    \rho^{i}_{s,h,c} =\rho_s\otimes\tau_h\otimes\tau_c = \sum_{i=1}^d p_i|i\rangle\langle i|\otimes \tau_h\otimes \tau_c.
\end{equation}
and define the projector 
\begin{equation}
    M_i := |i\rangle\langle i|\otimes\mathbb{I}_h\otimes\mathbb{I}_c.
\end{equation}
For a given total Hamiltonian $H$ of the engine, and the initial state $\rho^{i}_{s,h,c}$, we can express the amount of produced work $W(\cdot)$ as a function of initial state of the working body $\rho^{i}_{s,h,c}$ via the following optimization problem:   
\begin{equation}\label{minima23}
\begin{aligned}
& {\text{maximize}}
& & W(\rho^{i}_{s,h,c}):=\dbra H|\rho^{i}_{s,h,c}\dket-\dbra H|B|\rho^{i}_{s,h,c}\dket \\
& \text{subject to}
& & B_{xy}=|U_{xy}|^2 \quad{\text{for some unitary $U$}}\\
& & &  \forall k\;\;\;\;\dbra M_k|B|\rho^{i}_{s,h,c}\dket = \dbra M_k|\rho^{i}_{s,h,c}\dket   
\end{aligned}
\end{equation}
The solution of the introduced linear program would give optimal work produced by the two-stroke engine having a working body consists of catalyst, the hot and the cold $d$-level system. But, this problem is very hard to solve due to constraint $B_{xy} = |U_{xy}|^2$ is not a linear constraint. Nonetheless, we consider a simplified optimization problem that gives an upper bound for the work extraction, namely
\begin{equation}\label{minima35}
\begin{aligned}
& {\text{maximize}}
& & W(\rho^{i}_{s,h,c}):=\dbra H|\rho^{i}_{s,h,c}\dket-\dbra H|B|\rho^{i}_{s,h,c}\dket \\
& \text{subject to}
& & B \quad \text{is a bistochastic matrix}\\
& & &  \forall k\;\;\;\;\dbra M_k|B|\rho^{i}_{s,h,c}\dket = \dbra M_k|\rho^{i}_{s,h,c}\dket,   
\end{aligned}
\end{equation}
The optimal solution $B^*$ of the introduced problem \eqref{minima35} is a solution of the previous one \eqref{minima23} whenever $B^*$ is an unistochastic matrix (i.e., $B^*_{xy} = |U_{xy}|^2$ for some unitary $U$).
We can write optimal work produced by the engine having a working body consists of catalyst, the hot and the cold $d$-level system as a linear program. In order to do so, we first write the initial state of the working body as 
\begin{equation}
    \rho^{i}_{s,h,c} =\rho_s\otimes\tau_h\otimes\tau_c = \sum_{i=1}^d p_i|i\rangle\langle i|\otimes \tau_h\otimes \tau_c.
\end{equation}
and define the projector 
\begin{equation}
    M_i := |i\rangle\langle i|\otimes\mathbb{I}_h\otimes\mathbb{I}_c.
\end{equation}
For a given total Hamiltonian $H$ of the working body of the engine (i.e., catalyst, the hot and the cold $d$-level system), and the initial state of the working body $\rho^{i}_{s,h,c}$, we can express the amount of produced work $W(\cdot)$ as a function of initial state of the working body $\rho^{i}_{s,h,c}$ via the following problem:   
\begin{equation}\label{minima2}
\begin{aligned}
& {\text{maximize}}
& & W(\rho^{i}_{s,h,c}):=\dbra H|\rho^{i}_{s,h,c}\dket-\dbra H|B|\rho^{i}_{s,h,c}\dket \\
& \text{subject to}
& & B_{xy}=|U_{xy}|^2 \quad{\text{for some unitary $U$}}\\
& & &  \forall k\;\;\;\;\dbra M_k|B|\rho^{i}_{s,h,c}\dket = \dbra M_k|\rho^{i}_{s,h,c}\dket   
\end{aligned}
\end{equation}
The solution of the introduced linear program would give an optimal work produced by the engine having a working body that contains the catalyst. In the following, we consider a simplified optimization problem that gives an upper bound for the work produced by the engine, namely
\begin{equation}\label{minima3}
\begin{aligned}
& {\text{maximize}}
& & W(\rho^{i}_{s,h,c}):=\dbra H|\rho^{i}_{s,h,c}\dket-\dbra H|B|\rho^{i}_{s,h,c}\dket \\
& \text{subject to}
& & B \quad \text{is a bistochastic matrix}\\
& & &  \forall k\;\;\;\;\dbra M_k|B|\rho^{i}_{s,h,c}\dket = \dbra M_k|\rho^{i}_{s,h,c}\dket   
\end{aligned}
\end{equation}
The optimal solution $B^*$ of the introduced problem \eqref{minima3} is a solution of the previous one \eqref{minima2} whenever $B^*$ is an unistochastic matrix (i.e., $B^*_{xy} = |U_{xy}|^2$ for some unitary $U$). 

From now on, let us concentrate on the linear program \eqref{minima3}. $B$ can be decomposed into a convex sum of permutations, such that
\begin{eqnarray}
    B = \sum_{m} \alpha_m \Pi_m, \quad \sum_m \alpha_m = 1, \quad \alpha_m \ge 0.
\end{eqnarray}
Then, we define
\begin{align}
    & w_m = \dbra H|\rho^{i}_{s,h,c}\dket-\dbra H|\Pi_m|\rho^{i}_{s,h,c}\dket, \\
    & a^{(k)}_m = \dbra M_k|\Pi_m|\rho^{i}_{s,h,c}\dket, \\
    & a^{(k)} = \dbra M_k|\rho^{i}_{s,h,c}\dket,
\end{align}
and finally we may rewritten the optimization problem in the following way
\begin{equation}\label{minima4}
\begin{aligned}
& {\text{maximize}}
& & \sum_m \alpha_m w_m \\
& \text{subject to}
& & \sum_m \alpha_m = 1, \quad \alpha_m \ge 0, \\
& & &  \forall k \ \sum_m \alpha_m a^{(k)}_m = a^{(k)}.
\end{aligned}
\end{equation}
Then, the dual problem is given by
\begin{equation}\label{minima5}
\begin{aligned}
& {\text{minimize}}
& & y + \sum_k a^{(k)} x_k \\
& \text{subject to}
& & \forall m \ y \ge w_m - \sum_k a^{(k)}_m x_k
\end{aligned}
\end{equation}
This dual program describes an optimization over the faces of $(d-1)$-polytope, which is defined as:
\begin{eqnarray} \label{polytope1}
y = \max_{m} \left[w_m - \sum_k a^{(k)}_m x_k \right],
\end{eqnarray}
such that the optimal work is given by:
\begin{eqnarray}
    \min_{x_1, \dots, x_{d-1}} \left[\max_{m} \left[w_m - \sum_k a^{(k)}_m x_k\right] + \sum_k a^{(k)} x_k \right]. \nonumber \\
\end{eqnarray}

Let us examine the situation when the solution is uniquely achieved in the vertex of the polytope. Then, the vertex given by the point $\vec x = (y, x_1, x_2, \dots, x_{d-1})^T$ is a solution of the problem \eqref{minima5}, such that 
\begin{eqnarray}
    \vec a^T \vec x = y + \sum_k a^{(k)} x_k
\end{eqnarray}
is minimal on the polytope \eqref{polytope1}, where $\vec a = (1, a^{(1)}, a^{(2)}, \dots, a^{(d-1)})^T$. The vertex $\vec x$ comes from the intersection of $d$ hyperplanes, namely 
\begin{eqnarray} \label{polytope}
y = w_{m_i} - \sum_k a^{(k)}_{m_i} x_k,
\end{eqnarray}
where $i = 1, 2, \dots, d$, such that it is given as a solution of the linear system:
\begin{eqnarray}
    A \vec x = \vec w,
\end{eqnarray}
where $\vec w = (w_1, w_2, \dots, w_n)^T$, and 
\begin{eqnarray}
A = 
    \begin{pmatrix}
        1 & a_{m_1}^{(1)} & a_{m_1}^{(2)} & \dots a_{m_1}^{(d-1)} \\
        1 & a_{m_2}^{(1)} & a_{m_2}^{(2)} & \dots a_{m_2}^{(d-1)} \\
        & & \dots & \\
        1 & a_{m_d}^{(1)} & a_{m_d}^{(2)} & \dots a_{m_{d}}^{(d-1)} 
    \end{pmatrix}.
\end{eqnarray}
Then, the final solution can be written as 
\begin{eqnarray}
    \vec a^T \vec x = \vec a^T A^{-1} \vec w \equiv \vec \alpha^T \vec w,
\end{eqnarray}
where $\vec \alpha = [A^{-1}]^T \vec a$ forms a set of mixing coefficient $\alpha_k$, satisfying all the properties stated in \eqref{minima4}.


\bibliography{main}

\providecommand{\noopsort}[1]{}\providecommand{\singleletter}[1]{#1}%
\begin{thebibliography}{75}%
\makeatletter
\providecommand \@ifxundefined [1]{%
 \@ifx{#1\undefined}
}%
\providecommand \@ifnum [1]{%
 \ifnum #1\expandafter \@firstoftwo
 \else \expandafter \@secondoftwo
 \fi
}%
\providecommand \@ifx [1]{%
 \ifx #1\expandafter \@firstoftwo
 \else \expandafter \@secondoftwo
 \fi
}%
\providecommand \natexlab [1]{#1}%
\providecommand \enquote  [1]{``#1''}%
\providecommand \bibnamefont  [1]{#1}%
\providecommand \bibfnamefont [1]{#1}%
\providecommand \citenamefont [1]{#1}%
\providecommand \href@noop [0]{\@secondoftwo}%
\providecommand \href [0]{\begingroup \@sanitize@url \@href}%
\providecommand \@href[1]{\@@startlink{#1}\@@href}%
\providecommand \@@href[1]{\endgroup#1\@@endlink}%
\providecommand \@sanitize@url [0]{\catcode `\\12\catcode `\$12\catcode
  `\&12\catcode `\#12\catcode `\^12\catcode `\_12\catcode `\%12\relax}%
\providecommand \@@startlink[1]{}%
\providecommand \@@endlink[0]{}%
\providecommand \url  [0]{\begingroup\@sanitize@url \@url }%
\providecommand \@url [1]{\endgroup\@href {#1}{\urlprefix }}%
\providecommand \urlprefix  [0]{URL }%
\providecommand \Eprint [0]{\href }%
\providecommand \doibase [0]{https://doi.org/}%
\providecommand \selectlanguage [0]{\@gobble}%
\providecommand \bibinfo  [0]{\@secondoftwo}%
\providecommand \bibfield  [0]{\@secondoftwo}%
\providecommand \translation [1]{[#1]}%
\providecommand \BibitemOpen [0]{}%
\providecommand \bibitemStop [0]{}%
\providecommand \bibitemNoStop [0]{.\EOS\space}%
\providecommand \EOS [0]{\spacefactor3000\relax}%
\providecommand \BibitemShut  [1]{\csname bibitem#1\endcsname}%
\let\auto@bib@innerbib\@empty
\bibitem [{\citenamefont {{Brand\~ao}}\ \emph {et~al.}(2015)\citenamefont
  {{Brand\~ao}}, \citenamefont {{Horodecki}}, \citenamefont {{Ng}},
  \citenamefont {{Oppenheim}},\ and\ \citenamefont
  {{Wehner}}}]{brandao2015second}%
  \BibitemOpen
  \bibfield  {author} {\bibinfo {author} {\bibfnamefont {F.~G.~S.~L.}\
  \bibnamefont {{Brand\~ao}}}, \bibinfo {author} {\bibfnamefont
  {M.}~\bibnamefont {{Horodecki}}}, \bibinfo {author} {\bibfnamefont
  {N.~H.~Y.}\ \bibnamefont {{Ng}}}, \bibinfo {author} {\bibfnamefont
  {J.}~\bibnamefont {{Oppenheim}}},\ and\ \bibinfo {author} {\bibfnamefont
  {S.}~\bibnamefont {{Wehner}}},\ }\bibfield  {title} {\bibinfo {title} {{The
  second laws of quantum thermodynamics}},\ }\href
  {https://doi.org/10.1073/pnas.1411728112} {\bibfield  {journal} {\bibinfo
  {journal} {Proc. Natl. Acad. Sci. U.S.A.}\ }\textbf {\bibinfo {volume}
  {112}},\ \bibinfo {pages} {3275} (\bibinfo {year} {2015})}\BibitemShut
  {NoStop}%
\bibitem [{\citenamefont {Shiraishi}\ and\ \citenamefont
  {Sagawa}(2021)}]{ShiraishiSagawa21}%
  \BibitemOpen
  \bibfield  {author} {\bibinfo {author} {\bibfnamefont {N.}~\bibnamefont
  {Shiraishi}}\ and\ \bibinfo {author} {\bibfnamefont {T.}~\bibnamefont
  {Sagawa}},\ }\bibfield  {title} {\bibinfo {title} {Quantum thermodynamics of
  correlated-catalytic state conversion at small scale},\ }\href
  {https://doi.org/10.1103/PhysRevLett.126.150502} {\bibfield  {journal}
  {\bibinfo  {journal} {Phys. Rev. Lett.}\ }\textbf {\bibinfo {volume} {126}},\
  \bibinfo {pages} {150502} (\bibinfo {year} {2021})}\BibitemShut {NoStop}%
\bibitem [{\citenamefont {Wilming}(2021)}]{WilmingPRL}%
  \BibitemOpen
  \bibfield  {author} {\bibinfo {author} {\bibfnamefont {H.}~\bibnamefont
  {Wilming}},\ }\bibfield  {title} {\bibinfo {title} {Entropy and reversible
  catalysis},\ }\href {https://doi.org/10.1103/PhysRevLett.127.260402}
  {\bibfield  {journal} {\bibinfo  {journal} {Phys. Rev. Lett.}\ }\textbf
  {\bibinfo {volume} {127}},\ \bibinfo {pages} {260402} (\bibinfo {year}
  {2021})}\BibitemShut {NoStop}%
\bibitem [{\citenamefont {Datta}\ \emph {et~al.}(2023)\citenamefont {Datta},
  \citenamefont {Kondra}, \citenamefont {Miller},\ and\ \citenamefont
  {Streltsov}}]{CDatta_Review}%
  \BibitemOpen
  \bibfield  {author} {\bibinfo {author} {\bibfnamefont {C.}~\bibnamefont
  {Datta}}, \bibinfo {author} {\bibfnamefont {T.~V.}\ \bibnamefont {Kondra}},
  \bibinfo {author} {\bibfnamefont {M.}~\bibnamefont {Miller}},\ and\ \bibinfo
  {author} {\bibfnamefont {A.}~\bibnamefont {Streltsov}},\ }\bibfield  {title}
  {\bibinfo {title} {Catalysis of entanglement and other quantum resources},\
  }\href {https://doi.org/10.1088/1361-6633/acfbec} {\bibfield  {journal}
  {\bibinfo  {journal} {Reports on Progress in Physics}\ }\textbf {\bibinfo
  {volume} {86}},\ \bibinfo {pages} {116002} (\bibinfo {year}
  {2023})}\BibitemShut {NoStop}%
\bibitem [{\citenamefont {Lipka-Bartosik}\ \emph {et~al.}(2023)\citenamefont
  {Lipka-Bartosik}, \citenamefont {Wilming},\ and\ \citenamefont
  {H.~Y.~Ng}}]{Bartosik_review}%
  \BibitemOpen
  \bibfield  {author} {\bibinfo {author} {\bibfnamefont {P.}~\bibnamefont
  {Lipka-Bartosik}}, \bibinfo {author} {\bibfnamefont {H.}~\bibnamefont
  {Wilming}},\ and\ \bibinfo {author} {\bibfnamefont {N.}~\bibnamefont
  {H.~Y.~Ng}},\ }\bibfield  {title} {\bibinfo {title} {Catalysis in quantum
  information theory},\ }\href {https://doi.org/10.48550/arXiv.2306.00798}
  {\bibfield  {journal} {\bibinfo  {journal} {arXiv:2306.00798}\ } (\bibinfo
  {year} {2023})}\BibitemShut {NoStop}%
\bibitem [{\citenamefont {Henao}\ and\ \citenamefont
  {Uzdin}(2023)}]{HenaoUzdin}%
  \BibitemOpen
  \bibfield  {author} {\bibinfo {author} {\bibfnamefont {I.}~\bibnamefont
  {Henao}}\ and\ \bibinfo {author} {\bibfnamefont {R.}~\bibnamefont {Uzdin}},\
  }\bibfield  {title} {\bibinfo {title} {Catalytic leverage of correlations and
  mitigation of dissipation in information erasure},\ }\href
  {https://doi.org/10.1103/PhysRevLett.130.020403} {\bibfield  {journal}
  {\bibinfo  {journal} {Phys. Rev. Lett.}\ }\textbf {\bibinfo {volume} {130}},\
  \bibinfo {pages} {020403} (\bibinfo {year} {2023})}\BibitemShut {NoStop}%
\bibitem [{\citenamefont {Henao}\ and\ \citenamefont
  {Uzdin}(2021)}]{Henao2021catalytic}%
  \BibitemOpen
  \bibfield  {author} {\bibinfo {author} {\bibfnamefont {I.}~\bibnamefont
  {Henao}}\ and\ \bibinfo {author} {\bibfnamefont {R.}~\bibnamefont {Uzdin}},\
  }\bibfield  {title} {\bibinfo {title} {Catalytic transformations with
  finite-size environments: applications to cooling and thermometry},\ }\href
  {https://doi.org/10.22331/q-2021-09-21-547} {\bibfield  {journal} {\bibinfo
  {journal} {{Quantum}}\ }\textbf {\bibinfo {volume} {5}},\ \bibinfo {pages}
  {547} (\bibinfo {year} {2021})}\BibitemShut {NoStop}%
\bibitem [{\citenamefont {Sparaciari}\ \emph {et~al.}(2017)\citenamefont
  {Sparaciari}, \citenamefont {Jennings},\ and\ \citenamefont
  {Oppenheim}}]{Sparaciari2017}%
  \BibitemOpen
  \bibfield  {author} {\bibinfo {author} {\bibfnamefont {C.}~\bibnamefont
  {Sparaciari}}, \bibinfo {author} {\bibfnamefont {D.}~\bibnamefont
  {Jennings}},\ and\ \bibinfo {author} {\bibfnamefont {J.}~\bibnamefont
  {Oppenheim}},\ }\bibfield  {title} {\bibinfo {title} {Energetic instability
  of passive states in thermodynamics},\ }\href
  {https://doi.org/10.1038/s41467-017-01505-4} {\bibfield  {journal} {\bibinfo
  {journal} {Nature Communications}\ }\textbf {\bibinfo {volume} {8}},\
  \bibinfo {pages} {1895} (\bibinfo {year} {2017})}\BibitemShut {NoStop}%
\bibitem [{\citenamefont {Son}\ and\ \citenamefont {Ng}(2024)}]{Son_2024NJP}%
  \BibitemOpen
  \bibfield  {author} {\bibinfo {author} {\bibfnamefont {J.}~\bibnamefont
  {Son}}\ and\ \bibinfo {author} {\bibfnamefont {N.~H.~Y.}\ \bibnamefont
  {Ng}},\ }\bibfield  {title} {\bibinfo {title} {Catalysis in action via
  elementary thermal operations},\ }\href
  {https://doi.org/10.1088/1367-2630/ad2413} {\bibfield  {journal} {\bibinfo
  {journal} {New Journal of Physics}\ }\textbf {\bibinfo {volume} {26}},\
  \bibinfo {pages} {033029} (\bibinfo {year} {2024})}\BibitemShut {NoStop}%
\bibitem [{\citenamefont {Ramsey}(1956)}]{Ramsey22}%
  \BibitemOpen
  \bibfield  {author} {\bibinfo {author} {\bibfnamefont {N.~F.}\ \bibnamefont
  {Ramsey}},\ }\bibfield  {title} {\bibinfo {title} {Thermodynamics and
  statistical mechanics at negative absolute temperatures},\ }\href
  {https://doi.org/10.1103/PhysRev.103.20} {\bibfield  {journal} {\bibinfo
  {journal} {Phys. Rev.}\ }\textbf {\bibinfo {volume} {103}},\ \bibinfo {pages}
  {20} (\bibinfo {year} {1956})}\BibitemShut {NoStop}%
\bibitem [{\citenamefont {Scovil}\ and\ \citenamefont
  {Schulz-DuBois}(1959)}]{Scovil1959}%
  \BibitemOpen
  \bibfield  {author} {\bibinfo {author} {\bibfnamefont {H.~E.~D.}\
  \bibnamefont {Scovil}}\ and\ \bibinfo {author} {\bibfnamefont {E.~O.}\
  \bibnamefont {Schulz-DuBois}},\ }\bibfield  {title} {\bibinfo {title}
  {Three-level masers as heat engines},\ }\href
  {https://doi.org/10.1103/PhysRevLett.2.262} {\bibfield  {journal} {\bibinfo
  {journal} {Phys. Rev. Lett.}\ }\textbf {\bibinfo {volume} {2}},\ \bibinfo
  {pages} {262} (\bibinfo {year} {1959})}\BibitemShut {NoStop}%
\bibitem [{\citenamefont {Geusic}\ \emph {et~al.}(1967)\citenamefont {Geusic},
  \citenamefont {Schulz-DuBios},\ and\ \citenamefont {Scovil}}]{Scovil2}%
  \BibitemOpen
  \bibfield  {author} {\bibinfo {author} {\bibfnamefont {J.~E.}\ \bibnamefont
  {Geusic}}, \bibinfo {author} {\bibfnamefont {E.~O.}\ \bibnamefont
  {Schulz-DuBios}},\ and\ \bibinfo {author} {\bibfnamefont {H.~E.~D.}\
  \bibnamefont {Scovil}},\ }\bibfield  {title} {\bibinfo {title} {Quantum
  equivalent of the carnot cycle},\ }\href
  {https://doi.org/10.1103/PhysRev.156.343} {\bibfield  {journal} {\bibinfo
  {journal} {Phys. Rev.}\ }\textbf {\bibinfo {volume} {156}},\ \bibinfo {pages}
  {343} (\bibinfo {year} {1967})}\BibitemShut {NoStop}%
\bibitem [{\citenamefont {Kosloff}(1984)}]{Kosloff1984}%
  \BibitemOpen
  \bibfield  {author} {\bibinfo {author} {\bibfnamefont {R.}~\bibnamefont
  {Kosloff}},\ }\bibfield  {title} {\bibinfo {title} {A quantum mechanical open
  system as a model of a heat engine},\ }\bibfield  {journal} {\bibinfo
  {journal} {J. Chem. Phys.}\ }\textbf {\bibinfo {volume} {80}},\ \href
  {https://doi.org/10.1063/1.446862} {10.1063/1.446862} (\bibinfo {year}
  {1984})\BibitemShut {NoStop}%
\bibitem [{\citenamefont {Kosloff}\ and\ \citenamefont
  {Levy}(2014{\natexlab{a}})}]{Kosloff2014}%
  \BibitemOpen
  \bibfield  {author} {\bibinfo {author} {\bibfnamefont {R.}~\bibnamefont
  {Kosloff}}\ and\ \bibinfo {author} {\bibfnamefont {A.}~\bibnamefont {Levy}},\
  }\bibfield  {title} {\bibinfo {title} {Quantum heat engines and
  refrigerators: Continuous devices},\ }\href
  {https://doi.org/10.1146/annurev-physchem-040513-103724} {\bibfield
  {journal} {\bibinfo  {journal} {Annual Review of Physical Chemistry}\
  }\textbf {\bibinfo {volume} {65}},\ \bibinfo {pages} {365} (\bibinfo {year}
  {2014}{\natexlab{a}})},\ \bibinfo {note} {pMID: 24689798},\ \Eprint
  {https://arxiv.org/abs/https://doi.org/10.1146/annurev-physchem-040513-103724}
  {https://doi.org/10.1146/annurev-physchem-040513-103724} \BibitemShut
  {NoStop}%
\bibitem [{\citenamefont {Scully}(2002)}]{Scully2002}%
  \BibitemOpen
  \bibfield  {author} {\bibinfo {author} {\bibfnamefont {M.~O.}\ \bibnamefont
  {Scully}},\ }\bibfield  {title} {\bibinfo {title} {Quantum afterburner:
  Improving the efficiency of an ideal heat engine},\ }\href
  {https://doi.org/10.1103/PhysRevLett.88.050602} {\bibfield  {journal}
  {\bibinfo  {journal} {Phys. Rev. Lett.}\ }\textbf {\bibinfo {volume} {88}},\
  \bibinfo {pages} {050602} (\bibinfo {year} {2002})}\BibitemShut {NoStop}%
\bibitem [{\citenamefont {Linden}\ \emph
  {et~al.}(2010{\natexlab{a}})\citenamefont {Linden}, \citenamefont {Popescu},\
  and\ \citenamefont {Skrzypczyk}}]{Popescurefrigerator}%
  \BibitemOpen
  \bibfield  {author} {\bibinfo {author} {\bibfnamefont {N.}~\bibnamefont
  {Linden}}, \bibinfo {author} {\bibfnamefont {S.}~\bibnamefont {Popescu}},\
  and\ \bibinfo {author} {\bibfnamefont {P.}~\bibnamefont {Skrzypczyk}},\
  }\bibfield  {title} {\bibinfo {title} {How small can thermal machines be? the
  smallest possible refrigerator},\ }\href
  {https://doi.org/10.1103/PhysRevLett.105.130401} {\bibfield  {journal}
  {\bibinfo  {journal} {Phys. Rev. Lett.}\ }\textbf {\bibinfo {volume} {105}},\
  \bibinfo {pages} {130401} (\bibinfo {year} {2010}{\natexlab{a}})}\BibitemShut
  {NoStop}%
\bibitem [{\citenamefont {Allahverdyan}\ \emph {et~al.}(2004)\citenamefont
  {Allahverdyan}, \citenamefont {Balian},\ and\ \citenamefont
  {Nieuwenhuizen}}]{Allahverdyan2004}%
  \BibitemOpen
  \bibfield  {author} {\bibinfo {author} {\bibfnamefont {A.~E.}\ \bibnamefont
  {Allahverdyan}}, \bibinfo {author} {\bibfnamefont {R.}~\bibnamefont
  {Balian}},\ and\ \bibinfo {author} {\bibfnamefont {T.~M.}\ \bibnamefont
  {Nieuwenhuizen}},\ }\bibfield  {title} {\bibinfo {title} {Maximal work
  extraction from finite quantum systems},\ }\href
  {https://doi.org/10.1209/epl/i2004-10101-2} {\bibfield  {journal} {\bibinfo
  {journal} {Europhysics Letters ({EPL})}\ }\textbf {\bibinfo {volume} {67}},\
  \bibinfo {pages} {565} (\bibinfo {year} {2004})}\BibitemShut {NoStop}%
\bibitem [{\citenamefont {Segal}\ and\ \citenamefont {Nitzan}(2006)}]{Segal}%
  \BibitemOpen
  \bibfield  {author} {\bibinfo {author} {\bibfnamefont {D.}~\bibnamefont
  {Segal}}\ and\ \bibinfo {author} {\bibfnamefont {A.}~\bibnamefont {Nitzan}},\
  }\bibfield  {title} {\bibinfo {title} {Molecular heat pump},\ }\href
  {https://doi.org/10.1103/PhysRevE.73.026109} {\bibfield  {journal} {\bibinfo
  {journal} {Phys. Rev. E}\ }\textbf {\bibinfo {volume} {73}},\ \bibinfo
  {pages} {026109} (\bibinfo {year} {2006})}\BibitemShut {NoStop}%
\bibitem [{\citenamefont {Henrich}\ \emph {et~al.}(2007)\citenamefont
  {Henrich}, \citenamefont {Mahler},\ and\ \citenamefont {Michel}}]{Mahler}%
  \BibitemOpen
  \bibfield  {author} {\bibinfo {author} {\bibfnamefont {M.~J.}\ \bibnamefont
  {Henrich}}, \bibinfo {author} {\bibfnamefont {G.}~\bibnamefont {Mahler}},\
  and\ \bibinfo {author} {\bibfnamefont {M.}~\bibnamefont {Michel}},\
  }\bibfield  {title} {\bibinfo {title} {Driven spin systems as quantum
  thermodynamic machines: Fundamental limits},\ }\href
  {https://doi.org/10.1103/PhysRevE.75.051118} {\bibfield  {journal} {\bibinfo
  {journal} {Phys. Rev. E}\ }\textbf {\bibinfo {volume} {75}},\ \bibinfo
  {pages} {051118} (\bibinfo {year} {2007})}\BibitemShut {NoStop}%
\bibitem [{\citenamefont {Linden}\ \emph
  {et~al.}(2010{\natexlab{b}})\citenamefont {Linden}, \citenamefont {Popescu},\
  and\ \citenamefont {Skrzypczyk}}]{PopescuSmallestHeatEngine2010}%
  \BibitemOpen
  \bibfield  {author} {\bibinfo {author} {\bibfnamefont {N.}~\bibnamefont
  {Linden}}, \bibinfo {author} {\bibfnamefont {S.}~\bibnamefont {Popescu}},\
  and\ \bibinfo {author} {\bibfnamefont {P.}~\bibnamefont {Skrzypczyk}},\
  }\bibfield  {title} {\bibinfo {title} {The smallest possible heat engines},\
  }\href {https://doi.org/10.48550/arXiv.1010.6029} {\bibfield  {journal}
  {\bibinfo  {journal} {arXiv:1010.6029}\ } (\bibinfo {year}
  {2010}{\natexlab{b}})}\BibitemShut {NoStop}%
\bibitem [{\citenamefont {Brunner}\ \emph {et~al.}(2012)\citenamefont
  {Brunner}, \citenamefont {Linden}, \citenamefont {Popescu},\ and\
  \citenamefont {Skrzypczyk}}]{BrunnerVirtualqubit2012}%
  \BibitemOpen
  \bibfield  {author} {\bibinfo {author} {\bibfnamefont {N.}~\bibnamefont
  {Brunner}}, \bibinfo {author} {\bibfnamefont {N.}~\bibnamefont {Linden}},
  \bibinfo {author} {\bibfnamefont {S.}~\bibnamefont {Popescu}},\ and\ \bibinfo
  {author} {\bibfnamefont {P.}~\bibnamefont {Skrzypczyk}},\ }\bibfield  {title}
  {\bibinfo {title} {Virtual qubits, virtual temperatures, and the foundations
  of thermodynamics},\ }\href {https://doi.org/10.1103/PhysRevE.85.051117}
  {\bibfield  {journal} {\bibinfo  {journal} {Phys. Rev. E}\ }\textbf {\bibinfo
  {volume} {85}},\ \bibinfo {pages} {051117} (\bibinfo {year}
  {2012})}\BibitemShut {NoStop}%
\bibitem [{\citenamefont {Alicki}\ \emph {et~al.}(2004)\citenamefont {Alicki},
  \citenamefont {Horodecki}, \citenamefont {Horodecki},\ and\ \citenamefont
  {Horodecki}}]{Alicki2004}%
  \BibitemOpen
  \bibfield  {author} {\bibinfo {author} {\bibfnamefont {R.}~\bibnamefont
  {Alicki}}, \bibinfo {author} {\bibfnamefont {M.}~\bibnamefont {Horodecki}},
  \bibinfo {author} {\bibfnamefont {P.}~\bibnamefont {Horodecki}},\ and\
  \bibinfo {author} {\bibfnamefont {R.}~\bibnamefont {Horodecki}},\ }\bibfield
  {title} {\bibinfo {title} {Thermodynamics of quantum information systems ---
  hamiltonian description},\ }\href
  {https://doi.org/10.1023/B:OPSY.0000047566.72717.71} {\bibfield  {journal}
  {\bibinfo  {journal} {Open Syst. Inf. Dyn.}\ }\textbf {\bibinfo {volume}
  {11}},\ \bibinfo {pages} {205} (\bibinfo {year} {2004})}\BibitemShut
  {NoStop}%
\bibitem [{\citenamefont {Zhang}\ \emph {et~al.}(2022)\citenamefont {Zhang},
  \citenamefont {Zhang}, \citenamefont {Ding}, \citenamefont {Li},
  \citenamefont {Bu}, \citenamefont {Wang}, \citenamefont {Yan}, \citenamefont
  {Su}, \citenamefont {Chen}, \citenamefont {Nori}, \citenamefont
  {{\"O}zdemir}, \citenamefont {Zhou}, \citenamefont {Jing},\ and\
  \citenamefont {Feng}}]{Zhang2022}%
  \BibitemOpen
  \bibfield  {author} {\bibinfo {author} {\bibfnamefont {J.-W.}\ \bibnamefont
  {Zhang}}, \bibinfo {author} {\bibfnamefont {J.-Q.}\ \bibnamefont {Zhang}},
  \bibinfo {author} {\bibfnamefont {G.-Y.}\ \bibnamefont {Ding}}, \bibinfo
  {author} {\bibfnamefont {J.-C.}\ \bibnamefont {Li}}, \bibinfo {author}
  {\bibfnamefont {J.-T.}\ \bibnamefont {Bu}}, \bibinfo {author} {\bibfnamefont
  {B.}~\bibnamefont {Wang}}, \bibinfo {author} {\bibfnamefont {L.-L.}\
  \bibnamefont {Yan}}, \bibinfo {author} {\bibfnamefont {S.-L.}\ \bibnamefont
  {Su}}, \bibinfo {author} {\bibfnamefont {L.}~\bibnamefont {Chen}}, \bibinfo
  {author} {\bibfnamefont {F.}~\bibnamefont {Nori}}, \bibinfo {author}
  {\bibfnamefont {{\c{S}}.~K.}\ \bibnamefont {{\"O}zdemir}}, \bibinfo {author}
  {\bibfnamefont {F.}~\bibnamefont {Zhou}}, \bibinfo {author} {\bibfnamefont
  {H.}~\bibnamefont {Jing}},\ and\ \bibinfo {author} {\bibfnamefont
  {M.}~\bibnamefont {Feng}},\ }\bibfield  {title} {\bibinfo {title} {Dynamical
  control of quantum heat engines using exceptional points},\ }\href
  {https://doi.org/10.1038/s41467-022-33667-1} {\bibfield  {journal} {\bibinfo
  {journal} {Nature Communications}\ }\textbf {\bibinfo {volume} {13}},\
  \bibinfo {pages} {6225} (\bibinfo {year} {2022})}\BibitemShut {NoStop}%
\bibitem [{\citenamefont {Myers}\ \emph {et~al.}(2022)\citenamefont {Myers},
  \citenamefont {Abah},\ and\ \citenamefont {Deffner}}]{Myers2022}%
  \BibitemOpen
  \bibfield  {author} {\bibinfo {author} {\bibfnamefont {N.~M.}\ \bibnamefont
  {Myers}}, \bibinfo {author} {\bibfnamefont {O.}~\bibnamefont {Abah}},\ and\
  \bibinfo {author} {\bibfnamefont {S.}~\bibnamefont {Deffner}},\ }\bibfield
  {title} {\bibinfo {title} {{Quantum thermodynamic devices: From theoretical
  proposals to experimental reality}},\ }\href
  {https://doi.org/10.1116/5.0083192} {\bibfield  {journal} {\bibinfo
  {journal} {AVS Quantum Science}\ }\textbf {\bibinfo {volume} {4}},\ \bibinfo
  {pages} {027101} (\bibinfo {year} {2022})},\ \Eprint
  {https://arxiv.org/abs/https://pubs.aip.org/avs/aqs/article-pdf/doi/10.1116/5.0083192/16494008/027101\_1\_online.pdf}
  {https://pubs.aip.org/avs/aqs/article-pdf/doi/10.1116/5.0083192/16494008/027101\_1\_online.pdf}
  \BibitemShut {NoStop}%
\bibitem [{\citenamefont {{Gelbwaser-Klimovsky, D.}}\ \emph
  {et~al.}(2013)\citenamefont {{Gelbwaser-Klimovsky, D.}}, \citenamefont
  {{Alicki, R.}},\ and\ \citenamefont {{Kurizki, G.}}}]{Klimkovsky2013}%
  \BibitemOpen
  \bibfield  {author} {\bibinfo {author} {\bibnamefont {{Gelbwaser-Klimovsky,
  D.}}}, \bibinfo {author} {\bibnamefont {{Alicki, R.}}},\ and\ \bibinfo
  {author} {\bibnamefont {{Kurizki, G.}}},\ }\bibfield  {title} {\bibinfo
  {title} {Work and energy gain of heat-pumped quantized amplifiers},\ }\href
  {https://doi.org/10.1209/0295-5075/103/60005} {\bibfield  {journal} {\bibinfo
   {journal} {EPL}\ }\textbf {\bibinfo {volume} {103}},\ \bibinfo {pages}
  {60005} (\bibinfo {year} {2013})}\BibitemShut {NoStop}%
\bibitem [{\citenamefont {Kosloff}\ and\ \citenamefont
  {Levy}(2014{\natexlab{b}})}]{KosloffLevy2014}%
  \BibitemOpen
  \bibfield  {author} {\bibinfo {author} {\bibfnamefont {R.}~\bibnamefont
  {Kosloff}}\ and\ \bibinfo {author} {\bibfnamefont {A.}~\bibnamefont {Levy}},\
  }\bibfield  {title} {\bibinfo {title} {Quantum heat engines and
  refrigerators: Continuous devices},\ }\href
  {https://doi.org/10.1146/annurev-physchem-040513-103724} {\bibfield
  {journal} {\bibinfo  {journal} {Annual Review of Physical Chemistry}\
  }\textbf {\bibinfo {volume} {65}},\ \bibinfo {pages} {365} (\bibinfo {year}
  {2014}{\natexlab{b}})}\BibitemShut {NoStop}%
\bibitem [{\citenamefont
  {Popescu}(2010)}]{PopescuSmallestHeatEngineprinciple2010}%
  \BibitemOpen
  \bibfield  {author} {\bibinfo {author} {\bibfnamefont {S.}~\bibnamefont
  {Popescu}},\ }\bibfield  {title} {\bibinfo {title} {Maximally efficient
  quantum thermal machines: The basic principles},\ }\href
  {https://arxiv.org/pdf/1009.2536.pdf} {\bibfield  {journal} {\bibinfo
  {journal} {arXiv:1009.2536}\ } (\bibinfo {year} {2010})}\BibitemShut
  {NoStop}%
\bibitem [{\citenamefont {Skrzypczyk}\ \emph {et~al.}(2011)\citenamefont
  {Skrzypczyk}, \citenamefont {Brunner}, \citenamefont {Linden},\ and\
  \citenamefont {Popescu}}]{Skrzypczyk_2011}%
  \BibitemOpen
  \bibfield  {author} {\bibinfo {author} {\bibfnamefont {P.}~\bibnamefont
  {Skrzypczyk}}, \bibinfo {author} {\bibfnamefont {N.}~\bibnamefont {Brunner}},
  \bibinfo {author} {\bibfnamefont {N.}~\bibnamefont {Linden}},\ and\ \bibinfo
  {author} {\bibfnamefont {S.}~\bibnamefont {Popescu}},\ }\bibfield  {title}
  {\bibinfo {title} {The smallest refrigerators can reach maximal efficiency},\
  }\href {https://doi.org/10.1088/1751-8113/44/49/492002} {\bibfield  {journal}
  {\bibinfo  {journal} {Journal of Physics A: Mathematical and Theoretical}\
  }\textbf {\bibinfo {volume} {44}},\ \bibinfo {pages} {492002} (\bibinfo
  {year} {2011})}\BibitemShut {NoStop}%
\bibitem [{\citenamefont {Uzdin}\ \emph {et~al.}(2015)\citenamefont {Uzdin},
  \citenamefont {Levy},\ and\ \citenamefont {Kosloff}}]{UzdinLevyKosloff2015}%
  \BibitemOpen
  \bibfield  {author} {\bibinfo {author} {\bibfnamefont {R.}~\bibnamefont
  {Uzdin}}, \bibinfo {author} {\bibfnamefont {A.}~\bibnamefont {Levy}},\ and\
  \bibinfo {author} {\bibfnamefont {R.}~\bibnamefont {Kosloff}},\ }\bibfield
  {title} {\bibinfo {title} {Equivalence of quantum heat machines, and
  quantum-thermodynamic signatures},\ }\href
  {https://doi.org/10.1103/PhysRevX.5.031044} {\bibfield  {journal} {\bibinfo
  {journal} {Phys. Rev. X}\ }\textbf {\bibinfo {volume} {5}},\ \bibinfo {pages}
  {031044} (\bibinfo {year} {2015})}\BibitemShut {NoStop}%
\bibitem [{\citenamefont {Mitchison}(2019)}]{MitchisonContemp}%
  \BibitemOpen
  \bibfield  {author} {\bibinfo {author} {\bibfnamefont {M.~T.}\ \bibnamefont
  {Mitchison}},\ }\bibfield  {title} {\bibinfo {title} {Quantum thermal
  absorption machines: refrigerators, engines and clocks},\ }\href
  {https://doi.org/10.1080/00107514.2019.1631555} {\bibfield  {journal}
  {\bibinfo  {journal} {Contemporary Physics}\ }\textbf {\bibinfo {volume}
  {60}},\ \bibinfo {pages} {164} (\bibinfo {year} {2019})},\ \Eprint
  {https://arxiv.org/abs/https://doi.org/10.1080/00107514.2019.1631555}
  {https://doi.org/10.1080/00107514.2019.1631555} \BibitemShut {NoStop}%
\bibitem [{\citenamefont {Woods}\ \emph {et~al.}(2019)\citenamefont {Woods},
  \citenamefont {Ng},\ and\ \citenamefont
  {Wehner}}]{Woods2019maximumefficiencyof}%
  \BibitemOpen
  \bibfield  {author} {\bibinfo {author} {\bibfnamefont {M.~P.}\ \bibnamefont
  {Woods}}, \bibinfo {author} {\bibfnamefont {N.~H.~Y.}\ \bibnamefont {Ng}},\
  and\ \bibinfo {author} {\bibfnamefont {S.}~\bibnamefont {Wehner}},\
  }\bibfield  {title} {\bibinfo {title} {The maximum efficiency of nano heat
  engines depends on more than temperature},\ }\href
  {https://doi.org/10.22331/q-2019-08-19-177} {\bibfield  {journal} {\bibinfo
  {journal} {{Quantum}}\ }\textbf {\bibinfo {volume} {3}},\ \bibinfo {pages}
  {177} (\bibinfo {year} {2019})}\BibitemShut {NoStop}%
\bibitem [{\citenamefont {Hofer}\ \emph {et~al.}(2017)\citenamefont {Hofer},
  \citenamefont {Perarnau-Llobet}, \citenamefont {Miranda}, \citenamefont
  {Haack}, \citenamefont {Silva}, \citenamefont {Brask},\ and\ \citenamefont
  {Brunner}}]{Hofer_2017}%
  \BibitemOpen
  \bibfield  {author} {\bibinfo {author} {\bibfnamefont {P.~P.}\ \bibnamefont
  {Hofer}}, \bibinfo {author} {\bibfnamefont {M.}~\bibnamefont
  {Perarnau-Llobet}}, \bibinfo {author} {\bibfnamefont {L.~D.~M.}\ \bibnamefont
  {Miranda}}, \bibinfo {author} {\bibfnamefont {G.}~\bibnamefont {Haack}},
  \bibinfo {author} {\bibfnamefont {R.}~\bibnamefont {Silva}}, \bibinfo
  {author} {\bibfnamefont {J.~B.}\ \bibnamefont {Brask}},\ and\ \bibinfo
  {author} {\bibfnamefont {N.}~\bibnamefont {Brunner}},\ }\bibfield  {title}
  {\bibinfo {title} {Markovian master equations for quantum thermal machines:
  local versus global approach},\ }\href
  {https://doi.org/10.1088/1367-2630/aa964f} {\bibfield  {journal} {\bibinfo
  {journal} {New Journal of Physics}\ }\textbf {\bibinfo {volume} {19}},\
  \bibinfo {pages} {123037} (\bibinfo {year} {2017})}\BibitemShut {NoStop}%
\bibitem [{\citenamefont {Brask}\ \emph {et~al.}(2015)\citenamefont {Brask},
  \citenamefont {Haack}, \citenamefont {Brunner},\ and\ \citenamefont
  {Huber}}]{Bohr_Brask_2015}%
  \BibitemOpen
  \bibfield  {author} {\bibinfo {author} {\bibfnamefont {J.~B.}\ \bibnamefont
  {Brask}}, \bibinfo {author} {\bibfnamefont {G.}~\bibnamefont {Haack}},
  \bibinfo {author} {\bibfnamefont {N.}~\bibnamefont {Brunner}},\ and\ \bibinfo
  {author} {\bibfnamefont {M.}~\bibnamefont {Huber}},\ }\bibfield  {title}
  {\bibinfo {title} {Autonomous quantum thermal machine for generating
  steady-state entanglement},\ }\href
  {https://doi.org/10.1088/1367-2630/17/11/113029} {\bibfield  {journal}
  {\bibinfo  {journal} {New Journal of Physics}\ }\textbf {\bibinfo {volume}
  {17}},\ \bibinfo {pages} {113029} (\bibinfo {year} {2015})}\BibitemShut
  {NoStop}%
\bibitem [{\citenamefont {Taranto}\ \emph {et~al.}(2023)\citenamefont
  {Taranto}, \citenamefont {Bakhshinezhad}, \citenamefont {Bluhm},
  \citenamefont {Silva}, \citenamefont {Friis}, \citenamefont {Lock},
  \citenamefont {Vitagliano}, \citenamefont {Binder}, \citenamefont {Debarba},
  \citenamefont {Schwarzhans}, \citenamefont {Clivaz},\ and\ \citenamefont
  {Huber}}]{Landauer_vs_Nernst}%
  \BibitemOpen
  \bibfield  {author} {\bibinfo {author} {\bibfnamefont {P.}~\bibnamefont
  {Taranto}}, \bibinfo {author} {\bibfnamefont {F.}~\bibnamefont
  {Bakhshinezhad}}, \bibinfo {author} {\bibfnamefont {A.}~\bibnamefont
  {Bluhm}}, \bibinfo {author} {\bibfnamefont {R.}~\bibnamefont {Silva}},
  \bibinfo {author} {\bibfnamefont {N.}~\bibnamefont {Friis}}, \bibinfo
  {author} {\bibfnamefont {M.~P.}\ \bibnamefont {Lock}}, \bibinfo {author}
  {\bibfnamefont {G.}~\bibnamefont {Vitagliano}}, \bibinfo {author}
  {\bibfnamefont {F.~C.}\ \bibnamefont {Binder}}, \bibinfo {author}
  {\bibfnamefont {T.}~\bibnamefont {Debarba}}, \bibinfo {author} {\bibfnamefont
  {E.}~\bibnamefont {Schwarzhans}}, \bibinfo {author} {\bibfnamefont
  {F.}~\bibnamefont {Clivaz}},\ and\ \bibinfo {author} {\bibfnamefont
  {M.}~\bibnamefont {Huber}},\ }\bibfield  {title} {\bibinfo {title} {Landauer
  versus nernst: What is the true cost of cooling a quantum system?},\ }\href
  {https://doi.org/10.1103/PRXQuantum.4.010332} {\bibfield  {journal} {\bibinfo
   {journal} {PRX Quantum}\ }\textbf {\bibinfo {volume} {4}},\ \bibinfo {pages}
  {010332} (\bibinfo {year} {2023})}\BibitemShut {NoStop}%
\bibitem [{\citenamefont {Pusz}\ and\ \citenamefont
  {Woronowicz}(1978)}]{Pusz1978}%
  \BibitemOpen
  \bibfield  {author} {\bibinfo {author} {\bibfnamefont {W.}~\bibnamefont
  {Pusz}}\ and\ \bibinfo {author} {\bibfnamefont {S.~L.}\ \bibnamefont
  {Woronowicz}},\ }\bibfield  {title} {\bibinfo {title} {Passive states and kms
  states for general quantum systems},\ }\href
  {https://doi.org/https://doi.org/10.1007/BF01614224} {\bibfield  {journal}
  {\bibinfo  {journal} {Comm. Math. Phys.}\ }\textbf {\bibinfo {volume} {58}},\
  \bibinfo {pages} {273} (\bibinfo {year} {1978})}\BibitemShut {NoStop}%
\bibitem [{\citenamefont {Alicki}(1979)}]{Alicki_1979}%
  \BibitemOpen
  \bibfield  {author} {\bibinfo {author} {\bibfnamefont {R.}~\bibnamefont
  {Alicki}},\ }\bibfield  {title} {\bibinfo {title} {The quantum open system as
  a model of the heat engine},\ }\href
  {https://doi.org/10.1088/0305-4470/12/5/007} {\bibfield  {journal} {\bibinfo
  {journal} {J. Phys. A: Math. Gen.}\ }\textbf {\bibinfo {volume} {12}},\
  \bibinfo {pages} {L103} (\bibinfo {year} {1979})}\BibitemShut {NoStop}%
\bibitem [{\citenamefont {Davies}(1978)}]{Davies_1978}%
  \BibitemOpen
  \bibfield  {author} {\bibinfo {author} {\bibfnamefont {P.~C.~W.}\
  \bibnamefont {Davies}},\ }\bibfield  {title} {\bibinfo {title}
  {Thermodynamics of black holes},\ }\href
  {https://doi.org/10.1088/0034-4885/41/8/004} {\bibfield  {journal} {\bibinfo
  {journal} {Reports on Progress in Physics}\ }\textbf {\bibinfo {volume}
  {41}},\ \bibinfo {pages} {1313} (\bibinfo {year} {1978})}\BibitemShut
  {NoStop}%
\bibitem [{\citenamefont {Esposito}\ \emph {et~al.}(2009)\citenamefont
  {Esposito}, \citenamefont {Harbola},\ and\ \citenamefont
  {Mukamel}}]{Esposito2009}%
  \BibitemOpen
  \bibfield  {author} {\bibinfo {author} {\bibfnamefont {M.}~\bibnamefont
  {Esposito}}, \bibinfo {author} {\bibfnamefont {U.}~\bibnamefont {Harbola}},\
  and\ \bibinfo {author} {\bibfnamefont {S.}~\bibnamefont {Mukamel}},\
  }\bibfield  {title} {\bibinfo {title} {Nonequilibrium fluctuations,
  fluctuation theorems, and counting statistics in quantum systems},\ }\href
  {https://doi.org/10.1103/RevModPhys.81.1665} {\bibfield  {journal} {\bibinfo
  {journal} {Rev. Mod. Phys.}\ }\textbf {\bibinfo {volume} {81}},\ \bibinfo
  {pages} {1665} (\bibinfo {year} {2009})}\BibitemShut {NoStop}%
\bibitem [{\citenamefont {{Horodecki}}\ and\ \citenamefont
  {{Oppenheim}}(2013)}]{horodecki2013fundamental}%
  \BibitemOpen
  \bibfield  {author} {\bibinfo {author} {\bibfnamefont {M.}~\bibnamefont
  {{Horodecki}}}\ and\ \bibinfo {author} {\bibfnamefont {J.}~\bibnamefont
  {{Oppenheim}}},\ }\bibfield  {title} {\bibinfo {title} {{Fundamental
  limitations for quantum and nanoscale thermodynamics}},\ }\href
  {https://www.nature.com/articles/ncomms3059} {\bibfield  {journal} {\bibinfo
  {journal} {Nat. Commun.}\ }\textbf {\bibinfo {volume} {4}},\ \bibinfo {eid}
  {2059} (\bibinfo {year} {2013})}\BibitemShut {NoStop}%
\bibitem [{\citenamefont {Skrzypczyk}\ \emph {et~al.}(2014)\citenamefont
  {Skrzypczyk}, \citenamefont {Short},\ and\ \citenamefont
  {Popescu}}]{Skrzypczyk2014}%
  \BibitemOpen
  \bibfield  {author} {\bibinfo {author} {\bibfnamefont {P.}~\bibnamefont
  {Skrzypczyk}}, \bibinfo {author} {\bibfnamefont {A.~J.}\ \bibnamefont
  {Short}},\ and\ \bibinfo {author} {\bibfnamefont {S.}~\bibnamefont
  {Popescu}},\ }\bibfield  {title} {\bibinfo {title} {Work extraction and
  thermodynamics for individual quantum systems},\ }\href
  {https://doi.org/10.1038/ncomms5185} {\bibfield  {journal} {\bibinfo
  {journal} {Nat. Commun.}\ }\textbf {\bibinfo {volume} {5}},\ \bibinfo {pages}
  {4185} (\bibinfo {year} {2014})}\BibitemShut {NoStop}%
\bibitem [{\citenamefont {Blickle}\ and\ \citenamefont
  {Bechinger}(2011)}]{Blickle2011}%
  \BibitemOpen
  \bibfield  {author} {\bibinfo {author} {\bibfnamefont {V.}~\bibnamefont
  {Blickle}}\ and\ \bibinfo {author} {\bibfnamefont {C.}~\bibnamefont
  {Bechinger}},\ }\bibfield  {title} {\bibinfo {title} {Realization of a
  micrometre-sized stochastic heat engine},\ }\href
  {https://api.semanticscholar.org/CorpusID:29057092} {\bibfield  {journal}
  {\bibinfo  {journal} {Nature Physics}\ }\textbf {\bibinfo {volume} {8}},\
  \bibinfo {pages} {143 } (\bibinfo {year} {2011})}\BibitemShut {NoStop}%
\bibitem [{\citenamefont {Abah}\ \emph {et~al.}(2012)\citenamefont {Abah},
  \citenamefont {Ro\ss{}nagel}, \citenamefont {Jacob}, \citenamefont {Deffner},
  \citenamefont {Schmidt-Kaler}, \citenamefont {Singer},\ and\ \citenamefont
  {Lutz}}]{Abah_2012}%
  \BibitemOpen
  \bibfield  {author} {\bibinfo {author} {\bibfnamefont {O.}~\bibnamefont
  {Abah}}, \bibinfo {author} {\bibfnamefont {J.}~\bibnamefont {Ro\ss{}nagel}},
  \bibinfo {author} {\bibfnamefont {G.}~\bibnamefont {Jacob}}, \bibinfo
  {author} {\bibfnamefont {S.}~\bibnamefont {Deffner}}, \bibinfo {author}
  {\bibfnamefont {F.}~\bibnamefont {Schmidt-Kaler}}, \bibinfo {author}
  {\bibfnamefont {K.}~\bibnamefont {Singer}},\ and\ \bibinfo {author}
  {\bibfnamefont {E.}~\bibnamefont {Lutz}},\ }\bibfield  {title} {\bibinfo
  {title} {Single-ion heat engine at maximum power},\ }\href
  {https://doi.org/10.1103/PhysRevLett.109.203006} {\bibfield  {journal}
  {\bibinfo  {journal} {Phys. Rev. Lett.}\ }\textbf {\bibinfo {volume} {109}},\
  \bibinfo {pages} {203006} (\bibinfo {year} {2012})}\BibitemShut {NoStop}%
\bibitem [{\citenamefont {{Ro{\ss}nagel}}\ \emph {et~al.}(2016)\citenamefont
  {{Ro{\ss}nagel}}, \citenamefont {{Dawkins}}, \citenamefont {{Tolazzi}},
  \citenamefont {{Abah}}, \citenamefont {{Lutz}}, \citenamefont
  {{Schmidt-Kaler}},\ and\ \citenamefont {{Singer}}}]{Rossnagel2016}%
  \BibitemOpen
  \bibfield  {author} {\bibinfo {author} {\bibfnamefont {J.}~\bibnamefont
  {{Ro{\ss}nagel}}}, \bibinfo {author} {\bibfnamefont {S.~T.}\ \bibnamefont
  {{Dawkins}}}, \bibinfo {author} {\bibfnamefont {K.~N.}\ \bibnamefont
  {{Tolazzi}}}, \bibinfo {author} {\bibfnamefont {O.}~\bibnamefont {{Abah}}},
  \bibinfo {author} {\bibfnamefont {E.}~\bibnamefont {{Lutz}}}, \bibinfo
  {author} {\bibfnamefont {F.}~\bibnamefont {{Schmidt-Kaler}}},\ and\ \bibinfo
  {author} {\bibfnamefont {K.}~\bibnamefont {{Singer}}},\ }\bibfield  {title}
  {\bibinfo {title} {{A single-atom heat engine}},\ }\href
  {https://doi.org/10.1126/science.aad6320} {\bibfield  {journal} {\bibinfo
  {journal} {Science}\ }\textbf {\bibinfo {volume} {352}},\ \bibinfo {pages}
  {325} (\bibinfo {year} {2016})},\ \Eprint {https://arxiv.org/abs/1510.03681}
  {arXiv:1510.03681 [cond-mat.stat-mech]} \BibitemShut {NoStop}%
\bibitem [{\citenamefont {von Lindenfels}\ \emph {et~al.}(2019)\citenamefont
  {von Lindenfels}, \citenamefont {Gr\"ab}, \citenamefont {Schmiegelow},
  \citenamefont {Kaushal}, \citenamefont {Schulz}, \citenamefont {Mitchison},
  \citenamefont {Goold}, \citenamefont {Schmidt-Kaler},\ and\ \citenamefont
  {Poschinger}}]{Lindenfels_2019}%
  \BibitemOpen
  \bibfield  {author} {\bibinfo {author} {\bibfnamefont {D.}~\bibnamefont {von
  Lindenfels}}, \bibinfo {author} {\bibfnamefont {O.}~\bibnamefont {Gr\"ab}},
  \bibinfo {author} {\bibfnamefont {C.~T.}\ \bibnamefont {Schmiegelow}},
  \bibinfo {author} {\bibfnamefont {V.}~\bibnamefont {Kaushal}}, \bibinfo
  {author} {\bibfnamefont {J.}~\bibnamefont {Schulz}}, \bibinfo {author}
  {\bibfnamefont {M.~T.}\ \bibnamefont {Mitchison}}, \bibinfo {author}
  {\bibfnamefont {J.}~\bibnamefont {Goold}}, \bibinfo {author} {\bibfnamefont
  {F.}~\bibnamefont {Schmidt-Kaler}},\ and\ \bibinfo {author} {\bibfnamefont
  {U.~G.}\ \bibnamefont {Poschinger}},\ }\bibfield  {title} {\bibinfo {title}
  {Spin heat engine coupled to a harmonic-oscillator flywheel},\ }\href
  {https://doi.org/10.1103/PhysRevLett.123.080602} {\bibfield  {journal}
  {\bibinfo  {journal} {Phys. Rev. Lett.}\ }\textbf {\bibinfo {volume} {123}},\
  \bibinfo {pages} {080602} (\bibinfo {year} {2019})}\BibitemShut {NoStop}%
\bibitem [{\citenamefont {Maslennikov}\ \emph {et~al.}(2019)\citenamefont
  {Maslennikov}, \citenamefont {Ding}, \citenamefont {Hablützel},
  \citenamefont {Gan}, \citenamefont {Roulet}, \citenamefont {Nimmrichter},
  \citenamefont {Dai}, \citenamefont {Scarani},\ and\ \citenamefont
  {Matsukevich}}]{Maslennikov_2019}%
  \BibitemOpen
  \bibfield  {author} {\bibinfo {author} {\bibfnamefont {G.}~\bibnamefont
  {Maslennikov}}, \bibinfo {author} {\bibfnamefont {S.}~\bibnamefont {Ding}},
  \bibinfo {author} {\bibfnamefont {R.}~\bibnamefont {Hablützel}}, \bibinfo
  {author} {\bibfnamefont {J.}~\bibnamefont {Gan}}, \bibinfo {author}
  {\bibfnamefont {A.}~\bibnamefont {Roulet}}, \bibinfo {author} {\bibfnamefont
  {S.}~\bibnamefont {Nimmrichter}}, \bibinfo {author} {\bibfnamefont
  {J.}~\bibnamefont {Dai}}, \bibinfo {author} {\bibfnamefont {V.}~\bibnamefont
  {Scarani}},\ and\ \bibinfo {author} {\bibfnamefont {D.}~\bibnamefont
  {Matsukevich}},\ }\bibfield  {title} {\bibinfo {title} {Quantum absorption
  refrigerator with trapped ions},\ }\bibfield  {journal} {\bibinfo  {journal}
  {Nature Communications}\ }\textbf {\bibinfo {volume} {10}},\ \href
  {https://doi.org/10.1038/s41467-018-08090-0} {10.1038/s41467-018-08090-0}
  (\bibinfo {year} {2019})\BibitemShut {NoStop}%
\bibitem [{\citenamefont {Klatzow}\ \emph {et~al.}(2019)\citenamefont
  {Klatzow}, \citenamefont {Becker}, \citenamefont {Ledingham}, \citenamefont
  {Weinzetl}, \citenamefont {Kaczmarek}, \citenamefont {Saunders},
  \citenamefont {Nunn}, \citenamefont {Walmsley}, \citenamefont {Uzdin},\ and\
  \citenamefont {Poem}}]{Klatzow_2019}%
  \BibitemOpen
  \bibfield  {author} {\bibinfo {author} {\bibfnamefont {J.}~\bibnamefont
  {Klatzow}}, \bibinfo {author} {\bibfnamefont {J.~N.}\ \bibnamefont {Becker}},
  \bibinfo {author} {\bibfnamefont {P.~M.}\ \bibnamefont {Ledingham}}, \bibinfo
  {author} {\bibfnamefont {C.}~\bibnamefont {Weinzetl}}, \bibinfo {author}
  {\bibfnamefont {K.~T.}\ \bibnamefont {Kaczmarek}}, \bibinfo {author}
  {\bibfnamefont {D.~J.}\ \bibnamefont {Saunders}}, \bibinfo {author}
  {\bibfnamefont {J.}~\bibnamefont {Nunn}}, \bibinfo {author} {\bibfnamefont
  {I.~A.}\ \bibnamefont {Walmsley}}, \bibinfo {author} {\bibfnamefont
  {R.}~\bibnamefont {Uzdin}},\ and\ \bibinfo {author} {\bibfnamefont
  {E.}~\bibnamefont {Poem}},\ }\bibfield  {title} {\bibinfo {title}
  {Experimental demonstration of quantum effects in the operation of
  microscopic heat engines},\ }\href
  {https://doi.org/10.1103/PhysRevLett.122.110601} {\bibfield  {journal}
  {\bibinfo  {journal} {Phys. Rev. Lett.}\ }\textbf {\bibinfo {volume} {122}},\
  \bibinfo {pages} {110601} (\bibinfo {year} {2019})}\BibitemShut {NoStop}%
\bibitem [{\citenamefont {Peterson}\ \emph {et~al.}(2019)\citenamefont
  {Peterson}, \citenamefont {Batalh\~ao}, \citenamefont {Herrera},
  \citenamefont {Souza}, \citenamefont {Sarthour}, \citenamefont {Oliveira},\
  and\ \citenamefont {Serra}}]{Peterson_2019}%
  \BibitemOpen
  \bibfield  {author} {\bibinfo {author} {\bibfnamefont {J.~P.~S.}\
  \bibnamefont {Peterson}}, \bibinfo {author} {\bibfnamefont {T.~B.}\
  \bibnamefont {Batalh\~ao}}, \bibinfo {author} {\bibfnamefont
  {M.}~\bibnamefont {Herrera}}, \bibinfo {author} {\bibfnamefont {A.~M.}\
  \bibnamefont {Souza}}, \bibinfo {author} {\bibfnamefont {R.~S.}\ \bibnamefont
  {Sarthour}}, \bibinfo {author} {\bibfnamefont {I.~S.}\ \bibnamefont
  {Oliveira}},\ and\ \bibinfo {author} {\bibfnamefont {R.~M.}\ \bibnamefont
  {Serra}},\ }\bibfield  {title} {\bibinfo {title} {Experimental
  characterization of a spin quantum heat engine},\ }\href
  {https://doi.org/10.1103/PhysRevLett.123.240601} {\bibfield  {journal}
  {\bibinfo  {journal} {Phys. Rev. Lett.}\ }\textbf {\bibinfo {volume} {123}},\
  \bibinfo {pages} {240601} (\bibinfo {year} {2019})}\BibitemShut {NoStop}%
\bibitem [{\citenamefont {Bouton}\ \emph {et~al.}(2021)\citenamefont {Bouton},
  \citenamefont {Nettersheim}, \citenamefont {Burgardt}, \citenamefont {Adam},
  \citenamefont {Lutz},\ and\ \citenamefont {Widera}}]{Bouton_2021}%
  \BibitemOpen
  \bibfield  {author} {\bibinfo {author} {\bibfnamefont {Q.}~\bibnamefont
  {Bouton}}, \bibinfo {author} {\bibfnamefont {J.}~\bibnamefont {Nettersheim}},
  \bibinfo {author} {\bibfnamefont {S.}~\bibnamefont {Burgardt}}, \bibinfo
  {author} {\bibfnamefont {D.}~\bibnamefont {Adam}}, \bibinfo {author}
  {\bibfnamefont {E.}~\bibnamefont {Lutz}},\ and\ \bibinfo {author}
  {\bibfnamefont {A.}~\bibnamefont {Widera}},\ }\bibfield  {title} {\bibinfo
  {title} {A quantum heat engine driven by atomic collisions},\ }\bibfield
  {journal} {\bibinfo  {journal} {Nature Communications}\ }\textbf {\bibinfo
  {volume} {12}},\ \href {https://doi.org/10.1038/s41467-021-22222-z}
  {10.1038/s41467-021-22222-z} (\bibinfo {year} {2021})\BibitemShut {NoStop}%
\bibitem [{\citenamefont {Zanin}\ \emph {et~al.}(2022)\citenamefont {Zanin},
  \citenamefont {Antesberger}, \citenamefont {Jacquet}, \citenamefont
  {Ribeiro}, \citenamefont {Rozema},\ and\ \citenamefont
  {Walther}}]{Zanin_2022}%
  \BibitemOpen
  \bibfield  {author} {\bibinfo {author} {\bibfnamefont {G.~L.}\ \bibnamefont
  {Zanin}}, \bibinfo {author} {\bibfnamefont {M.}~\bibnamefont {Antesberger}},
  \bibinfo {author} {\bibfnamefont {M.~J.}\ \bibnamefont {Jacquet}}, \bibinfo
  {author} {\bibfnamefont {P.~H.~S.}\ \bibnamefont {Ribeiro}}, \bibinfo
  {author} {\bibfnamefont {L.~A.}\ \bibnamefont {Rozema}},\ and\ \bibinfo
  {author} {\bibfnamefont {P.}~\bibnamefont {Walther}},\ }\bibfield  {title}
  {\bibinfo {title} {Enhanced photonic maxwell's demon with correlated baths},\
  }\href {https://doi.org/10.22331/q-2022-09-20-810} {\bibfield  {journal}
  {\bibinfo  {journal} {Quantum}\ }\textbf {\bibinfo {volume} {6}},\ \bibinfo
  {pages} {810} (\bibinfo {year} {2022})}\BibitemShut {NoStop}%
\bibitem [{\citenamefont {Koski}\ \emph {et~al.}(2014)\citenamefont {Koski},
  \citenamefont {Maisi}, \citenamefont {Pekola},\ and\ \citenamefont
  {Averin}}]{Koski_2014}%
  \BibitemOpen
  \bibfield  {author} {\bibinfo {author} {\bibfnamefont {J.~V.}\ \bibnamefont
  {Koski}}, \bibinfo {author} {\bibfnamefont {V.~F.}\ \bibnamefont {Maisi}},
  \bibinfo {author} {\bibfnamefont {J.~P.}\ \bibnamefont {Pekola}},\ and\
  \bibinfo {author} {\bibfnamefont {D.~V.}\ \bibnamefont {Averin}},\ }\bibfield
   {title} {\bibinfo {title} {Experimental realization of a szilard engine with
  a single electron},\ }\href {https://doi.org/10.1073/pnas.1406966111}
  {\bibfield  {journal} {\bibinfo  {journal} {Proceedings of the National
  Academy of Sciences}\ }\textbf {\bibinfo {volume} {111}},\ \bibinfo {pages}
  {13786} (\bibinfo {year} {2014})},\ \Eprint
  {https://arxiv.org/abs/https://www.pnas.org/doi/pdf/10.1073/pnas.1406966111}
  {https://www.pnas.org/doi/pdf/10.1073/pnas.1406966111} \BibitemShut {NoStop}%
\bibitem [{\citenamefont {Koski}\ \emph {et~al.}(2015)\citenamefont {Koski},
  \citenamefont {Kutvonen}, \citenamefont {Khaymovich}, \citenamefont
  {Ala-Nissila},\ and\ \citenamefont {Pekola}}]{Koski_2015}%
  \BibitemOpen
  \bibfield  {author} {\bibinfo {author} {\bibfnamefont {J.~V.}\ \bibnamefont
  {Koski}}, \bibinfo {author} {\bibfnamefont {A.}~\bibnamefont {Kutvonen}},
  \bibinfo {author} {\bibfnamefont {I.~M.}\ \bibnamefont {Khaymovich}},
  \bibinfo {author} {\bibfnamefont {T.}~\bibnamefont {Ala-Nissila}},\ and\
  \bibinfo {author} {\bibfnamefont {J.~P.}\ \bibnamefont {Pekola}},\ }\bibfield
   {title} {\bibinfo {title} {On-chip maxwell's demon as an information-powered
  refrigerator},\ }\href {https://doi.org/10.1103/PhysRevLett.115.260602}
  {\bibfield  {journal} {\bibinfo  {journal} {Phys. Rev. Lett.}\ }\textbf
  {\bibinfo {volume} {115}},\ \bibinfo {pages} {260602} (\bibinfo {year}
  {2015})}\BibitemShut {NoStop}%
\bibitem [{\citenamefont {Allahverdyan}\ \emph {et~al.}(2010)\citenamefont
  {Allahverdyan}, \citenamefont {Hovhannisyan},\ and\ \citenamefont
  {Mahler}}]{two_stroke_Allahverdyan}%
  \BibitemOpen
  \bibfield  {author} {\bibinfo {author} {\bibfnamefont {A.~E.}\ \bibnamefont
  {Allahverdyan}}, \bibinfo {author} {\bibfnamefont {K.}~\bibnamefont
  {Hovhannisyan}},\ and\ \bibinfo {author} {\bibfnamefont {G.}~\bibnamefont
  {Mahler}},\ }\bibfield  {title} {\bibinfo {title} {Optimal refrigerator},\
  }\href {https://doi.org/10.1103/PhysRevE.81.051129} {\bibfield  {journal}
  {\bibinfo  {journal} {Phys. Rev. E}\ }\textbf {\bibinfo {volume} {81}},\
  \bibinfo {pages} {051129} (\bibinfo {year} {2010})}\BibitemShut {NoStop}%
\bibitem [{\citenamefont {Silva}\ \emph {et~al.}(2016)\citenamefont {Silva},
  \citenamefont {Manzano}, \citenamefont {Skrzypczyk},\ and\ \citenamefont
  {Brunner}}]{Silvadimension}%
  \BibitemOpen
  \bibfield  {author} {\bibinfo {author} {\bibfnamefont {R.}~\bibnamefont
  {Silva}}, \bibinfo {author} {\bibfnamefont {G.}~\bibnamefont {Manzano}},
  \bibinfo {author} {\bibfnamefont {P.}~\bibnamefont {Skrzypczyk}},\ and\
  \bibinfo {author} {\bibfnamefont {N.}~\bibnamefont {Brunner}},\ }\bibfield
  {title} {\bibinfo {title} {Performance of autonomous quantum thermal
  machines: Hilbert space dimension as a thermodynamical resource},\ }\href
  {https://doi.org/10.1103/PhysRevE.94.032120} {\bibfield  {journal} {\bibinfo
  {journal} {Phys. Rev. E}\ }\textbf {\bibinfo {volume} {94}},\ \bibinfo
  {pages} {032120} (\bibinfo {year} {2016})}\BibitemShut {NoStop}%
\bibitem [{\citenamefont {Clivaz}\ \emph
  {et~al.}(2019{\natexlab{a}})\citenamefont {Clivaz}, \citenamefont {Silva},
  \citenamefont {Haack}, \citenamefont {Brask}, \citenamefont {Brunner},\ and\
  \citenamefont {Huber}}]{ClivazPRL}%
  \BibitemOpen
  \bibfield  {author} {\bibinfo {author} {\bibfnamefont {F.}~\bibnamefont
  {Clivaz}}, \bibinfo {author} {\bibfnamefont {R.}~\bibnamefont {Silva}},
  \bibinfo {author} {\bibfnamefont {G.}~\bibnamefont {Haack}}, \bibinfo
  {author} {\bibfnamefont {J.~B.}\ \bibnamefont {Brask}}, \bibinfo {author}
  {\bibfnamefont {N.}~\bibnamefont {Brunner}},\ and\ \bibinfo {author}
  {\bibfnamefont {M.}~\bibnamefont {Huber}},\ }\bibfield  {title} {\bibinfo
  {title} {Unifying paradigms of quantum refrigeration: A universal and
  attainable bound on cooling},\ }\href
  {https://doi.org/10.1103/PhysRevLett.123.170605} {\bibfield  {journal}
  {\bibinfo  {journal} {Phys. Rev. Lett.}\ }\textbf {\bibinfo {volume} {123}},\
  \bibinfo {pages} {170605} (\bibinfo {year} {2019}{\natexlab{a}})}\BibitemShut
  {NoStop}%
\bibitem [{\citenamefont {Clivaz}\ \emph
  {et~al.}(2019{\natexlab{b}})\citenamefont {Clivaz}, \citenamefont {Silva},
  \citenamefont {Haack}, \citenamefont {Brask}, \citenamefont {Brunner},\ and\
  \citenamefont {Huber}}]{ClivazPRE}%
  \BibitemOpen
  \bibfield  {author} {\bibinfo {author} {\bibfnamefont {F.}~\bibnamefont
  {Clivaz}}, \bibinfo {author} {\bibfnamefont {R.}~\bibnamefont {Silva}},
  \bibinfo {author} {\bibfnamefont {G.}~\bibnamefont {Haack}}, \bibinfo
  {author} {\bibfnamefont {J.~B.}\ \bibnamefont {Brask}}, \bibinfo {author}
  {\bibfnamefont {N.}~\bibnamefont {Brunner}},\ and\ \bibinfo {author}
  {\bibfnamefont {M.}~\bibnamefont {Huber}},\ }\bibfield  {title} {\bibinfo
  {title} {Unifying paradigms of quantum refrigeration: Fundamental limits of
  cooling and associated work costs},\ }\href
  {https://doi.org/10.1103/PhysRevE.100.042130} {\bibfield  {journal} {\bibinfo
   {journal} {Phys. Rev. E}\ }\textbf {\bibinfo {volume} {100}},\ \bibinfo
  {pages} {042130} (\bibinfo {year} {2019}{\natexlab{b}})}\BibitemShut
  {NoStop}%
\bibitem [{\citenamefont {\L{}obejko}\ \emph {et~al.}(2024)\citenamefont
  {\L{}obejko}, \citenamefont {Biswas}, \citenamefont {Mazurek},\ and\
  \citenamefont {Horodecki}}]{BiswasLobejko}%
  \BibitemOpen
  \bibfield  {author} {\bibinfo {author} {\bibfnamefont {M.}~\bibnamefont
  {\L{}obejko}}, \bibinfo {author} {\bibfnamefont {T.}~\bibnamefont {Biswas}},
  \bibinfo {author} {\bibfnamefont {P.}~\bibnamefont {Mazurek}},\ and\ \bibinfo
  {author} {\bibfnamefont {M.}~\bibnamefont {Horodecki}},\ }\bibfield  {title}
  {\bibinfo {title} {Catalytic advantage in otto-like two-stroke quantum
  engines},\ }\href {https://doi.org/10.1103/PhysRevLett.132.260403} {\bibfield
   {journal} {\bibinfo  {journal} {Phys. Rev. Lett.}\ }\textbf {\bibinfo
  {volume} {132}},\ \bibinfo {pages} {260403} (\bibinfo {year}
  {2024})}\BibitemShut {NoStop}%
\bibitem [{\citenamefont {Melo}\ \emph {et~al.}(2022)\citenamefont {Melo},
  \citenamefont {S\'a}, \citenamefont {Roditi}, \citenamefont {Souza},
  \citenamefont {Oliveira}, \citenamefont {Sarthour},\ and\ \citenamefont
  {Landi}}]{Melo2022}%
  \BibitemOpen
  \bibfield  {author} {\bibinfo {author} {\bibfnamefont {F.~V.}\ \bibnamefont
  {Melo}}, \bibinfo {author} {\bibfnamefont {N.}~\bibnamefont {S\'a}}, \bibinfo
  {author} {\bibfnamefont {I.}~\bibnamefont {Roditi}}, \bibinfo {author}
  {\bibfnamefont {A.~M.}\ \bibnamefont {Souza}}, \bibinfo {author}
  {\bibfnamefont {I.~S.}\ \bibnamefont {Oliveira}}, \bibinfo {author}
  {\bibfnamefont {R.~S.}\ \bibnamefont {Sarthour}},\ and\ \bibinfo {author}
  {\bibfnamefont {G.~T.}\ \bibnamefont {Landi}},\ }\bibfield  {title} {\bibinfo
  {title} {Implementation of a two-stroke quantum heat engine with a
  collisional model},\ }\href {https://doi.org/10.1103/PhysRevA.106.032410}
  {\bibfield  {journal} {\bibinfo  {journal} {Phys. Rev. A}\ }\textbf {\bibinfo
  {volume} {106}},\ \bibinfo {pages} {032410} (\bibinfo {year}
  {2022})}\BibitemShut {NoStop}%
\bibitem [{\citenamefont {Molitor}\ and\ \citenamefont
  {Landi}(2020)}]{Stroboscopic_Molitor}%
  \BibitemOpen
  \bibfield  {author} {\bibinfo {author} {\bibfnamefont {O.~A.~D.}\
  \bibnamefont {Molitor}}\ and\ \bibinfo {author} {\bibfnamefont {G.~T.}\
  \bibnamefont {Landi}},\ }\bibfield  {title} {\bibinfo {title} {Stroboscopic
  two-stroke quantum heat engines},\ }\href
  {https://doi.org/10.1103/PhysRevA.102.042217} {\bibfield  {journal} {\bibinfo
   {journal} {Phys. Rev. A}\ }\textbf {\bibinfo {volume} {102}},\ \bibinfo
  {pages} {042217} (\bibinfo {year} {2020})}\BibitemShut {NoStop}%
\bibitem [{\citenamefont {Bhattacharjee}\ \emph {et~al.}(2020)\citenamefont
  {Bhattacharjee}, \citenamefont {Bhattacharya}, \citenamefont {Niedenzu},
  \citenamefont {Mukherjee},\ and\ \citenamefont {Dutta}}]{Bhattacharjee_2020}%
  \BibitemOpen
  \bibfield  {author} {\bibinfo {author} {\bibfnamefont {S.}~\bibnamefont
  {Bhattacharjee}}, \bibinfo {author} {\bibfnamefont {U.}~\bibnamefont
  {Bhattacharya}}, \bibinfo {author} {\bibfnamefont {W.}~\bibnamefont
  {Niedenzu}}, \bibinfo {author} {\bibfnamefont {V.}~\bibnamefont
  {Mukherjee}},\ and\ \bibinfo {author} {\bibfnamefont {A.}~\bibnamefont
  {Dutta}},\ }\bibfield  {title} {\bibinfo {title} {Quantum magnetometry using
  two-stroke thermal machines},\ }\href
  {https://doi.org/10.1088/1367-2630/ab61d6} {\bibfield  {journal} {\bibinfo
  {journal} {New Journal of Physics}\ }\textbf {\bibinfo {volume} {22}},\
  \bibinfo {pages} {013024} (\bibinfo {year} {2020})}\BibitemShut {NoStop}%
\bibitem [{\citenamefont {Piccione}\ \emph {et~al.}(2021)\citenamefont
  {Piccione}, \citenamefont {De~Chiara},\ and\ \citenamefont
  {Bellomo}}]{Piccione}%
  \BibitemOpen
  \bibfield  {author} {\bibinfo {author} {\bibfnamefont {N.}~\bibnamefont
  {Piccione}}, \bibinfo {author} {\bibfnamefont {G.}~\bibnamefont
  {De~Chiara}},\ and\ \bibinfo {author} {\bibfnamefont {B.}~\bibnamefont
  {Bellomo}},\ }\bibfield  {title} {\bibinfo {title} {Power maximization of
  two-stroke quantum thermal machines},\ }\href
  {https://doi.org/10.1103/PhysRevA.103.032211} {\bibfield  {journal} {\bibinfo
   {journal} {Phys. Rev. A}\ }\textbf {\bibinfo {volume} {103}},\ \bibinfo
  {pages} {032211} (\bibinfo {year} {2021})}\BibitemShut {NoStop}%
\bibitem [{\citenamefont {Rodr\'{\i}guez}\ \emph {et~al.}(2023)\citenamefont
  {Rodr\'{\i}guez}, \citenamefont {Ahmadi}, \citenamefont {Mazurek},
  \citenamefont {Barzanjeh}, \citenamefont {Alicki},\ and\ \citenamefont
  {Horodecki}}]{RRR_Ahmadi_2023}%
  \BibitemOpen
  \bibfield  {author} {\bibinfo {author} {\bibfnamefont {R.~R.}\ \bibnamefont
  {Rodr\'{\i}guez}}, \bibinfo {author} {\bibfnamefont {B.}~\bibnamefont
  {Ahmadi}}, \bibinfo {author} {\bibfnamefont {P.}~\bibnamefont {Mazurek}},
  \bibinfo {author} {\bibfnamefont {S.}~\bibnamefont {Barzanjeh}}, \bibinfo
  {author} {\bibfnamefont {R.}~\bibnamefont {Alicki}},\ and\ \bibinfo {author}
  {\bibfnamefont {P.}~\bibnamefont {Horodecki}},\ }\bibfield  {title} {\bibinfo
  {title} {Catalysis in charging quantum batteries},\ }\href
  {https://doi.org/10.1103/PhysRevA.107.042419} {\bibfield  {journal} {\bibinfo
   {journal} {Phys. Rev. A}\ }\textbf {\bibinfo {volume} {107}},\ \bibinfo
  {pages} {042419} (\bibinfo {year} {2023})}\BibitemShut {NoStop}%
\bibitem [{\citenamefont {Buffoni}\ \emph {et~al.}(2019)\citenamefont
  {Buffoni}, \citenamefont {Solfanelli}, \citenamefont {Verrucchi},
  \citenamefont {Cuccoli},\ and\ \citenamefont {Campisi}}]{Campisi1}%
  \BibitemOpen
  \bibfield  {author} {\bibinfo {author} {\bibfnamefont {L.}~\bibnamefont
  {Buffoni}}, \bibinfo {author} {\bibfnamefont {A.}~\bibnamefont {Solfanelli}},
  \bibinfo {author} {\bibfnamefont {P.}~\bibnamefont {Verrucchi}}, \bibinfo
  {author} {\bibfnamefont {A.}~\bibnamefont {Cuccoli}},\ and\ \bibinfo {author}
  {\bibfnamefont {M.}~\bibnamefont {Campisi}},\ }\bibfield  {title} {\bibinfo
  {title} {Quantum measurement cooling},\ }\href
  {https://doi.org/10.1103/PhysRevLett.122.070603} {\bibfield  {journal}
  {\bibinfo  {journal} {Phys. Rev. Lett.}\ }\textbf {\bibinfo {volume} {122}},\
  \bibinfo {pages} {070603} (\bibinfo {year} {2019})}\BibitemShut {NoStop}%
\bibitem [{\citenamefont {Solfanelli}\ \emph {et~al.}(2020)\citenamefont
  {Solfanelli}, \citenamefont {Falsetti},\ and\ \citenamefont
  {Campisi}}]{Campisi2}%
  \BibitemOpen
  \bibfield  {author} {\bibinfo {author} {\bibfnamefont {A.}~\bibnamefont
  {Solfanelli}}, \bibinfo {author} {\bibfnamefont {M.}~\bibnamefont
  {Falsetti}},\ and\ \bibinfo {author} {\bibfnamefont {M.}~\bibnamefont
  {Campisi}},\ }\bibfield  {title} {\bibinfo {title} {Nonadiabatic single-qubit
  quantum otto engine},\ }\href {https://doi.org/10.1103/PhysRevB.101.054513}
  {\bibfield  {journal} {\bibinfo  {journal} {Phys. Rev. B}\ }\textbf {\bibinfo
  {volume} {101}},\ \bibinfo {pages} {054513} (\bibinfo {year}
  {2020})}\BibitemShut {NoStop}%
\bibitem [{\citenamefont {Campisi}\ and\ \citenamefont
  {Fazio}(2016)}]{Campisi3}%
  \BibitemOpen
  \bibfield  {author} {\bibinfo {author} {\bibfnamefont {M.}~\bibnamefont
  {Campisi}}\ and\ \bibinfo {author} {\bibfnamefont {R.}~\bibnamefont
  {Fazio}},\ }\bibfield  {title} {\bibinfo {title} {Dissipation, correlation
  and lags in heat engines},\ }\href
  {https://doi.org/10.1088/1751-8113/49/34/345002} {\bibfield  {journal}
  {\bibinfo  {journal} {Journal of Physics A: Mathematical and Theoretical}\
  }\textbf {\bibinfo {volume} {49}},\ \bibinfo {pages} {345002} (\bibinfo
  {year} {2016})}\BibitemShut {NoStop}%
\bibitem [{\citenamefont {Campisi}\ \emph {et~al.}(2015)\citenamefont
  {Campisi}, \citenamefont {Pekola},\ and\ \citenamefont {Fazio}}]{Campisi4}%
  \BibitemOpen
  \bibfield  {author} {\bibinfo {author} {\bibfnamefont {M.}~\bibnamefont
  {Campisi}}, \bibinfo {author} {\bibfnamefont {J.}~\bibnamefont {Pekola}},\
  and\ \bibinfo {author} {\bibfnamefont {R.}~\bibnamefont {Fazio}},\ }\bibfield
   {title} {\bibinfo {title} {Nonequilibrium fluctuations in quantum heat
  engines: theory, example, and possible solid state experiments},\ }\href
  {https://doi.org/10.1088/1367-2630/17/3/035012} {\bibfield  {journal}
  {\bibinfo  {journal} {New Journal of Physics}\ }\textbf {\bibinfo {volume}
  {17}},\ \bibinfo {pages} {035012} (\bibinfo {year} {2015})}\BibitemShut
  {NoStop}%
\bibitem [{\citenamefont {{\L{}}obejko}\ \emph {et~al.}(2020)\citenamefont
  {{\L{}}obejko}, \citenamefont {Mazurek},\ and\ \citenamefont
  {Horodecki}}]{Lobejko2020}%
  \BibitemOpen
  \bibfield  {author} {\bibinfo {author} {\bibfnamefont {M.}~\bibnamefont
  {{\L{}}obejko}}, \bibinfo {author} {\bibfnamefont {P.}~\bibnamefont
  {Mazurek}},\ and\ \bibinfo {author} {\bibfnamefont {M.}~\bibnamefont
  {Horodecki}},\ }\bibfield  {title} {\bibinfo {title} {Thermodynamics of
  {M}inimal {C}oupling {Q}uantum {H}eat {E}ngines},\ }\href
  {https://doi.org/10.22331/q-2020-12-23-375} {\bibfield  {journal} {\bibinfo
  {journal} {{Quantum}}\ }\textbf {\bibinfo {volume} {4}},\ \bibinfo {pages}
  {375} (\bibinfo {year} {2020})}\BibitemShut {NoStop}%
\bibitem [{\citenamefont {Biswas}\ \emph {et~al.}(2022)\citenamefont {Biswas},
  \citenamefont {{\L{}}obejko}, \citenamefont {Mazurek}, \citenamefont
  {Ja{\l{}}owiecki},\ and\ \citenamefont {Horodecki}}]{BiswasQuantum}%
  \BibitemOpen
  \bibfield  {author} {\bibinfo {author} {\bibfnamefont {T.}~\bibnamefont
  {Biswas}}, \bibinfo {author} {\bibfnamefont {M.}~\bibnamefont
  {{\L{}}obejko}}, \bibinfo {author} {\bibfnamefont {P.}~\bibnamefont
  {Mazurek}}, \bibinfo {author} {\bibfnamefont {K.}~\bibnamefont
  {Ja{\l{}}owiecki}},\ and\ \bibinfo {author} {\bibfnamefont {M.}~\bibnamefont
  {Horodecki}},\ }\bibfield  {title} {\bibinfo {title} {Extraction of
  ergotropy: free energy bound and application to open cycle engines},\ }\href
  {https://doi.org/10.22331/q-2022-10-17-841} {\bibfield  {journal} {\bibinfo
  {journal} {{Quantum}}\ }\textbf {\bibinfo {volume} {6}},\ \bibinfo {pages}
  {841} (\bibinfo {year} {2022})}\BibitemShut {NoStop}%
\bibitem [{\citenamefont {Niedenzu}\ \emph {et~al.}(2019)\citenamefont
  {Niedenzu}, \citenamefont {Huber},\ and\ \citenamefont
  {Boukobza}}]{Niedenzu2019}%
  \BibitemOpen
  \bibfield  {author} {\bibinfo {author} {\bibfnamefont {W.}~\bibnamefont
  {Niedenzu}}, \bibinfo {author} {\bibfnamefont {M.}~\bibnamefont {Huber}},\
  and\ \bibinfo {author} {\bibfnamefont {E.}~\bibnamefont {Boukobza}},\
  }\bibfield  {title} {\bibinfo {title} {Concepts of work in autonomous quantum
  heat engines},\ }\href {https://doi.org/10.22331/q-2019-10-14-195} {\bibfield
   {journal} {\bibinfo  {journal} {{Quantum}}\ }\textbf {\bibinfo {volume}
  {3}},\ \bibinfo {pages} {195} (\bibinfo {year} {2019})}\BibitemShut {NoStop}%
\bibitem [{\citenamefont {Ptaszy\ifmmode~\acute{n}\else
  \'{n}\fi{}ski}(2022)}]{Non-markovian_Pstaz}%
  \BibitemOpen
  \bibfield  {author} {\bibinfo {author} {\bibfnamefont {K.}~\bibnamefont
  {Ptaszy\ifmmode~\acute{n}\else \'{n}\fi{}ski}},\ }\bibfield  {title}
  {\bibinfo {title} {Non-markovian thermal operations boosting the performance
  of quantum heat engines},\ }\href
  {https://doi.org/10.1103/PhysRevE.106.014114} {\bibfield  {journal} {\bibinfo
   {journal} {Phys. Rev. E}\ }\textbf {\bibinfo {volume} {106}},\ \bibinfo
  {pages} {014114} (\bibinfo {year} {2022})}\BibitemShut {NoStop}%
\bibitem [{\citenamefont {Biswas}\ and\ \citenamefont
  {Datta}(2024)}]{BiswasDatta3stroke}%
  \BibitemOpen
  \bibfield  {author} {\bibinfo {author} {\bibfnamefont {T.}~\bibnamefont
  {Biswas}}\ and\ \bibinfo {author} {\bibfnamefont {C.}~\bibnamefont {Datta}},\
  }\bibfield  {title} {\bibinfo {title} {Optimal performance of a three stroke
  heat engine in the microscopic regime},\ }\href
  {https://doi.org/10.48550/arXiv.2404.13461} {\bibfield  {journal} {\bibinfo
  {journal} {arXiv:2404.13461}\ } (\bibinfo {year} {2024})}\BibitemShut
  {NoStop}%
\bibitem [{\citenamefont {Francica}\ \emph {et~al.}(2020)\citenamefont
  {Francica}, \citenamefont {Binder}, \citenamefont {Guarnieri}, \citenamefont
  {Mitchison}, \citenamefont {Goold},\ and\ \citenamefont
  {Plastina}}]{Francica2020}%
  \BibitemOpen
  \bibfield  {author} {\bibinfo {author} {\bibfnamefont {G.}~\bibnamefont
  {Francica}}, \bibinfo {author} {\bibfnamefont {F.~C.}\ \bibnamefont
  {Binder}}, \bibinfo {author} {\bibfnamefont {G.}~\bibnamefont {Guarnieri}},
  \bibinfo {author} {\bibfnamefont {M.~T.}\ \bibnamefont {Mitchison}}, \bibinfo
  {author} {\bibfnamefont {J.}~\bibnamefont {Goold}},\ and\ \bibinfo {author}
  {\bibfnamefont {F.}~\bibnamefont {Plastina}},\ }\bibfield  {title} {\bibinfo
  {title} {Quantum coherence and ergotropy},\ }\href
  {https://doi.org/10.1103/PhysRevLett.125.180603} {\bibfield  {journal}
  {\bibinfo  {journal} {Phys. Rev. Lett.}\ }\textbf {\bibinfo {volume} {125}},\
  \bibinfo {pages} {180603} (\bibinfo {year} {2020})}\BibitemShut {NoStop}%
\bibitem [{\citenamefont {Cangemi}\ \emph {et~al.}(2024)\citenamefont
  {Cangemi}, \citenamefont {Bhadra},\ and\ \citenamefont
  {Levy}}]{cangemni_Levy_engines}%
  \BibitemOpen
  \bibfield  {author} {\bibinfo {author} {\bibfnamefont {L.~M.}\ \bibnamefont
  {Cangemi}}, \bibinfo {author} {\bibfnamefont {C.}~\bibnamefont {Bhadra}},\
  and\ \bibinfo {author} {\bibfnamefont {A.}~\bibnamefont {Levy}},\ }\bibfield
  {title} {\bibinfo {title} {Quantum engines and refrigerators},\ }\href
  {https://doi.org/https://doi.org/10.1016/j.physrep.2024.07.001} {\bibfield
  {journal} {\bibinfo  {journal} {Physics Reports}\ }\textbf {\bibinfo {volume}
  {1087}},\ \bibinfo {pages} {1} (\bibinfo {year} {2024})}\BibitemShut
  {NoStop}%
\bibitem [{\citenamefont {Czartowski}\ \emph {et~al.}(2023)\citenamefont
  {Czartowski}, \citenamefont {de~Oliveira~Junior},\ and\ \citenamefont
  {Korzekwa}}]{Kuba_and_Alex}%
  \BibitemOpen
  \bibfield  {author} {\bibinfo {author} {\bibfnamefont {J.}~\bibnamefont
  {Czartowski}}, \bibinfo {author} {\bibfnamefont {A.}~\bibnamefont
  {de~Oliveira~Junior}},\ and\ \bibinfo {author} {\bibfnamefont
  {K.}~\bibnamefont {Korzekwa}},\ }\bibfield  {title} {\bibinfo {title}
  {Thermal recall: Memory-assisted markovian thermal processes},\ }\href
  {https://doi.org/10.1103/PRXQuantum.4.040304} {\bibfield  {journal} {\bibinfo
   {journal} {PRX Quantum}\ }\textbf {\bibinfo {volume} {4}},\ \bibinfo {pages}
  {040304} (\bibinfo {year} {2023})}\BibitemShut {NoStop}%
\bibitem [{\citenamefont {Yu}\ \emph {et~al.}(2019)\citenamefont {Yu},
  \citenamefont {Guo},\ and\ \citenamefont {Liu}}]{Yu19}%
  \BibitemOpen
  \bibfield  {author} {\bibinfo {author} {\bibfnamefont {C.~S.}\ \bibnamefont
  {Yu}}, \bibinfo {author} {\bibfnamefont {B.~Q.}\ \bibnamefont {Guo}},\ and\
  \bibinfo {author} {\bibfnamefont {T.}~\bibnamefont {Liu}},\ }\bibfield
  {title} {\bibinfo {title} {Quantum self-contained refrigerator in terms of
  the cavity quantum electrodynamics in the weak internal-coupling regime},\
  }\href {https://doi.org/10.1364/OE.27.006863} {\bibfield  {journal} {\bibinfo
   {journal} {Opt. Express}\ }\textbf {\bibinfo {volume} {27}},\ \bibinfo
  {pages} {6863} (\bibinfo {year} {2019})}\BibitemShut {NoStop}%
\bibitem [{\citenamefont {Ali~Aamir}\ \emph {et~al.}(2023)\citenamefont
  {Ali~Aamir}, \citenamefont {Jamet~Suria}, \citenamefont {Guzmán},
  \citenamefont {Castillo-Moreno}, \citenamefont {Epstein}, \citenamefont
  {Yunger~Halpern},\ and\ \citenamefont {Gasparinetti}}]{NYHSG}%
  \BibitemOpen
  \bibfield  {author} {\bibinfo {author} {\bibfnamefont {M.}~\bibnamefont
  {Ali~Aamir}}, \bibinfo {author} {\bibfnamefont {P.}~\bibnamefont
  {Jamet~Suria}}, \bibinfo {author} {\bibfnamefont {J.~A.~M.}\ \bibnamefont
  {Guzmán}}, \bibinfo {author} {\bibfnamefont {C.}~\bibnamefont
  {Castillo-Moreno}}, \bibinfo {author} {\bibfnamefont {J.~M.}\ \bibnamefont
  {Epstein}}, \bibinfo {author} {\bibfnamefont {N.}~\bibnamefont
  {Yunger~Halpern}},\ and\ \bibinfo {author} {\bibfnamefont {S.}~\bibnamefont
  {Gasparinetti}},\ }\bibfield  {title} {\bibinfo {title} {Thermally driven
  quantum refrigerator autonomously resets superconducting qubit},\ }\href
  {https://doi.org/10.48550/arXiv.2305.16710} {\bibfield  {journal} {\bibinfo
  {journal} {arXiv:2305.16710}\ } (\bibinfo {year} {2023})}\BibitemShut
  {NoStop}%
\end{thebibliography}%
\end{document}